%% file: main.tex
\documentclass[11pt]{article}
\usepackage[margin=1in]{geometry}
\usepackage[utf8]{inputenc}
\usepackage[english]{babel}
\usepackage{amsmath}
\usepackage[all=normal,paragraphs=tight,bibliography=tight]{savetrees}
\usepackage{amsthm}
\usepackage{amssymb}
\usepackage{microtype}
\usepackage{xspace}
\usepackage{comment}
\usepackage{letltxmacro}
\usepackage{comment}
\usepackage{tikz}
\usepackage{url}
\usepackage[colorlinks=true,citecolor=red,linkcolor=blue]{hyperref}
\usepackage{import}

\usepackage{graphicx}
\usepackage{caption}
\usepackage{subcaption}
\usepackage{multirow}
\usepackage{enumitem}
\usepackage{mathtools}
\usepackage[absolute]{textpos}

\newtheorem{theorem}{Theorem}[section]
\newtheorem{notheorem}[theorem]{``Wishful-thinking Statement''}
\newtheorem{lemma}[theorem]{Lemma}
\newtheorem{claim}[theorem]{Claim}

\newtheorem{proposition}[theorem]{Proposition}
\newtheorem{observation}[theorem]{Observation}
\theoremstyle{definition}
\newtheorem{definition}[theorem]{Definition}

\newcommand{\N}{\mathbb{N}}
\newcommand{\Oh}{\ensuremath{\mathcal{O}}}

\input{NE-macros}

\newcommand{\bnd}{\partial}

\def\cqedsymbol{\ifmmode$\lrcorner$\else{\unskip\nobreak\hfil
\penalty50\hskip1em\null\nobreak\hfil$\lrcorner$
\parfillskip=0pt\finalhyphendemerits=0\endgraf}\fi} 

\newcommand{\cqed}{\renewcommand{\qed}{\cqedsymbol}}

\newcommand{\executeiffilenewer}[3]{%
\ifnum\pdfstrcmp{\pdffilemoddate{#1}}%
{\pdffilemoddate{#2}}>0%
{\immediate\write18{#3}}\fi%
} 

\newcommand{\cO}{\mathcal{O}\xspace}
\newcommand{\pot}{\Phi}

\newcommand{\bagthresh}{\Upsilon}
\newcommand{\bigadh}{\eta}
\newcommand{\bigbag}{\zeta}

\newcommand\Nat{\mathbb{N}}

\title{Fixed-parameter tractability of Graph Isomorphism\\in graphs with an excluded minor}

\author{
  Daniel Lokshtanov\thanks{
    University of California, Santa Barbara, USA, \texttt{daniello@ucsb.edu}.
    Supported by BSF award 2018302 and NSF award CCF-2008838.
  }
  \and
  Marcin Pilipczuk\thanks{
    Institute of Informatics, University of Warsaw, Poland, \texttt{marcin.pilipczuk@mimuw.edu.pl}.
 This work is 
a part of project CUTACOMBS that has received funding from the European Research Council (ERC) 
under the European Union's Horizon 2020 research and innovation programme (grant agreement No.~714704).
  }
  \and
  Micha\l{} Pilipczuk\thanks{
    Institute of Informatics, University of Warsaw, Poland, \texttt{michal.pilipczuk@mimuw.edu.pl}.
 This work is 
a part of projects TOTAL and BOBR that have received funding from the European Research Council (ERC) 
under the European Union's Horizon 2020 research and innovation programme (grant agreements No.~677651 and~948057, respectively).
  }
  \and 
  Saket Saurabh\thanks{
    Institute of Mathematical Sciences, India, \texttt{saket@imsc.res.in}, and
    Department of Informatics, University of Bergen, Norway, \texttt{Saket.Saurabh@ii.uib.no}. Supported by the European Research Council (ERC) under the European Union's Horizon 2020 research and innovation programme (grant agreement No. 819416), and Swarnajayanti Fellowship (No. DST/SJF/MSA01/2017-18).
  }
}

\date{}

\usepackage{centernot}
\newcommand{\imp}[2]{#1^{\langle#2\rangle}}
\newcommand{\conn}{\mu}

\newcommand{\parent}{\mathrm{parent}}

\begin{document}

\begin{titlepage}
\def\thepage{}
\thispagestyle{empty}
\maketitle

\begin{textblock}{20}(0, 12.7)
\includegraphics[width=40px]{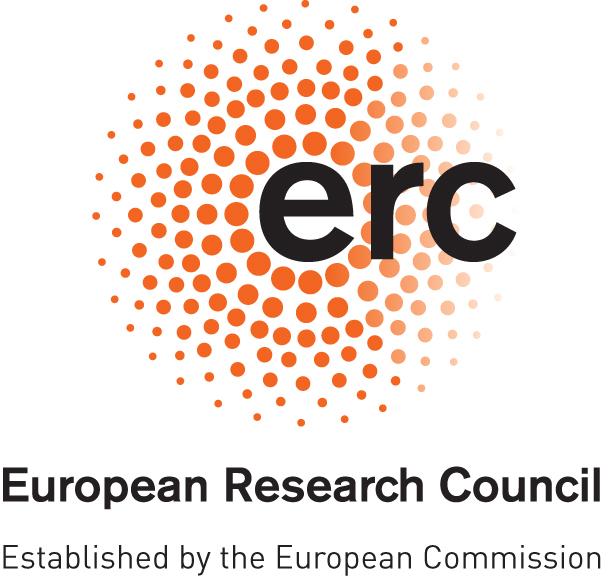}%
\end{textblock}
\begin{textblock}{20}(-0.25, 13.1)
\includegraphics[width=60px]{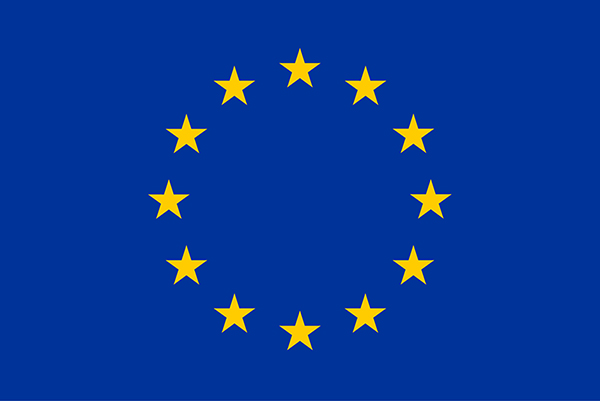}%
\end{textblock}

\begin{abstract}
\input{abstract}

\end{abstract}
\newpage
\tableofcontents
\end{titlepage}

\section{Introduction}
\input{intro}

\section{Overview}\label{sec:overview}

\input{overviewDeco.tex}

\section{Preliminaries}\label{sec:prelims}
\input{prelims.tex}

\section{Statement and Proof of Main Theorems}\label{sec:DS}
\input{statements}

\section{Recursive Understanding for Canonization}\label{sec:recursiveUnderstanding}
This part of the paper contains the proof of Theorem~\ref{thm:intro:toUnbreakable}. In Section~\ref{sec:canonBoundaried} we develop the language and basic tools for dealing with canonization of boundaried graphs while also being mindful of cuts (because of unbreakability) and excluding certain (boundaried) graphs as a topological minor. 

In Section~\ref{sec:unpumping} we define the crucial {\em unpumping} and {\em lifting} operations: If we have canonized some large parts of the graph, and these parts have been carefully selected, then unpumping allows us to essentially ``shrink'' the already canonized parts to constant size. We can then continue working on the reduced graph $G^\star$. The {\em lifting} operation then lifts canonical labelings of $G^\star$ back to canonical labelings of $G$.

In Section~\ref{sec:redToCliqueUnbreakable} we prove Lemma~\ref{lem:cliqueToGeneral}, namely that in order to canonize ${\cal F}$-free graphs $G$ it is sufficient to be able to canonize ${\cal F}$-free graphs that additionally are $(k, q)$-improved-clique unbreakable for sufficiently large $k$ and any $q$ depending only on $k$ and ${\cal F}$.

In Section~\ref{sec:reductionToUnbreakable} we prove Lemma~\ref{lem:unbreakToClique}, namely that in order to canonize $(k, q)$-improved-clique unbreakable graphs, it is sufficient to canonize $q^\star$-unbreakable ${\cal F}$-free graphs where $q^\star : [k] \rightarrow \mathbb{N}$ is upper bounded by a function $\tilde{q} : \mathbb{N} \rightarrow \mathbb{N}$ which depends on ${\cal F}$ but is {\em independent of} $k$.

\subsection{Canonization Tools for Boundaried  Graphs}\label{sec:canonBoundaried}\label{sec:firstDecomp}
\input{canonizationTools.tex}

\subsection{Unpumping and Lifting}\label{sec:unpumping}
\input{unpumping.tex}

\subsection{Reduction To Improved-Clique-Unbreakable Graphs}\label{sec:redToCliqueUnbreakable}

\input{imp-clique-unbreakable.tex}

\subsection{Reduction To Unbreakable Graphs}\label{sec:reductionToUnbreakable}\label{sec:lastDecomp}
In this section, we design an algorithm for {\sc Canonization} in $(q,k,k)$-improved-clique-unbreakable graphs, excluding a fixed family $\cal F$ of graphs as topological minors, assuming we have such an algorithm for graphs that, in addition to excluding $\cal F$, are also suitably unbreakable. We first develop a notion of $k$-important separator extension, which is based on important separators. Then, we define a concept of unbreakability-tied sets. In particular, if a set $A$ is $q$-unbreakability-tied to $D$, then every separation of small order that splits $A$ into two large parts also has to split $D$ into two large parts. This is followed by the notion of stable separation and the main algorithm. 

For some intuition about the notions of $k$-important separator extension and unbreakability-tied sets we give a very high-level description of how the notions are used in the canonization algorithm. 
The canonization algorithm will at some point be faced with a situation where it needs to canonize a graph $G$ with a distiguished set $D$. To do a single divide and conquer step the algorithm needs to compute in an isomorphism-invariant way some set $B$ that contains $D$ and sufficiently disconnects the graph.
If the set $D$ is large and unbreakable our algorithm will use $k$-important separator extension to compute $B$. In this case we will need to prove that the set $B$ is also sufficiently unbreakable in $G$, and for this we will use the concept of ``unbreakability-tied''. %

\input{unbreakabilityTools}

\input{reductionToUnbreakable}

\section{Canonizing unbreakable graphs with an excluded minor}\label{sec:unbrk-minor}
\input{NE-interface}

\bibliographystyle{alpha}
\bibliography{genusIsomorphism,tw-iso,references}

\end{document}

%% file: NE-macros.tex
\renewcommand{\leq}{\leqslant}
\renewcommand{\geq}{\geqslant}
\renewcommand{\le}{\leqslant}

\newcommand{\famgenus}{\mathcal{F}_{\mathrm{genus}}}

\newcommand{\Vtw}{V_{\mathrm{tw}}}
\newcommand{\Vgenus}{V_{\mathrm{genus}}}

\newcommand{\funq}{\widehat{q}}
\newcommand{\funrtime}{\widehat{f}_{\mathrm{time}}}

\newcommand{\ctime}{c_{\mathrm{time}}}
\newcommand{\cgenus}{c_{\mathrm{genus}}}

\newcommand{\funtw}{f_{\mathrm{tw}}}
\newcommand{\fungenus}{f_{\mathrm{genus}}}
\newcommand{\funchain}{f_{\mathrm{chain}}}
\newcommand{\funtime}{f_{\mathrm{time}}}
\newcommand{\funcut}{f_{\mathrm{cut}}}
\newcommand{\funstep}{f_{\mathrm{step}}}

\newcommand{\tw}{\mathrm{tw}}

%% file: abstract.tex
We prove that \textsc{Graph Isomorphism} and \textsc{Canonization} in graphs excluding a fixed graph $H$ as a minor can be solved by an algorithm working in time $f(H)\cdot n^{\Oh(1)}$, where $f$ is some function. In other words, we show that these problems are fixed-parameter tractable when parameterized by the size of the excluded minor, with the caveat that the bound on the running time is not necessarily computable. The underlying approach is based on decomposing the graph in a canonical way into {\em{unbreakable}} (intuitively, well-connected) parts, which essentially provides a reduction to the case where the given $H$-minor-free graph is unbreakable itself. This is complemented by an analysis of unbreakable $H$-minor-free graphs, performed in a second subordinate manuscript, which reveals that every such graph can be canonically decomposed into a part that admits few automorphisms and a part that has bounded treewidth.

%% file: intro.tex
The \textsc{Graph Isomorphism} problem is arguably the most widely known problem whose membership in $\mathsf{P}$ is unknown, but which is not believed to be $\mathsf{NP}$-hard. 
After decades of research, a quasi-polynomial time algorithm was proposed by Babai in 2015~\cite{Babai16}.

While the existence of a polynomial-time algorithm on general graphs is still elusive, the complexity of \textsc{Graph Isomorphism} has been well understood on several classes of graphs, where structural properties of graphs in question have been used to design polynomial-time procedures solving the problem.
Classic results in this area include an $n^{\Oh(d)}$-time algorithm
on graphs of maximum degree~$d$~\cite{BabaiL83,Luks82}, 
a polynomial-time algorithm for planar graphs~\cite{HopcroftT72,HopcroftT73,HopcroftW74,Weinberg},
an $n^{\Oh(g)}$-time algorithm on graphs of Euler genus~$g$~\cite{FilottiM80,Miller80},
an $\Oh(n^{k+4.5})$-time algorithm for graphs of treewidth~$k$~\cite{Bodlaender90},
an $n^{\Oh(k)}$-time algorithm for graphs of rankwidth~$k$~\cite{rankwidth1,rankwidth2},
an $n^{f(|H|)}$-time algorithm for graphs excluding a fixed graph $H$ as a minor~\cite{ponomarenko},
and an $n^{f(|H|)}$-time algorithm for graphs excluding a fixed graph $H$ as a topological
minor~\cite{marx-grohe} (where $f$ is some computable function). 

In all the results mentioned above, the degree of the polynomial bound on the running time depends
on the parameter --- maximum degree, genus, treewidth\footnote{In order to avoid an avalanche of definitions in the introduction, we give formal definitions for the concepts used in the paper in Section~\ref{sec:prelims}}, rankwidth, 
   or the size of the excluded (topological) minor --- 
in at least a linear fashion. Since the parameter can be as high as linear in the size of the graph,
for large values of the parameter the running time bound of the quasi-polynomial-time algorithm
of Babai~\cite{Babai16}, which works on general graph, is preferable.
During the last few years, there has been several successful attempts of bridging this gap by using the group-theoretic approach of Babai in conjunction with structural insight about considered graph classes. This led to algorithms with running time of the form $n^{\mathrm{polylog}(p)}$, where $p$ is any of the following parameters: maximum degree~\cite{GroheNS18}, Euler genus~\cite{Neuen20}, treewidth~\cite{Wiebking20}, and the size of a fixed graph $H$ excluded as a minor~\cite{GroheNW20}.
We refer to a recent survey~\cite{GN-survey} for an excellent exposition.

A parallel line of research is to turn the aforementioned algorithms into fixed-parameter tractable
algorithms for the parameters in question. That is, instead of a running time bound of the form $n^{f(p)}$ for a computable function $f$ and a parameter $p$, we would like to have an algorithm with running time bound
$f(p) \cdot n^c$ for a universal constant $c$. In other words, the degree of the polynomial governing the running time
bound should be independent of the parameter; only the leading multiplicative factor may depend on it.

In this line of research, the current authors developed an FPT algorithm for \textsc{Graph Isomorphism} parameterized by treewidth~\cite{LokshtanovPPS17}. This result has been subsequently improved and simplified~\cite{GroheNSW20}, as well as used to give a slice-wise logspace algorithm~\cite{ElberfeldS17}.
In 2015, Kawarabayashi~\cite{Kawarabayashi15} announced an FPT algorithm for \textsc{Graph Isomorphism} parameterized by the Euler genus of the input graph, with a linear dependency of the running time on the input size. Very recently, Neuen~\cite{Neuen21} proposed a different and simpler algorithm for this case, which runs in time $2^{\Oh(g^4\log g)}\cdot n^{c}$, for some constant $c$.
The recent survey~\cite{GN-survey} mentions obtaining FPT algorithms with parameterizations by the size of an excluded minor, maximum degree, and the size of an excluded topological minor as important open problems. (Note that the last parameter, the size of an excluded topological minor, generalizes the other two.)

\paragraph{Our contribution.} In this work we essentially solve the first of these open problems by proving the following statement.

\begin{theorem}\label{thm:main}
There exists an algorithm that given a graph $H$ and an $H$-minor-free graph $G$, works in time $f(H)\cdot n^{\Oh(1)}$ for some function $f$ and outputs a canonical labeling of $G$.
\end{theorem}

Here, a {\em{canonical labeling}} of $G$ is a labeling of the vertices of $G$ with labels from $[|V(G)|]=\{1,\ldots,|V(G)|\}$ so that for any two isomorphic graphs $G$ and $G'$, the mapping matching vertices with equal labels in $G$ and $G'$ is an isomorphism between $G$ and $G'$. Thus, our algorithm solves the more general \textsc{Canonization} problems, while \textsc{Graph Isomorphism} can be solved by computing the canonical labelings for both input graphs and comparing the obtained labeled graphs.

The caveat in Theorem~\ref{thm:main} is that we are not able to guarantee that $f$ is computable. This makes the algorithm formally fall outside of the usual definitions of fixed-parameter tractability (see e.g.~\cite{platypus}), but we still allow ourselves to call it an FPT algorithm for \textsc{Graph Isomorphism} and {\sc{Canonization}} parameterized by the size of an excluded minor. We stress that we obtain a single algorithm that takes $H$ on input, and not a different algorithm for every fixed $H$.

We now describe the ideas standing behind the proof of Theorem~\ref{thm:main}. First, we need to take a closer look at the FPT algorithm for \textsc{Graph Isomorphism} and \textsc{Canonization} for graphs of bounded treewidth~\cite{LokshtanovPPS17,GroheNSW20}. There, the goal is to obtain an \emph{isomorphism-invariant} tree decomposition of the input graph of width bounded in parameter; then, to test isomorphism of two graphs it suffices to test isomorphism of such decompositions.
Unfortunately, it is actually impossible to find one such isomorphism-invariant tree decomposition of the input graph; for instance, in a long cycle one needs to arbitrarily break symmetry at some moment. However, up to technical details, it suffices and is possible to find a small isomorphism-invariant family of tree decompositions.
To achieve this goal, the algorithm of~\cite{LokshtanovPPS17} heavily relies on the existing understanding of parameterized algorithms for approximating the treewidth of a graph.

The initial idea behind our approach is to borrow more tools from the literature on parameterized graph separation problems, in particular from the work on the technique of {\em{recursive understanding}}~\cite{KawarabayashiT11,ChitnisCHPP16,CyganLPPS19,lean-decomp}. This work culminated in the following decomposition theorem for graphs, proved by a superset of authors.

\begin{theorem}[\cite{lean-decomp}]\label{thm:lean}
Given a graph $G$ and an integer $k$, one can in time $2^{\Oh(k \log k)} n^{\Oh(1)}$ compute
a tree decomposition of $G$ where every adhesion is of size at most $k$
and for every bag $S$ the following holds: for every separation $(A,B)$ in 
$G$ of order at most $k$, either $|A \cap S| \leq |A \cap B|$ or $|B \cap S| \leq |A \cap B|$.
\end{theorem}
The last property of Theorem~\ref{thm:lean} says that a bag $S$ of the tree decomposition
cannot be broken by a separation of order at most $k$ into two parts that both contain a large portion of $S$. This property is often called {\em{unbreakability}}.
More formally,
let us recall the notion of a $(q,k)$-unbreakable set introduced in~\cite{ChitnisCHPP16,CyganLPPS19}.
Given a graph $G$ and integers $q \geq k \geq 0$, we say that a set $S \subseteq V(G)$
is \emph{$(q,k)$-unbreakable} if for every separation $(A,B)$ in $G$ of order at most $k$,
   we have $|A \cap S| \leq q$ or $|B \cap S| \leq q$.
Theorem~\ref{thm:lean} of~\cite{lean-decomp} can be seen as a simpler and cleaner version
of an analogous result of~\cite{CyganLPPS19}, where the bags are only guaranteed to be 
$(q,k)$-unbreakable for some $q$ bounded exponentially in $k$, and adhesion sizes are also bounded only exponentially in $k$.

Theorem~\ref{thm:lean} and its predecessor from~\cite{CyganLPPS19} have been used to design parameterized algorithms for several graph separation problems~\cite{CyganLPPS19,lean-decomp}, most notably for {\sc{Minimum Bisection}}. In most cases, when looking for a deletion set of size at most $k$, one performs dynamic programming on the  tree decomposition provided by Theorem~\ref{thm:lean}, where every step handling a single bag can be solved using color-coding thanks to the unbreakability of the bag. More recent applications include a parameterized approximation scheme for {\sc{Min $k$-Cut}}~\cite{Lokshtanov0S20}, as well as algorithms and data structures for problems definable in first-order logic with connectivity predicates~\cite{PilipczukSSTV21}.

In this paper we propose to use the ideas behind the tree decomposition of Theorem~\ref{thm:lean} in the context of \textsc{Graph Isomorphism} and \textsc{Canonization} in order to provide a reduction to the case when the input graph is suitably unbreakable. 
The next statement is a wishful-thinking theorem that in some variant we were able to prove,
but in the end the result turned out to be too cumbersome to use. 
In essence, it says the following: in FPT time one can compute a tree decomposition with the same qualitative properties as the one provided by Theorem~\ref{thm:lean}, but in an isomorphism-invariant way. 

\begin{notheorem}\label{thm:intro:decomp}
There is a computable function $f$ and an algorithm that,
given a graph $G$ and an integer $k$, in time $f(k) n^{\Oh(1)}$ computes
an isomorphism-invariant tree decomposition of $G$ such that
\begin{itemize}[nosep]
\item the sizes of adhesions are bounded by $f(k)$; and
\item every bag is either of size at most $f(k)$ or is $(f(k),k)$-unbreakable. 
\end{itemize}
\end{notheorem}

Instead of proving and using (a formal and correct version of) ``Theorem''~\ref{thm:intro:decomp}, we resort to the strategy used in the earlier works, namely recursive understanding.
Here, in some sense we compute a variant of the tree decomposition of ``Theorem''~\ref{thm:intro:decomp}
on the fly, shrinking the already processed part of the graph into constant-size representatives. Importantly, the run of this process is isomorphism invariant.
An overview of this approach is provided in Section~\ref{sec:overview}, while the full statement and its proof are in Section~\ref{sec:DS} and Section~\ref{sec:recursiveUnderstanding}.

At first glance,
``Theorem''~\ref{thm:intro:decomp} may seem unrelated to the setting of graphs excluding a fixed minor. However, let us recall one of the main results of the Graph Minor Theory, namely the Structure Theorem for graphs excluding a fixed minor, proved by Robertson and Seymour in~\cite{GM16}. 
Informally speaking, this statement says that if a graph $G$ excludes a fixed graph $H$ as a minor, then $G$ admits a tree decomposition (called henceforth the \emph{RS-decomposition}) in which the sizes of adhesions are bounded in terms of $H$ and the torso of every bag is nearly embeddable in a surface of Euler genus bounded in terms of $H$.
Without going into the precise definition of ``nearly'', let us make the following observation:
if we are given an $H$-minor-free graph $G$ and we apply to $G$ the algorithm of Theorem~\ref{thm:lean}
with $k$ equal to the expected bound on adhesion sizes in an RS-decomposition,
then the resulting tree decomposition should roughly resemble the RS-decomposition.
In particular, 
one may expect that the ``large'' bags of the decomposition of Theorem~\ref{thm:lean}
should also be nearly embeddable in some sense, as (due to unbreakability) they cannot be partitioned much finer in an RS-decomposition whose adhesions are of size at most $k$. 

For the \textsc{Graph Isomorphism} and \textsc{Canonization} problems,
the main issue now is that the output of Theorem~\ref{thm:lean} is not
necessarily isomorphism-invariant; this is where the wishful-thinking ``Theorem''~\ref{thm:intro:decomp} comes into play.
In particular, one could expect that the ``large'' bags of the decomposition
of ``Theorem''~\ref{thm:intro:decomp} would be nearly embeddable in some sense.
We are able to carry this intuition through 
the recursive understanding point of view, and show the following.

\begin{theorem}[simplified variant of Theorem~\ref{thm:mainReduction}]\label{thm:intro:toUnbreakable}
For every graph $H$ there exists a function $q_H\colon \mathbb{N}\to \mathbb{N}$ such that the following holds. Suppose there exists an integer $k$ and a canonization algorithm ${\mathcal{A}}$ that works on all graphs that are $H$-minor-free and $(q_H(i),i)$-unbreakable for all $i\leq k$. Then there is also a canonization algorithm ${\mathcal{B}}$ for all $H$-minor-free graphs. Furthermore, ${\mathcal{B}}$ can be chosen to run in time upper bounded by $g(H)\cdot n^{\Oh(1)}$ plus the total time taken by at most $g(H)\cdot n^{\Oh(1)}$ invocations of ${\mathcal{A}}$ on $H$-minor-free graphs with no more than $n$ vertices, where $g$ is some function.
\end{theorem}

The proof of Theorem~\ref{thm:intro:toUnbreakable} is sketched in Section~\ref{sec:overview}. The full version --- Theorem~\ref{thm:mainReduction} --- is proved in Sections~\ref{sec:DS} and~\ref{sec:recursiveUnderstanding}. In essence, the full version differs from the one above in that
it is uniform in $H$: if there is a single algorithm ${\mathcal{A}}$ that works for all $H$, then there is one resulting algorithm ${\mathcal{B}}$
that works for all $H$.

With Theorem~\ref{thm:intro:toUnbreakable} understood, 
it suffices to ``only'' provide a \textsc{Canonization} algorithm working on graphs that are $H$-minor-free and $(q(i),i)$-unbreakable for all $i\leq k$.
We comment here on the quantifier order: the algorithm needs to accept any unbreakability guarantee $q$, and then we have to choose the threshold $k$ to which this guarantee will be used based on $H$ and $q$. However,
both $k$ and the values of $q$ (up to $q(k)$) can be included as parameters in the final running time bound. That said, we prove the following statement.

\begin{theorem}[simplified variant of Theorem~\ref{thm:unbrk-minor}]\label{thm:intro:wrapper}
For every graph $H$ and function $q\colon \mathbb{N}\to \mathbb{N}$ there exists a constant $k$ and an algorithm that given an $H$-minor-free graph $G$ with a promise that $G$ is $(q(i),i)$-unbreakable for all $i\leq k$, outputs a canonical labeling of $G$ in time $f(\sum_{i=1}^{k} q(i))\cdot n^{\Oh(1)}$, where $f$ is a computable function.
\end{theorem}

Note that Theorem~\ref{thm:intro:wrapper} combined with Theorem~\ref{thm:intro:toUnbreakable} yield a non-uniform version of Theorem~\ref{thm:main}. Again, the full version of Theorem~\ref{thm:intro:wrapper} --- Theorem~\ref{thm:unbrk-minor} --- is an appropriately uniform statement that together with Theorem~\ref{thm:mainReduction} proves Theorem~\ref{thm:main} in full generality.

Let us now discuss the proof of Theorem~\ref{thm:unbrk-minor}.
Observe that if we set $k$
higher than the bound on adhesion sizes in an RS-decomposition for $H$-minor-free graphs, then
we expect any RS-decomposition of a $(q(k),k)$-unbreakable $H$-minor-free graph $G$
to have one central bag (with a nearly embeddable torso) that may be huge, while all other bags
have sizes at most $q(k)$. 
Thus, it seems natural that a procedure for  \textsc{Canonization} on nearly embeddable graphs can be lifted to a procedure working on graphs as above.

Proving Theorem~\ref{thm:intro:wrapper} requires involved analysis of near-embeddings of highly
unbreakable graphs with a fixed excluded minor. We perform it in a subordinate paper~\cite{subordinate}.
The main result of~\cite{subordinate} can be phrased in the following simplified form.

\begin{theorem}[simplified version of the main result of~\cite{subordinate}, stated formally as Theorem~\ref{thm:rigid}]\label{thm:intro:rigid}
For every graph $H$ and function $q\colon \mathbb{N}\to \mathbb{N}$ there exists a constant $k$,
   computable functions $\funtime,\funtw,\fungenus\colon \mathbb{N}\to \mathbb{N}$,
   and an algorithm that given a graph $G$ that is $H$-minor-free and $(q(i),i)$-unbreakable for all $i\leq k$, works in time bounded by $\funtime(\sum_{i=1}^{k} q(i))\cdot n^{\Oh(1)}$ and computes a partition of 
$V(G) = \Vtw \uplus \Vgenus$ and a nonempty family $\famgenus$ of bijections $\Vgenus \mapsto [|\Vgenus|]$ 
such that 
\begin{enumerate}[nosep]
\item the partition $V(G) = \Vtw \uplus \Vgenus$ is isomorphism-invariant;
\item $G[\Vtw]$ has treewidth bounded by $\funtw(\sum_{i=1}^{k} q(i))$;
\item $|\famgenus| \leq \fungenus(\sum_{i=1}^{k} q(i)) \cdot n^{\Oh(1)}$; and
\item $\famgenus$ is isomorphism-invariant.
\end{enumerate}
\end{theorem}

The statement of Theorem~\ref{thm:intro:rigid} is quite technical, so let us provide more intuitive explanation. We are given an $H$-minor-free graph $G$ with sufficiently strong 
unbreakability properties. Then the claim is that $G$ can be vertex-partitioned  in an isomorphism-invariant way
into two parts $\Vtw$ and $\Vgenus$. The part $\Vgenus$ is {\em{rigid}} in the sense that the subgraph induced by it admits a polynomially-sized isomorphism-invariant family of labelings. The other part $\Vtw$ may have multiple automorphisms, but the subgraph induced by it has bounded treewidth. 
Informally speaking, the bounded-treewidth part $\Vtw$ always contains all the ``near-'' elements of a 
near-embedding: vortices, apices, etc., while the rigid part $\Vgenus$ contains the core of the embedded part; its rigidity is witnessed by the embedding.

Theorem~\ref{thm:intro:wrapper} follows quite easily from Theorem~\ref{thm:intro:rigid} combined with the FPT canonization procedure on graphs of bounded treewidth~\cite{LokshtanovPPS17}
(formal arguments can be found in Section~\ref{sec:unbrk-minor}). So the main weight of argumentation is contained in the proof of Theorem~\ref{thm:intro:rigid}. This proof is
the main result of the subordinate paper~\cite{subordinate}.

\medskip

Finally, let us remark that Theorem~\ref{thm:mainReduction} --- the full version of Theorem~\ref{thm:intro:toUnbreakable} --- is stated and proved in terms of graph classes defined by forbidding {\em{topological minors}}. Thus, it can be readily used also to reduce the {\sc{Canonization}} problem on $H$-topological-minor-free graphs to the same problem on suitably unbreakable $H$-topological-minor-free graphs. For the latter setting, one could apply the Structure Theorem of Grohe and Marx~\cite{marx-grohe} to reason that unbreakable $H$-topological-minor-free graphs are either close to being nearly embeddable (and this case is treated in this paper), or they essentially have bounded maximum degree. In spirit, this reduces the problem of finding an FPT algorithm for {\sc{Canonization}} on $H$-topological-minor-free graphs to the problem of finding such an algorithm on graphs of bounded maximum degree. We refrain from expanding this discussion in this paper, as it is tangential to the main direction of our work.

%% file: overviewDeco.tex
In this section we provide an informal overview of the proof of Theorem~\ref{thm:intro:toUnbreakable}; the formal statement and full proof can be found in Section~\ref{sec:DS}. For convenience, we restate Theorem~\ref{thm:intro:toUnbreakable} here. One important simplification made in the statement below is that we only consider excluding a {\em single} graph $H$ as a {\em minor}, instead of a finite set of graphs as topological minors. 

\medskip
\noindent
{\bf Theorem~\ref{thm:intro:toUnbreakable} (restated, simplified variant of Theorem~\ref{thm:mainReduction}) }{\em 
For every graph $H$ there exists a function $q_H\colon \mathbb{N}\to \mathbb{N}$ such that the following holds. Suppose there exists an integer $k$ and a canonization algorithm ${\mathcal{A}}$ that works on graphs that are $H$-minor-free and $(q_H(i),i)$-unbreakable for all $i\leq k$. Then there is also a canonization algorithm ${\mathcal{B}}$ for $H$-minor-free graphs. Furthermore, ${\mathcal{B}}$ runs in time upper bounded by $g(H)\cdot n^{\Oh(1)}$ plus the total time taken by at most $g(H)\cdot n^{\Oh(1)}$ invocations of ${\mathcal{A}}$ on $H$-minor-free graphs with no more than $n$ vertices, where $g$ is some function.
}

\smallskip
The proof of Theorem~\ref{thm:intro:toUnbreakable} is based on the {\em recursive understanding} technique, pioneered in~\cite{GroheKMW11,KawarabayashiT11}, and explicitly defined and further refined in~\cite{ChitnisCHPP16}.

We first discuss the essence of the recursive understanding technique, as it has been applied in the past, in the context of some generic problem. Much of this discussion does {\em not} apply to {\sc Canonization}, but it helps motivate our approach and naturally introduces the crucial notions of representatives and replacement. 

Recursive understanding based algorithms proceed as follows. Either the input graph is already $(q, k)$-unbreakable, then we are already done, because the goal is to reduce the original problem to the problem on unbreakable graphs. 
Otherwise there exists a separation $(A, B)$ of order at most $k$ such that both $A$ and $B$ have size $q$, which is much bigger than $k$.
The algorithm calls itself recursively on $G[A]$, and after thoroughly analyzing the graph $G[A]$, produces a {\em representative} $G_A^R$ for it. 

The representative $G_A^R$ is a tiny graph on at most $f(k) < q$ vertices, and all of the vertices of $G_A^R$ are new vertices that are not present in $G$, except that $G_A^R$ shares the vertices $D = A \cap B$ with $G$ (we will require that $G_A^R[D] = G[D]$.) 
The algorithm then {\em replaces} $G[A]$ in $G$ with $G_A^R$. This means that the algorithm removes all vertices of $A \setminus B$ from $G$, and attaches $G_A^R$ to $D$ instead. We call the resulting graph $G^\star$. Note that $|V(G^\star)| \leq |B| + q$. 
The algorithm then calls itself recursively on $G^\star$, and finally lifts the solution on $G^\star$ to a solution of $G$.

We upper bound the running time of the algorithm by a function $T(n, k)$ of the number of vertices and the parameter $k$ - the size of the separators. $T(n, k)$ satisfies the following recurrence.
\begin{align}\label{eqn:recUndRecurrence}
T(n, k) \leq S(n, k) + T(|A|, k) + T(|B|+f(k), k) + L(n, k)
\end{align}
Here, $S(n, k)$ denotes the time to find the separation $(A, B)$ or decide that such a separation does not exist, $T(|A|,k)$ is the running time of the algorithm on $G[A]$, $T(|B|+f(k),k)$ is the running time of the algorithm on $G^\star$, and $L(n, k)$ is the time it takes to lift a solution to $G^\star$ back to a solution of $G$.
A simple recurrence analysis of Equation~\ref{eqn:recUndRecurrence} shows that $T(n, k) \leq (S(n,k) + L(n,k)) \cdot g(k) \cdot n$. 
Thus, as long as we are able to find a separator in FPT time, lift solutions of $G^\star$ back to $G$ in FPT time, and solve the base case of $(q, k)$ unbreakable graphs in FPT time, we get an FPT algorithm for the original problem.

Being able to execute the above scheme rests on a crucial property of the relationship between the graph $G[A]$ and the representative $G_A^R$ that we replace it with. 
In particular a solution to the reduced graph $G^\star$ needs to be useful for the lifting procedure in order to recover the solution to $G$. This is tricky, because the algorithm needs to compute $G_A^R$ from $G[A]$ without looking at $G[B]$.
In particular this means that the reduction from $G[A]$ to $G_A^R$ has to work for {\em every} possible $G[B]$. 
This is the reason for the name ``{\em recursive understanding}'' - the procedure needs to ``understand'' $G[A]$ to such an extent that it can efficiently lift a solution of $G^\star$ to a solution of $G$, pretty much independently of what $G^\star$ is (since $G[A]$ and $G^\star$ only have $k$ vertices in common). This typically means that our algorithm should not just solve the original problem, but a generalized version of the problem that takes as input the graph $G[A]$, together with the set $D \subseteq A$, and the promise that ``the rest of the graph'' will attach to $G[A]$ (and $G_A^R$) only via $D$. The task is to compute sufficient information about $(G[A], D)$ to be able to produce $G_A^R$.

While this task may appear much more difficult than the original problem, for many computational problems the ``recursive understanding'' task is no harder than the original problem in the formal sense that an FPT algorithm for the original problem implies the existence of an FPT algorithm for recursive understanding~\cite{LokshtanovR0Z18}. However no such results were previously known for {\sc Graph Isomorphism} or {\sc Canonization}, and it was far from obvious that such a statement would even be true for these problems. 

The first main hurdle in applying recursive understanding for {\sc Canonization} is that we need to be able to find a separation $(A, B)$ of order at most $k$ and satisfying $\min\{|A|, |B|\} > q$ in an {\em isomorphism invariant} way. 
Here, by isomorphism invariant we mean that if two graphs $G$ and $\hat{G}$ are isomorphic, then if we run the algorithm on $G$ and find a separation $(A, B)$ of $G$, and on $\hat{G}$, and find a separation $(\hat{A}, \hat{B})$ of $\hat{G}$ then {\em every} isomorphism from $G$ to $\hat{G}$ maps $A$ to $\hat{A}$ and $B$ to $\hat{B}$. 
We have no idea whatsoever how to find a single separation $(A, B)$ of $G$ of order at most $k$ with $\min\{|A|, |B|\} \geq q$ (or determine that such a separation does not exist) in an isomorphism invariant way in FPT time. Indeed, the task of finding such a separation looks as difficult as canonizing $G$.

For this reason we turn to an easier problem, finding a single set $B \subseteq V(G)$ in an isomorphism invariant way such that (i) $B$ is either small (that is has size at most $f(H)$), or it has some nice properties - for example $B$ could be $(q, k)$-unbreakable in $G$ (for some $q$ and $k$ upper bounded by a function of $H$), and (ii) every connected component $C$ of $G-B$ has at most $k$ neighbors in $B$.

It is convenient to formalize the decomposition of $G$ into $B$ and the remaining components in terms of tree decompositions. We follow the notation of~\cite{marx-grohe} for tree decompositions, for formal definitions see Section~\ref{sec:treeDecAndStarDeco}. We say that a {\em star decomposition} is a tree decomposition $(T, \chi)$ where $T$ is a star rooted at its center node $b$ (the unique node of degree larger than $1$ in $T$).
The {\em bags} of the decomposition are the sets $\{\chi(v) : v \in V(T)\}$ and the {\em adhesions} are the sets $\{ \sigma(\ell) = \chi(\ell) \cap \chi(b) : \ell \in V(T) \setminus \{b\} \}$. 

We will frequently need to find an isomorphism invariant star decomposition $(T, \chi)$ of $G$ with center bag $b$ such that $\chi(b)$ has some nice properties, and all adhesions have size at most $k$. Here isomorphism invariant means that if $G$ and $\hat{G}$ are isomorphic with star decompositions $(T, \chi)$ and $(\hat{T}, \hat{\chi})$ respectively, then for every isomorphism $\phi$ of $G$ to $\hat{G}$ there exists an isomorphism $\phi_T : V(T) \rightarrow V(\hat{T})$ such that for every vertex $u \in V(G)$ and node $v \in V(T)$ it holds that $u \in \chi(v)$ if and only if $\phi(u) \in \hat{\chi}(\phi_T(v))$. Informally, every isomorphism $\phi$ that maps $G$ to $\hat{G}$ also maps $(T, \chi)$ to $(\hat{T}, \hat{\chi})$.
We will defer the discussion of how exactly we find such decompositions to the end of this section, because there are many different cases, and how exactly we find $(T, \chi)$ depends on the context. 

\paragraph{Isomorphism Invariant Recursive Understanding.}
We are now ready to flesh out the overall scheme used in our algorithm. From now on, we will assume that all parameters that we introduce are upper bounded by some function of $H$, unless explicitly stated otherwise. 

We are trying to find a canonical labeling for a graph $G$. However, just as in the recursive understanding scheme discussed in the beginning of this section, our algorithm will also take as input a distinguished set $D$. The interpretation is again that $G$ is not necessarily the entire instance that we are trying to solve, and that the remaining ``hidden'' part of the instance attaches to $G$ via $D$.

The first step is to compute an isomorphism invariant star decomposition $(T, \chi)$ of $G$ with central bag $b \in V(T)$, such that $B = \chi(b)$, $D \subseteq B$, and $B$ either has size upper bounded by $f(H)$ or has some nice properties (such as being $(q, k)$-unbreakable in $G$, or that for every pair of vertices $u$, $v$ in $B$ there is no $u$-$v$ separator of size at most $k$.) Additionally we require all adhesions of $(T, \chi)$ to have size at most $f(k)$. 
Once we have identified $(T, \chi)$ we run the algorithm recursively on $G[\chi(\ell)]$ for every leaf $\ell$ in $V(T)$. Having ``{\em understood}'' each of the leaves we need to use this understanding to ``{\em understand}'' $G$ with distinguished set $D$. 

At this point we need to unpack precisely what we mean by understanding in the context of canonization. 
Because we know that every automorphism of $G$ maps $(T, \chi)$ to itself we know that the center bag of $(T, \chi)$ maps to itself, while every leaf $\ell$ of $(T, \chi)$ maps in its entirety to itself, or to some other leaf $\ell'$ of $(T, \chi)$.
We would like to use the ``understanding'' from the recursive calls in order to determine which leaves $\ell'$ the leaf $\ell$ can map to. Additionally, an automorphism of $G$ might map $\ell$ to itself, but permute the vertices of the adhesion $\sigma(\ell)$. Therefore we also need to be able to use the ``understanding'' from the recursive calls to determine the set of permutations from $\sigma(\ell)$ to $\sigma(\ell)$ that can be completed to an automorphism of $G[\chi(\ell)]$.

It can be shown that both of these goals can be achieved if we formalize the ``understanding'' task for the leaves as follows: for every bijection $\pi : \sigma(\ell) \rightarrow [|\sigma(\ell)|]$ we need to find a canonical labeling of $G[\chi(\ell)]$ that coincides with $\pi$ on $\sigma(\ell)$.
Because the distinguished set $D$ is actually one such $\sigma(\ell)$ in the recursive call corresponding to the parent of the current call in the recursion tree, the formalization of the ``understanding'' task for $G$ with distinguished set $D$ becomes to compute for every bijection $\pi : D \rightarrow [|D|]$ a canonical labeling of $G$ that coincides with $\pi$ on $D$.

The recursive scheme (compute $(T, \chi)$, understand all leaves recursively) now leaves us with the following task.
We have as input a graph $G$ and vertex set $D$ and an isomorphism invariant star decomposition $(T, \chi)$ with central bag $b$, such that $D \subseteq \chi(b)$, together with a set of canonical labelings 
$$\{\lambda_{\pi, \ell} : \ell \in V(T) \setminus \{b\} \mbox{ and } \pi : \sigma(\ell) \rightarrow [|\sigma(\ell)|] \mbox{ is a bijection}\}.$$
Here $\lambda_{\pi, \ell}$ is a canonical labeling of $G[\chi(\ell)]$ that coincides with $\pi$ on $\sigma(\ell)$.
The goal is to compute for every permutation $\pi : D \rightarrow [D]$ a canonical labeling $\lambda_\pi$ of $G$ that coincides with $\pi$ on $D$.

At this point some case work is in order. The easy case is when the size of the center bag $B$ is upper bounded by a function of $H$. In this case we can proceed in a manner almost identical to a single step of the dynamic programming step in the algorithm of~\cite{LokshtanovPPS17} or even~\cite{Bodlaender90}. The idea is to simply brute force over all permutations of $B$. 

The hard case is when $B$ is large. In this case we need to use the fact that $B$ has some ``nice'' property. For now, let us not worry precisely what the property is, because the next crucial step of the algorithm works independently of it. 

\paragraph{Unpumping and Lifting.}
The case when the central bag $B$ is large is the only step of the algorithm in which we actually employ the ``replacing a graph by its representative'' part of the recursive understanding technique. 
The {\em unpumping} procedure takes as input $G$ with a distinguished set $D$, $(T, \chi)$, as well as the set of labelings $\{\lambda_{\pi, \ell} \}$ and produces the graph $G^\star$ by replacing every leaf $G[\chi(\ell)]$ with distinguished set $\sigma(\ell)$ by representative $G_\ell^R$ whose size is upper bounded by a function only of $k$.
We will shortly discuss in more detail the properties of the representatives, but before that let us state the properties of the unpumped graph $G^\star$ that we need.

\begin{enumerate}
\setlength{\itemsep}{-2pt}
    \item (Lifting) For every permutation $\pi : D \rightarrow D$ a canonical labeling $\lambda_\pi^\star$ of $G^\star$ that coincides with $\pi$ on $D$ can be lifted in polynomial time to a canonical labeling $\lambda_\pi$, which coincides with $\pi$ on $D$, of $G$.
    
    \item (Feasibility) If $G$ is $H$-minor-free then $G^\star$ is $H$-minor-free.
    
    \item (Maintains Cut Properties Of $B$) If $B$ is $(q, k)$-unbreakable in $G$ then $B$ is also $(q, k)$-unbreakable in $G^\star$. If, for every $u, v \in B$ there is no $u$-$v$ separator of size at most $k$ in $G$ then there is no such separator in $G^\star$.
    
    \item (Small Leaves) $G^\star$ has a star decomposition $(T, \chi^\star)$ with the same decomposition star $T$ as $G$, with $\hat{\chi}(b) = \chi(b)$ and $|\hat{\chi}(\ell)| \leq g(k)$ for some function $g$. {\em Remark}: This function $g$ is possibly not computable. 
\end{enumerate}

\paragraph{Representatives For Canonization.}
Before proceeding to discussing how the algorithm uses the unpumping/lifting procedure we need to discuss how it is able to guarantee the key properties of $G^\star$. These properties follow by how we define the representative $G_\ell^R$ for each leaf $(G[\chi(\ell)], \sigma(\ell))$. Specifically, when we compute a representative of $\ell$ we need to ensure that:

\begin{itemize}
\item (Lifting I) The isomorphism class of $G_\ell^R$ depends only on the isomorphism class of $(G[\chi(\ell)], \sigma(\ell))$.

\item (Lifting II) For all permutations $\iota : \sigma(\ell) \rightarrow \sigma(\ell)$, $(G[\chi(\ell)], \sigma(\ell))$ has an automorphism that coincides with $\iota$ on $\sigma(\ell)$ if and only if $G_\ell^R$ does.

\item (Feasibility) $G_\ell^R$ contains exactly the same set of minors on at most $|V(H)|$ vertices, and these minors intersect with $\sigma(\ell)$ in exactly the same way in  $G_\ell^R$ and in  $G[\chi(\ell)]$.

\item (Maintains Cut Properties) For every separation $(L, R)$ of $G[\sigma(\ell)]$ the minimum order of a separation of $G[\chi(\ell)]$ that coincides with $(L, R)$ on $\sigma(\ell)$ and the minimum order of a separation of $G_\ell^R$ that coincides with $(L, R)$ on $\sigma(\ell)$ are the same.

\item (Small Leaves) $|V(G_\ell^R)|$ is upper bounded by a function only of $H$ and $k$, and therefore only of $H$.
\end{itemize}

It is fairly easy to see that the properties (Feasibility), (Maintains Cut Properties) and (Small Leaves) of unpumping follow from the corresponding properties of representatives. 
On the other hand, an attentive reader should be worried about (Lifting) for the unpumping property to follow from properties (Lifting I) and (Lifting II) of representatives.
Indeed, the graph $G[\chi(\ell)]$ is possibly a large graph on up to $n$ vertices, while the size of $G_\ell^R$ is bounded as a function of $H$.
Thus, by pigeon hole principle non-isomorphic graphs $G[\chi(\ell)]$ may end up having the same representative $G_\ell^R$.

How can we now ensure that such leaves do not get mapped to each other by automorphisms of $G^\star$? Well, before unpumping we can use the canonical labelings of $G[\sigma(\ell)]$ for every leaf $\ell$ to determine which leaves $\ell$ and $\ell'$ have isomorphic graphs $G[\sigma(\ell)]$ and $G[\sigma(\ell')]$. Then we encode non-isomorphism of non-isomorphic leaves in $G$ by assigning in $G^\star$ colors to vertices in the representatives
so that the representatives of non-isomorphic leaves receive distinct colors. 
Specifically, assign an ``identification number" to each isomorphism class of leaves $G[\sigma(\ell)]$. For each leaf $\ell$, think of colors as integers and add to the color of (every vertex in) $G^\star[\sigma(\ell)]$ the identification number of the isomorphism class of $G[\sigma(\ell)]$.

\paragraph{$G^\star$ inherits Properties of $B$.}
As we already have alluded to before, the properties of $B$ that we will be maintaining are either that $B$ is $(q, k)$-unbreakable in $G$ and in $G^\star$, or even stronger, that there is no $x$-$y$ separator of size at most $k$ in $G$ and in $G^\star$ for every pair of vertices $x$, $y$ in $B$. 

This, together with the property (Small Leaves) can be used to show that $G^\star$ itself is $(q', k)$-unbreakable for $q'$ that is only a function of $q$ and the size of the leaves (in the first case), or is is $(q, k, k)$-{\em improved clique unbreakable} (in the second case). Here $(q, k, k)$-improved clique unbreakable means that there is no separation $(L, R)$ of $G$ of order at most $k$ such that $\min\{|L|, |R|\} \geq q$ and for every pair $x$, $y$ of vertices in $L \cap R$ the minimum order of an $x$-$y$ separation is at least $k+1$.
Thus, the good property of the bag $B$ in $G$ is in an approximate sense inherited by the entire unpumped graph $G^\star$.

\paragraph{Overall Scheme, Again.}
Our reduction from general graphs to unbreakable graphs proceeds in two steps. In the first step we reduce from general graphs to improved clique-unbreakable graphs, while in the second step we reduce from improved clique-unbreakable graphs to unbreakable graphs.

In the first step, the reduction from general graphs to improved clique-unbreakable graphs, all we still have to do is to design an algorithm that given a general graph $G$ and set $D$, outputs an isomorphism invariant star decomposition $(T, \chi)$ of $G$ such that the central bag $B$ of $(T, \chi)$ satisfies that there is no $x$-$y$ separator of size at most $k$ in $G$ for every pair of vertices $x$, $y$ in $B$. If we can achieve this, then the argument in the previous section yields that $G^\star$ is improved clique unbreakable, and we have achieved a reduction to improved clique unbreakable graphs.  This can be done by essentially observing that it was already done by  Elberfeld and Schweitzer~\cite{ElberfeldS17}. However we need to modify their algorithm to work in FPT time rather than log space and $O(n^k)$ time.

In the second case we have $(G,D)$ and the promise that $G$ is improved clique unbreakable.  What we need is to find a star decomposition $(T,\chi)$ where center $B$ is unbreakable, because then $G^\star$ will be unbreakable as desired.

\paragraph{Finding next bag in improved clique unbreakable case.}

Instead of considering the $(q, k, k)$-improved clique unbreakable case, we will make the simplifying assumption that $G$ has no separator $(L, R)$ of order at most $k$, such that for every $x$, $y$ in $L \cap R$ the minimum order of an $x$-$y$ separation is at least $k+1$. 

Again we are given as input a graph $G$ with a distinguished vertex set $D$ and the task is to find a isomorphism invariant star decomposition $(T, \chi)$ with central bag $B$ such that $D \subseteq B$ and $B$ is $(q, k)$-unbreakable in $G$.

There are two cases, either $D$ is small ($|D| \leq k$) or $D$ is large. If $D$ is small and there exists a pair $x$, $y$ such that the minimum order of an $x$-$y$ separation in $G$ is at most $k$, we can proceed in a similar manner as the corresponding case in the algorithm of~\cite{LokshtanovPPS17}. 

On the other hand, if $D$ is small and no such pair exists, then $D$ yields precisely the type of separation in $G$ that we just assumed does not exist. Therefore this case does not apply. This is the only place in the algorithm where we use the assumption that $G$ is $(q, k, k)$-improved clique unbreakable.

If $B$ is large and breakable (there exists a separation $(L, R)$ of order $i < \min(|L \cap B|,|R \cap B|)$ then we can use the ``notion of stable separators'' defined in \cite{LokshtanovPPS17} and make progress. Finally, when $B$ is large and unbreakable, we observe that the idea of important separator extension used in designing an algorithm for {\sc Minimum Bisection}~\cite{CyganLPPS13} gives an unbreakable bag $B$ in an isomorphism invariant way. 

In either case we are able to find a star decomposition $(T, \chi)$ where the central bag $B$ is either small or unbreakable. When it is small we can brute force, while when it is big and unbreakable, the unpumped graph $G^\star$ is unbreakable and so we can solve the problem for the unbreakable graph $G^\star$ and lift the solution back to $G$ using the lifting algorithm. This concludes the proof sketch of Theorem~\ref{thm:intro:toUnbreakable}.

%% file: prelims.tex
In most cases, we use standard graph notation, see e.g.~\cite{diestel}. For a graph $G$ we use $V(G)$ and $E(G)$ to denote the set of vertices and edges respectively. We use $n$ to denote $|V(G)|$ and $m$ to denote $|E(G)|$. For an undirected graph $G$ the neighborhood of a vertex $v$ is denoted by $N(v)$ and defined as $N(v) = \{u \colon uv \in E(G)\}$. The {\em closed neighborhood} of v$ $ is $N[v] = N(v) \cup \{v\}$. For a vertex set $S$, $N[S] = \bigcup_{v \in S} N[v]$ while $N(S) = N[S] \setminus S$. For a vertex set $S$ the {\em subgraph of $G$ induced by $S$} is the graph $G[S]$ with vertex set $S$ and edge set $\{uv \in E(G) ~:~ \{u,v\} \subseteq S \}$. {\em Deleting} a vertex set $S$ from $G$ results in the graph $G - S$ which is defined as $G[V(G) \setminus S]$. Deleting an edge set $S$ from $G$ results in the graph $G - S $ with vertex set $V(G)$ and edge set $E(G) \setminus S$. Deleting a single vertex or edge from a graph $G$ is the same as deleting the set containing only the vertex or edge from $G$. A {\em subgraph} of $G$ is a graph $H$ which can be obtained by deleting a vertex set and an edge set from $G$. {\em Contracting} an edge $uv$ in a graph $G$ results in the graph $G / uv$ obtained from $G - \{u, v\}$ by adding a new vertex $w$ and adding edges from $w$ to all vertices in $N(\{u, v\})$. Contracting a set $S$ of edges results in the graph $G / S$, obtained from $G$ by contracting all edges in $S$. It is easy to see that the order of contraction does not matter and results in the same graph, and therefore the graph $G / S$ is well defined. If $H  = G / S$ for some edge set $S$ we say that $H$ is a {\em contraction} of $G$. $H$ is a {\em minor} of $G$ if $H$ is a contraction of a subgraph of $G$. 
A {\em partition} of a set $S$ is a family ${\cal P}$ of pairwise disjoint subsets of $S$ such that $S = \bigcup_{P \in {\cal P}} P$.

A {\em topological minor model} of a graph $H$ in a graph $G$ is a pair $(f, P)$ where $f : V(H) \rightarrow V(G)$ is an injective function and $P$ is a function that for every edge $uv \in E(H)$ outputs a path $P(uv)$ in $G$ from $f(u)$ to $f(v)$ such that $P(uv)$ does not contain any vertices from $f(V(H))$ other than $f(u)$ and $f(v)$, and for every pair $uv$, $u'v'$ of edges in $E(H)$ the two paths $p(uv)$ and $p(u'v')$ do not have any internal vertices in common. If there exists a topological minor model of $H$ in $G$ we say that $H$ is a {\em topological minor} of $G$, and we denote it by $H \sqsubseteq G$.

\paragraph{Functions, Inverses, Compositions}
A {\em function} $f \colon D \rightarrow R$ assigns a value $f(d)$ from $R$ to every element $d \in D$. The sets $R$ and $D$ are called the domain and range of $f$ respectively. A function $f \colon D \rightarrow R$ is {\em injective} if for every pair of distinct elements $d_1, d_2 \in D$ we have $f(d_1) \neq f(d_2)$. It is {\em surjective} if for every $r \in R$ there exists a $d \in D$ so that $f(d) = r$, and it is {\em bijective} if it is both injective and surjective. 

We will frequently abuse notation and a consider a function $f \colon D \rightarrow R$ to act also on subsets of $D$. In that case, for $S\subseteq D$ we write $f(S) = \{f(d)~\colon~ d \in S\}$. The {\em inverse} of a function $f$ is the function $f^{-1}$ that assigns to every element $r \in R$ the {\em set} $\{d \in D ~\mid~ f(d) = r\}$. For injective functions $f$ we have that $f^{-1}(r)$ is either the empty set or a set of size $1$. For injective functions we will often assign a different meaning to the function $f^{-1}$, namely it is the function $f^{-1} \colon f(D) \rightarrow D$ where for any $r\in f(D)$ we set $f^{-1}(r)$ to be the unique $d\in D$ satisfying $f(d)=r$. Which of the meanings we assign to the inverse will be always clear from context.

For two functions $f \colon D \rightarrow R$ and $g : R \rightarrow R'$ we define the {\em composition} of $f$ and $g$ to be the function $g \circ f$ where $(g \circ f)(x) = g(f(x))$ for every $x \in D$. We will often drop the $\circ$ and simply write $gf$ for the function $g \circ f$. We will sometimes compose functions $f$ and $g$ even if the domain of $g$ does not contain the full range of $f$. In particular if  $f \colon D \rightarrow R_f$ and $g \colon D_g \rightarrow R_g$ are functions then $g \circ f$ (and $gf$) is the function with domain $D_{gf} = \{d \in D ~\mid~ f(d) \in D_g\}$ so that $gf(d) = g(f(d))$ for every $d \in D_{gf}$.

\subsection{Separations, Separators, and Clique Separators} We recall here standard definitions and facts about separations and separators in graphs.
For two disjoint vertex sets $A,B \subseteq V(G)$, by $E(A,B)$ we denote the set of edges with one endpoint in $A$ and the second endpoint in $B$.

\begin{definition}[\bf separation, separator]
A pair $(A,B)$ where $A \cup B = V(G)$ is called a {\em separation}
if $E(A \setminus B, B \setminus A) = \emptyset$.
The {\em{separator}} of $(A,B)$ is $A\cap B$ and the {\em order} of a separation $(A,B)$ is $|A \cap B|$.
\end{definition}

Let $X,Y\subseteq V(G)$ be two not necessarily disjoint subsets of vertices. Then a separation $(A,B)$ is an {\em{$X$-$Y$ separation}} if $X\subseteq A$ and $Y\subseteq B$. The classic Menger's theorem states that for given $X,Y$, the minimum order of an $X$-$Y$ separation is equal to the maximum vertex-disjoint flow between $X$ and $Y$ in $G$. This minimum order is denoted further $\conn(X,Y)$. Moreover, among the $X$-$Y$ separations of minimum order there exists a unique one with inclusion-wise minimal $A$, and a unique one with inclusion-wise minimal $B$. We will call these minimum-order $X$-$Y$ separations {\em{pushed towards $X$}} and {\em{pushed towards $Y$}}, respectively.

For two vertices $x,y\in V(G)$, by $\conn(x,y)$ we denote the minimum order of a separation $(A,B)$ in $G$ such that $x\in A\setminus B$ and $y\in B\setminus A$. Note that if $xy\in E(G)$ then such a separation does not exist; in such a situation we put $\conn(x,y)=+\infty$. Again, classic Menger's theorem states that for nonadjacent $x$ and $y$, the value $\conn(x,y)$ is equal to the maximum number of internally vertex-disjoint $x$-$y$ paths that can be chosen in the graph. The notions of minimum-order separations pushed towards $x$ and $y$ are defined analogously as before.
If the graph we are referring to is not clear from the context, we put it in the subscript by the symbol $\conn$.

We emphasize here that, contrary to $X$-$Y$ separations, in an $x-y$ separation we require $x \in A \setminus B$ and $y \in B \setminus A$, that is, the separator $A \cap B$
cannot contain $x$ nor $y$. This, in particular, applies to minimum-order $x-y$ separations pushed towards $x$ or $y$.
It is known that the Ford-Fulkerson algorithm can actually find the aforementioned separations efficiently, as explained in the statement below.
\begin{lemma}\label{lem:f-f}
Given a graph $G$, two vertex sets $X,Y \subseteq V(G)$, and an integer $k$,
one can in time $\Oh(k(|V(G)| + |E(G)|))$ find the minimum-order $X$-$Y$-separation pushed towards $X$,
or correctly conclude that the minimum order of such separation is greater than $k$.
An analogous claim holds for vertices $x,y \in V(G)$ and the minimum-order $x$-$y$ separations pushed towards $x$.
\end{lemma}

We now move to the notion of a \emph{clique separation}.

\begin{definition}[\bf clique separation]
A separation $(A,B)$ is called a {\em{clique separation}} if $A\setminus B\neq \emptyset$, $B\setminus A\neq \emptyset$, and $A\cap B$ is a clique in $G$.
\end{definition}

Note that an empty set of vertices is also a clique, hence any separation with an empty separator is in particular a clique separation. We will say that a  graph is {\em{clique separator free}} if it does not admit any clique separation. Such graphs are often called also {\em{atoms}}, see e.g.~\cite{Leimer93,BerryPS10}.  From the previous remark it follows that every clique separator free graph is connected.

Next we define the notion of an important separator and of unbreakability.
\begin{definition}[\bf important separator]
An inclusion-wise minimal $X-Y$ separator $W$ is called
an {\emph{important}} $X-Y$ separator if there is no
$X-Y$ separator $W'$ with $|W'| \le |W|$ and
$R_{G\setminus W}(X \setminus W) \subsetneq R_{G \setminus W'}(X \setminus W')$,
where $R_H(A)$ is the set of vertices reachable from $A$ in the graph $H$.
\end{definition}

\begin{lemma}[\cite{imp-seps-chen,multicut}]
\label{lem:imp-seps}
For any two sets $X,Y \subseteq V(G)$ there are at most $4^k$
important $X-Y$ separators of size at most $k$
and one can list all of them in $\Oh(4^k k (n+m))$ time.
\end{lemma}

Finally, we proceed to the definition of unbreakability.

\begin{definition}[\bf $(q,k)$-unbreakable set]\label{def:unbreakSetSingle}
Let $G$ be an undirected graph.
We say that a subset of vertices $A$ is {\emph{$(q,k)$-unbreakable in $G$}},
if for any separation $(X,Y)$ of order at most $k$
we have 
$|X \cap A| \le q$ or $|Y \cap A| \le q$.
Otherwise $A$ is $(q,k)$-{\emph{breakable}}, and any separation $(X,Y)$ certifying this is called a {\emph{witnessing separation}}.
\end{definition}

\begin{definition}[\bf unbreakable set]
Let $G$ be an undirected graph.
We say that a subset of vertices $A$ is {\emph{unbreakable in $G$}},
if for any separation $(X,Y)$ 
we have  that  $|X \cap Y| \geq \min\{|X \cap A|, |Y \cap A|\}$. 
Otherwise $A$ is {\emph{breakable}}, and any separation $(X,Y)$ certifying this is called a {\emph{witnessing separation}}.
\end{definition}

If a set $A$ of size at least $3q$ is $(q,k)$-unbreakable then removing any $k$ vertices from $G$
leaves almost all of $A$, except for at most $q$ vertices, in the same connected component.
In other words, one cannot separate two large chunks of $A$ with a small separator. 
Observe that if a set $A$ is $(q,k)$-unbreakable in $G$, then any of its subset
$A' \subseteq A$ is also $(q,k)$-unbreakable in $G$. 
Moreover, if $A$ is $(q,k)$-unbreakable in $G$, then $A$ is also $(q,k)$-unbreakable in any supergraph of $G$. 
We will also need a weaker notion of unbreakability where we only care about clique separators in $G$.

\begin{definition} A graph $G$ is $(q, k)$-clique-unbreakable if for every clique separation $(L, R)$ of $G$ order at most $k$ it holds that $\min(|L|, |R|) \leq q$. 
\end{definition}

Our algorithms crucially rely on the following strengthening of the notion of unbreakability from Definition~\ref{def:unbreakSetSingle}. Here we care about separations of all orders $i$ for some $i \leq k$.

\begin{definition}
Let $k$ be an integer and $q : [k] \rightarrow  \mathbb{N}$ be a function.  A graph $G$ is $q$-unbreakable if it is $(q(i),i)$-unbreakable for every $i \leq k$.
\end{definition}

\subsection{Improved Graph}
For a positive integer $k$, we say that a graph $H$ is {\em{$k$-complemented}} if the implication $(\conn_H(x,y) > k) \Rightarrow (xy\in E(H))$ holds for every pair of vertices $x,y\in V(H)$. For every graph $G$ we can construct the {\em{$k$-improved graph}} $\imp{G}{k}$ by having $V(\imp{G}{k})=V(G)$ and $xy\in E(\imp{G}{k})$ if and only if $\conn_G(x,y)> k$. Observe that $\imp{G}{k}$ is a supergraph of $G$, since $\conn_G(x,y)=+\infty$ for all $x,y$ that are adjacent in $G$. Moreover, observe that every $k$-complemented supergraph of $G$ must be also a supergraph of $\imp{G}{k}$. It turns out that $\imp{G}{k}$ is $k$-complemented itself, and hence it is the unique minimal $k$-complemented supergraph of $G$.

\begin{lemma}[\cite{Bodlaender03,ClautiauxCMN03,LokshtanovPPS17}]
\label{lem:improval}
For every graph $G$ and positive integer $k$, the graph $\imp{G}{k}$ is $k$-complemented. Consequently, it is the unique minimal $k$-complemented supergraph of $G$.
\end{lemma}

Given $G$ and $k$ we can compute $\imp{G}{k}$ using $n^2$ applications of Lemma~\ref{lem:f-f}, leading to the following. 

\begin{lemma}\label{lem:comp-imp}
There exists an algorithm that, given a positive integer $k$ and a graph $G$ on $n$ vertices computes $\imp{G}{k}$ in $\Oh(kn^2(n+m))$ time.
\end{lemma}

A nice property of improved graphs is that the set of separations of order at most $k$ in $G$ are also separations in $\imp{G}{k}$, and vice versa (since an edge from $u \in A \setminus B$ to $v \in B \setminus A$ in $E(\imp{G}{k})$ would imply that $\mu(u, v) > k$, contradicting that the order of $(A, B)$ is less than $k$. This implies the following lemma.

\begin{lemma}\label{lem:unbreakableInImproved}
Let $G$ be a graph, $S$ be a vertex set in $G$ and $q$, $k$ be positive integers. Then $S$ is $(q, k)$-unbreakable in $\imp{G}{k}$  if and only if  $S$ is $(q, k)$-unbreakable in $G$.
\end{lemma}

Sometimes when discussing a graph $G$ we are interested in knowing whether $\imp{G}{k}$ is (clique) unbreakable. 

\begin{definition}
A pair $(L, R)$ of vertex sets of $G$ is an $h$-{\em improved clique separation} of $G$ if $(L, R)$ is a clique separation in the $h$-improved graph $\imp{G}{h}$ of $G$. 
The graph $G$ is $(q, k, h)$-improved-clique-unbreakable if for every clique separation $(L, R)$ of $\imp{G}{h}$ of order at most $k$ it holds that $\min(|L|, |R|) \leq q$.
\end{definition}

\input{boundariedGraphs.tex}

\subsection{Tree Decompositions and Star Decompositions} 
\label{sec:treeDecAndStarDeco}
In this paper it is most convenient to view tree decompositions of graphs as having a root. 
Let $T$ be a rooted tree and let $t$ be any non-root node of $T$. The parent of $t$ in $T$ will be denoted by $\parent(t)$. 

\begin{definition}[\bf Tree decomposition]
A {\em tree decomposition} of a graph $G$ is a pair $(T,\chi)$, where $T$ 
is a rooted tree and $\chi \colon V(T) \to 2^{V(G)}$ is a mapping such that:
\begin{itemize}
  \item for each node $v \in V(G)$, the set $\{t \in V(T)\ |\ v \in \chi(t)\}$ induces a nonempty and connected subtree of~$T$,
  \item for each edge $e \in E(G)$, there exists $t \in V(T)$ such that $e \subseteq \chi(t)$.
\end{itemize}
\end{definition}

The sets $\chi(t)$ for $t\in V(T)$ are called the {\em{bags}} of the decomposition, while sets $\chi(s)\cap \chi(t)$ for $st\in E(T)$ are called the {\em{adhesions}}. We sometimes implicitly identify a node of $T$ with the bag associated with it. The {\em{width}} of a tree decomposition $T$ is equal to its maximum bag size decremented by one, i.e., $\max_{t\in V(T)} |\chi(t)|-1$. The {\em{adhesion width}} of $T$ is equal to its maximum adhesion size, i.e., $\max_{st\in E(T)} |\chi(s)\cap \chi(t)|$. We additionally define the function $\sigma$ as follows:
\begin{align*}
\sigma(t) & = \begin{cases} \emptyset & \text{if }t\text{ is the root of }T \\ \chi(t) \cap \chi(\parent(t)) & \text{otherwise,}\end{cases} \\
\end{align*}
The {\em{treewidth}} of a graph, denoted $\tw(G)$, is equal to the minimum width among all its tree decompositions. When we define tree decomposition  as $(T,\chi_x)$, 
then the notation carries over to the auxiliary function $\sigma$, that is, we will denote this function by $\sigma_x$.

A {\em star decomposition} of a graph $G$ is a tree decomposition $(T, \chi)$ of $G$ where $T$ is a rooted tree with root $r$, and all vertices of $T$ other than $r$ are leaves. We call $r$ the {\em central node} of the decomposition.  A star decomposition $(T, \chi)$ with central node $r$ is called {\em connectivity-sensitive} if $\chi(\ell) \cap \chi(r) \neq \chi(\ell') \cap \chi(r)$ for every pair of distinct leaves $\ell, \ell'$ in $T$, and for every leaf $\ell$ of $T$ and every connected component $C$ of $G[\chi(\ell) \setminus \chi(r)]$ it holds that $N(C) = \chi(\ell)\cap \chi(r)$. The following observation follows immediately from the definition of connectivity-sensitive star decompositions. 

\begin{observation}\label{obs:uniqueStarDec} For every graph $G$ and set $B \subseteq V(G)$ there exists precisely one (up to isomorphism) connectivity-sensitive star decomposition $(T, \chi)$ with central node $r$ and $\chi(r) = B$. Given $G$ and $B$, $(T, \chi)$ can be computed in polynomial time.
\end{observation}

\newcommand{\obj}{\mathcal{D}}
\newcommand{\strobj}{\mathcal{E}}

\subsection{Isomorphisms of Colored and Labeled Graphs, Lexicographic Ordering}\label{sec:isoPreliminaries}
We say that two graphs $G_1,G_2$ are {\em{isomorphic}} if there exists a bijection $\phi\colon V(G_1)\to V(G_2)$, called {\em{isomorphism}}, such that $xy\in E(G_1)\Leftrightarrow \phi(x)\phi(y)\in E(G_2)$ for all $x,y\in V(G_1)$. An \emph{automorphism} of a graph is an isomorphism from the graph to itself. 

The notions of isomorphism and automorphism naturally extend to the setting of {\em{relational structures}}, which (in this work) consist of a finite universe plus a number of relations on this universe. The set of relations used in a structure is the {\em{signature}}. For instance, graphs can be modelled by taking the vertex set to be the universe and having one binary relation signifying adjacency. Various combinatorial objects, for instance tree decompositions, can be naturally modelled as relational structures. We refrain from introducing a formal notation for relational structures, as in all cases of usage in this paper all the terminology works naturally in the combinatorial setting. 

The {\em{union}} of two relational structures $A$ and $B$, denoted $A\cup B$, consists of the union of universes of $A$ and $B$ and the disjoint union of relations appearing in $A$ and in $B$. Note that in case the universes of $A$ and $B$ share some elements, these are represented only once in $A\cup B$.

In this work we will often consider graphs where some or all vertices are labelled with natural numbers $\mathbb{N}$ (we have already seen such objects in the context of boundaried graphs). Such graphs are modelled as relational structures by adding to the graph signature a new unary relation $R_t$ for every $t\in \mathbb{N}$, which selects all vertices that have label $t$. Thus, isomorphisms and automorphisms on such labelled graphs are expected to preserve the labels, that is, a vertex $u$ is labelled $t$ if and only if the image of $u$ is labelled $t$.

Our algorithms will some times output a special symbol $\bot$, meaning that the algorithm ``{\em failed}''. This symbol can be modeled by a single element $\bot$ and a unary relation $R_\bot$ which selects $\bot$ and only $\bot$. This has the effect that an isomorphism can only take $\bot$ (in the output of the algorithm for one input) to $\bot$ (in the output of the algorithm on another input).

A {\em coloring} of a graph $G$ is a function $f \colon V(G) \rightarrow \mathbb{N}$; thus, all the vertices are assigned a color. A colored graph is a pair $(G, {\sf col}_G)$ where ${\sf col}_G$ is a coloring of the graph $G$. We will frequently refer to colored graphs $(G, {\sf col}_G)$ by just the graph $G$. In this case the coloring of the colored graph $G$ is ${\sf col}_G$. The above definition of isomorphism implies that two colored graphs $G$ and $G'$ are considered isomorphic if there exists an isomorphism $\phi \colon V(G) \rightarrow V(G')$ between $G$ and $G'$ that preserves colors. In other words, ${\sf col}_{G'}(\phi(v)) = {\sf col}_G(v)$ for every $v \in V(G)$.

A {\em labeling} of a graph $G$ is an injective function %
$f : X \rightarrow \mathbb{N}$ for a subset $X \subseteq V(G)$;
in other words, it is an injective coloring of a vertex subset.
A labeling is \emph{full} if its domain is $V(G)$ (i.e., every vertex is labeled) 
and \emph{proper} if additionally its range is $[|V(G)|]$. 
A {\em properly labeled} graph is a pair $(G, {\sf lab}_G)$ where ${\sf lab}_G$ is a proper labeling of $G$. Just as for colorings we will often refer to properly labeled graphs $(G, {\sf lab}_G)$ by just the graph $G$, in which case the labeling of the properly labeled graph $G$ is ${\sf lab}_G$. Our definition of isomorphism implies that two properly labeled graphs $(G_1, f_1)$ and $(G_2, f_2)$ are isomorphic if and only if $\phi \colon V(G_1) \rightarrow V(G_2)$ defined as $\phi(v) = f_2^{-1}(f_1(v))$ is an isomorphism of the (unlabeled) graphs $G_1$ and $G_2$. Since checking isomorphism of properly labeled graphs is trivial, we will often say that two isomorphic properly labeled graphs are {\em{equal}}, even though formally their vertex sets might be different.

We now define a lexicographical order on properly labeled graphs. To that end we will say that a pair $(i, j)$ of natural numbers is {\em lexicographically smaller than or equal to} the pair $(i', j')$ if either $i < i'$, or $i = i'$ and $j \leq j'$.
Given two properly labeled graphs $G_1$, $G_2$ we will say that $G_1$ is {\em lexicographically smaller than or equal to} $G_2$, denoted by $G_1 \leq_{\sf{lex}} G_2$ if $|V(G_1)| < |V(G_2)|$, or $|V(G_1)| = |V(G_2)|$ and for the lexicographically smallest pair $(i, j)$ so that
$${\sf lab}_{G_1}^{-1}(i) {\sf lab}_{G_1}^{-1}(j) \in E(G_1)  \centernot\iff  {\sf lab}_{G_2}^{-1}(i) {\sf lab}_{G_2}^{-1}(j) \in E(G_2)$$
it holds that ${\sf lab}_{G_1}^{-1}(i) {\sf lab}_{G_1}^{-1}(j) \notin E(G_1)$, or no such pair $(i, j)$ exists.
It should be clear that $\leq_{\sf{lex}}$ is a total order on properly labeled graphs, and that if $G_1 \leq_{\sf{lex}} G_2$ and $G_2 \leq_{\sf{lex}}  G_1$ then $G_1$ and $G_2$ are isomorphic (equal) labeled graphs. We say that $G_1$ is {\em lexicographically less than} $G_2$, denoted by $G_1 <_{\sf{lex}} G_2$ if $G_1 \leq_{\sf{lex}} G_2$ and $G_1 \neq G_2$.
We extend the ordering $\leq_{\sf{lex}}$ to colored properly labeled graphs by breaking ties using the coloring. More concretely, for two colored, properly labeled graphs $G_1$ and $G_2$ we have that $G_1 \leq_{\sf{lex}} G_2$ if the uncolored copy of $G_1$ is lexicographically less than the uncolored copy of $G_2$, or the uncolored copies of $G_1$ and $G_2$ are equal, and the smallest $i$ such that 
$${\sf col}_{G_1}({\sf lab}_{G_1}^{-1}(i)) \neq {\sf col}_{G_2}({\sf lab}_{G_2}^{-1}(i))$$
satisfies ${\sf col}_{G_1}({\sf lab}_{G_1}^{-1}(i)) < {\sf col}_{G_2}({\sf lab}_{G_2}^{-1}(i))$ (or such an $i$ does not exist). Again, $\leq_{\sf{lex}}$ is a total order on properly labeled colored graphs and we define $<_{\sf{lex}}$ for colored properly labeled graphs in the same way we did for uncolored ones. 

We extend the lexicographic ordering to {\em unlabeled}, (possibly colored) graphs. Here $G_1 \leq_{\sf{lex}} G_2$ means that there exists a proper labeling ${\sf lab}_{G_1}$ of $G_1$ such that for every proper labeling ${\sf lab}_{G_2}$ of $G_2$ it holds that $(G_1, {\sf lab}_{G_1}) \leq_{\sf lex}  (G_2, {\sf lab}_{G_2})$.

\paragraph{Convention} Henceforth, all graphs in this 
work are possibly colored, unless explicitly noted otherwise. 
This in particular applies to all isomorphism-related claims. 

\subsection{(Weak) Isomorphism Invariance, Canonical Labelings}
We will often say that an function $f$ is {\em{isomorphism-invariant}}. By this we mean that whenever two structures $A$ and $B$ in the domain of $f$ are isomorphic through an isomorphism $\phi$, then $f(A)$ and $f(B)$ are isomorphic as well through an isomorphism $\psi$ that coincides with $\phi$ on the intersections of universes.
An algorithm is isomorphism-invariant if the function that maps inputs of the algorithm to the output produced by the algorithm is isomorphism-invariant.

As an example, consider an algorithm ${\cal A}$ that given a graph $G$ constructs a tree decomposition $(T, \chi)$ of $G$. For ${\cal A}$ to be isomorphism-invariant we require that if $G$ and $G'$ are isomorphic through an isomorphism $\phi$, $(T, \chi)$ is the output of ${\cal A}(G)$, and $(T', \chi')$ is the output of ${\cal A}(G')$, then there exists an isomorphism $\psi$ from $T$ to $T'$ so that for every $u \in V(G)$ and $t \in V(T)$ it holds that $u \in \chi(t)$ if and only if $\phi(u) \in \chi'(\psi(t))$. Most often, isomorphism-invariance of the functions or algorithms will be obvious, since the description does not depend on the internal representation of the graph, nor uses any tie-breaking rules for choosing arbitrary objects. 

We will often need a weaker condition of invariance, defined as follows.
A function $f$ is {\em weakly isomorphism invariant} if for every pair $A$, $B$ of isomorphic objects in the domain of $f$, there exist isomorphisms $\phi$ from $A$ to $B$ and $\psi$ from $f(A)$ to $f(B)$ that coincide on the intersection of universes of $A$ and $f(A)$.
Thus, we do not expect {\em{every}} isomorphism from $A$ to $B$ to extend to an isomorphism from $f(A)$ to $f(B)$, but at least one should extend.
Again, an algorithm is weakly isomorphism-invariant if the function that maps inputs of the algorithm to the output produced by the algorithm is weakly isomorphism-invariant. 

An important example of a function that is weakly isomorphism invariant, but not isomorphism invariant,  is the following: the function that takes as input an unlabeled graph $G$ and returns $G$ together with any proper labeling ${\sf lab}_G$ such that the properly labeled graph $(G, {\sf lab}_G)$ is lexicographically smallest.

A slightly more general version of the following observation will be used throughout our arguments, especially in Sections~\ref{sec:canonBoundaried},
~\ref{sec:unpumping},
~\ref{sec:redToCliqueUnbreakable}, and~\ref{sec:reductionToUnbreakable}.
\begin{lemma}\label{lem:selectWeaklyIsoInvariant}
The function that takes an unordered multiset of properly labeled, colored graphs and outputs the multiset sorted in lexicographical order is weakly isomorphism invariant. %
\end{lemma}

We will often abuse the terminology and say that an object (such as a graph, a set, or a tree-decomposition) is (weakly) isomorphism-invariant in some other object(s). This will mean that the object is the output of an isomporphism-invariant function / algorithm that takes the other objects on input. When we say that an algorithm or function is (weakly) isomorphism invariant with respect to an object this emphasizes (a) what the input domain of the algorithm is, and (b) that the algorithm is (weakly) isomorphim invariant. The following observation will be often used implicitly; the proof is obvious and omitted.

\begin{lemma}
Let %
function $f$ and $g$ be functions with the same domain, where $f$ is weakly isomorphism invariant, $g$ is isomorphism invariant, and the intersection of the universes of $f(A)$ and $g(A)$ is contained in the universe of $A$ for every $A$ in the domain of $f$. Further, let $h$ be a weakly isomorphism invariant function that has $A\cup f(A)$ in the domain for each $A$ in the domain of $f$. Then:
\begin{itemize}
    \item The function $A\mapsto f(A)\cup g(A)$ is weakly isomorphism invariant.
    \item The function $A\mapsto h(A\cup f(A))$ is weakly isomorphism invariant.
\end{itemize}
\end{lemma}

Finally, we  define a {\em canonical labeling} of a (possibly colored) unlabeled graph $G$ to be a weakly isomorphism invariant proper labeling ${\sf lab}_G$ of $G$. %

\paragraph{Compacting Colorings and Labelings. Weak Boundary Consistency.}
We will frequently deal with colored and/or labeled graphs, and we allow the colors and labels of the vertices to be any natural numbers. However, we will need to ensure that the bit length of the colors used in intermediate computations does not exceed polynomial in $n$. We now define some tools for that lets us manipulate the coloring of the input graph without changing its automorphism group. 

Let $G$ be a colored graph with coloring $\textsf{col}_G$. The {\em color partition} of $G$ is the partition ${\cal P}$ of $V(G)$ such that for every pair of vertices $u, v$ in $G$ we have  $\textsf{col}_G(u) = \textsf{col}_G(v)$ if and only if some part  $P \in {\cal P}$ contains $\{u,v\}$. The {\em boundary color partition} of a boundaried colored graph $(G, \iota)$ with coloring $\textsf{col}_G$ is the partition ${\cal P}$ of $\bnd(G, \iota)$ such that for every pair of vertices $u, v$ in $\bnd(G, \iota)$ we have  $\textsf{col}_G(u) = \textsf{col}_G(v)$ if and only if some part  $P \in {\cal P}$ contains $\{u,v\}$. The following observation, whose proof follows from the definition of color partitions, states that from the perspective of isomorphism only the color partition of a coloring matters, not the actual coloring. 

\begin{observation}\label{obs:samePartition}
Let $G$ be a graph and $\mathsf{col}_G^1$ and  $\mathsf{col}_G^2$ be colorings of $G$ with the same color partition. Then the automorphism group of $G$ colored with $\mathsf{col}_G^1$  and of  $G$ colored with $\mathsf{col}_G^2$ are the same. 
\end{observation}

We will say that two colored boundaried graphs $((G_1, \iota_1), \textsf{col}_{G_1})$ and $((G_2, \iota_2), \textsf{col}_{G_2})$ are {\em weakly boundary-consistent} if their uncolored versions $(G_1, \iota_1)$ and $(G_2, \iota_2)$ are boundary-consistent and the boundary partition ${\cal P}_1$ of $((G_1, \iota_1), \textsf{col}_{G_1})$ and ${\cal P}_2$ of $((G_2, \iota_2), \textsf{col}_{G_2})$ coincide on their common vertices. In other words, for every $i, j$ in $\iota_1(\bnd(G_1, \iota_1)) \cap \iota_2(\bnd(G_2, \iota_2))$ some part of ${\cal P}_1$ contains $\{\iota_1(i), \iota_1(j)\}$ if and only if some part of ${\cal P}_2$ contains $\{\iota_2(i), \iota_2(j)\}$. 

In light of Observation~\ref{obs:samePartition} we define the {\em color compacting} operation. We will say that a coloring (labeling) of a graph $G$ is {\em compact} if its range is a subset of $[|V(G)|]$. The {\em compacting} operation takes as input a colored graph $(G, \textsf{col}_G)$ and outputs the (compact) coloring $\textsf{col}_G' : V(G) \rightarrow [|V(G)|]$ defined as follows: $\textsf{col}_G'(v) = 1 + |\{u \in V(G) ~:~ \textsf{col}_G(u) <  \textsf{col}_G(v)\}$. The following observation follows immedately from the definition of compacting together with Observation~\ref{obs:samePartition}.

\begin{lemma}\label{lem:compact} The compacting operation is isomorphism invariant. Further, if $\mathsf{col}_G'$ is the result of compacting $(G, \mathsf{col}_G)$ then the automorphism groups of $(G, \mathsf{col}_G)$ and of $(G, \mathsf{col}_G')$ are the same. 
\end{lemma}

Lemma~\ref{lem:compact} implies that a canonical labeling of $(G, \textsf{col}_G)$ is also a canonical labeling of $(G, \textsf{col}_G')$ and vice versa. Note that every labeling of a graph is also a coloring. Therefore we can apply compacting also to a labeled graph to obtain a labeling with range $[|V(G)|]$. By Lemma~\ref{lem:compact} the result of compacting a canonical labeling of $G$ is also a canonical labeling of $G$.

\paragraph{Encoding Boundaried Colored Graphs as Colored Graphs}
Let $(G, \iota)$ be a compactly colored boundaried graph with coloring $\textsf{col}_G$. The {\em color-encoding} of $(G, \iota)$ is the colored (but un-boundaried) graph $G$ with coloring $\textsf{col}_G'$ defined as follows.
For every vertex $v \in V(G) \setminus \bnd(G, \iota)$ we set  $\textsf{col}_G' = \textsf{col}_G$. 
For every vertex $v \in \bnd(G, \iota)$  we set $\textsf{col}_G' = n + \textsf{col}_G$. 

It is easy to see that color-encoding has the same nice property as compacting: it is isomorphism invariant and also the automorphism groups of $((G, \iota), \textsf{col}_G)$ and of $(G, \textsf{col}_G')$ are the same. The {\em compact color-encoding} of $(G, \iota)$ is the result of compacting the coloring  $\textsf{col}_G'$. By Lemma~\ref{lem:compact} the property of color-encodings also holds for compact color encodings. We state this below as a lemma. 

\begin{lemma}\label{lem:compactColorEncoding} The compact color encoding operation is isomorphism invariant. Further if $\mathsf{col}_G'$ is the result of compactly color-encoding $(G, \iota)$ with coloring $\mathsf{col}_G$ then the automorphism groups of $((G, \iota), \mathsf{col}_G)$ and of $(G, \mathsf{col}_G')$ are the same. 
\end{lemma}

Throughout the paper, we assume that ${\cal F}$ denotes a finite family of graphs, and $n_{\cal F}$ denotes the maximum size of a graph (in terms of the number of vertices) in ${\cal F}$.

%% file: boundariedGraphs.tex
\subsection{Boundaried Graphs, Gluing and Forgetting, Boundary-Consistency}
For a finite set $\Lambda$ of labels, a \emph{$\Lambda$-boundaried graph} is a pair $(G,\iota)$ where $G$ is a graph and $\iota \colon B \to \Lambda$ is an injective function for some $B \subseteq V(G)$; then $B= \iota^{-1}(\Lambda)$ is called the {\em{boundary}} of $(G,\iota)$ and denoted by $\bnd(G, \iota)$. We will abuse notation and refer to $\bnd(G, \iota)$ by simply $\bnd G$ when the function $\iota$ is clear from context.
Note that by definition, if $\Lambda \subseteq \Lambda'$ then a $\Lambda$-boundaried graph is also a $\Lambda'$-boundaried graph. A \emph{$k$-boundaried graph} is a shorthand for a $[k]$-boundaried graph, where we denote $[k]=\{1,\ldots,k\}$. For $v \in B$, $\iota(v)$ is the \emph{label} of $v$. A boundaried graph is a $\Lambda$-boundaried graph for some set $\Lambda$. We define several operations on boundaried graphs.

If $(G_1,\iota_1)$ and $(G_2,\iota_2)$ are two $\Lambda$-boundaried graphs, then the result of \emph{gluing} them is the boundaried graph $(G_1,\iota_1) \oplus_\Lambda (G_2,\iota_2)$ that is obtained from the disjoint union of $(G_1, \iota_1) $ and $(G_2, \iota_2)$ by identifying vertices of the same label (so that the resulting labeling is again injective). When vertices are identified, for every $i, j \in \Lambda$ the resulting graph $(G_1,\iota_1) \oplus_\Lambda (G_2,\iota_2)$ has an edge between the vertex labeled $i$ and the vertex labeled $j$ if and only if at least one of $(G_1,\iota_1)$  and $(G_2,\iota_2)$ has an edge between the vertex labeled $i$ and the vertex labeled $j$. When the labeling $\Lambda$ is clear from context we will sometimes drop the subscript and refer to $\oplus_\Lambda$ by just $\oplus$.

The gluing operation is extended to the case when $(G_1,\iota_1)$ is a $\Lambda_1$-boundaried graph and $(G_2,\iota_2)$ is a $\Lambda_2$-boundaried graph by considering both  $(G_1,\iota_1)$ and $(G_2,\iota_2)$ as $\Lambda_1 \cup \Lambda_2$-boundaried graphs. If $(G_1,\iota_1)$ is a $\Lambda_1$-boundaried graph and $(G_2,\iota_2)$ is a $\Lambda_2$-boundaried graph and $G_1[\iota_1^{-1}(\iota_1(\bnd G_1) \cap \iota_2(\bnd G_2))] = G_2[\iota_2^{-1}(\iota_1(\bnd G_1) \cap \iota_2(\bnd G_2))]$ then we say that $(G_1,\iota_1)$ and $(G_2,\iota_2)$ are {\em boundary-consistent}.

A boundaried graph $(G, \iota)$ can also be colored - the graph $G$ may come with a coloring $\textsf{col}_G$.
We can extend {\em gluing} to colored boundaried graphs. In particular the result of $(G_1,\iota_1) \oplus (G_2, \iota_2)$ is a colored graph where every vertex $v \in (V(G_1) \setminus V(G_2))$ inherits its color from $G_1$, every vertex $v \in (V(G_2) \setminus V(G_1))$ inherits its color from $G_2$, while every vertex in $V(G_1) \cap V(G_2)$ inherits the smallest color out of its color in $G_1$ and its color in $G_2$. 
We say that  a $\Lambda_1$-boundaried colored graph $(G_1,\iota_1)$ and a $\Lambda_2$-boundaried colored graph $(G_2,\iota_2)$ are {\em boundary-consistent} if  the corresponding boundaried (uncolored) graphs are boundary-consistent and for every label $i \in \Lambda_1 \cap \Lambda_2$ so that $\iota_1^{-1}(i)$ and $\iota_2^{-1}(i)$ are defined we have that the color of $\iota_1^{-1}(i)$ in $G_1$ and of $\iota_2^{-1}(i)$ in $G_2$ are the same.

We now define the $\textsf{forget}$ operation. The $\textsf{forget}$ operation takes a $\Lambda$-boundaried graph $(G, \iota)$ and a set $F$ of labels. It returns the boundaried graph $\textsf{forget}(G, \iota, F) = (G, \iota')$ with boundary $\bnd(G, \iota') = \bnd(G, \iota) \setminus \iota^{-1}(F)$ and $\iota'(v) = \iota(v)$ for every $v \in \bnd(G, \iota')$.

\subsection{Topological Minors Of Boundaried Graphs}
We now extend the notion of topological minors to boundaried graphs. Let $(H, \iota_1)$ be a $\Lambda_H$-boundaried graph and $(G, \iota)$ be a $\Lambda$-boundaried graph such that $\Lambda_H \subseteq \Lambda$. A {\em topological minor model} of $(H, \iota_1)$ in $(G, \iota)$ is a topological minor model $(f,P)$ of $H$ in $G$ such that for every $u \in \bnd H $ we have $f(u) = \iota^{-1}\iota_1(u)$ and for every $uv \in E(H)$ the path $P(uv)$ has no internal vertices in $\bnd G$. If there exists a topological minor model of $(H, \iota_1)$ in $(G, \iota)$ we say that $(H, \iota_1)$ is a {\em topological minor} of $(G, \iota)$, and we denote it by $(H ,\iota_1) \sqsubseteq (G, \iota)$.

The {\em complexity} of a boundaried graph  $(H, \iota_1)$, denoted $\textsf{complexity}(H, \iota_1)$, is defined as the number of vertices in $H$ that are not in $\bnd H$.
For a $\Lambda$-boundaried graph $(G, \iota)$ the {\em folio} of $(G, \iota)$ is the set 
${\bf folio}(G, \iota) = \{(H, \iota')~\mid~(H, \iota') \sqsubseteq (G, \iota)\}$.
For a non-negative integer $q$, the $q$-{\em folio} of $(G, \iota)$ is the set 
$$q\textbf{-folio}(G, \iota) = \{(H, \iota')~:~ \textsf{complexity}(H, \iota') \leq q \mbox{ and } (H, \iota') \sqsubseteq (G, \iota)\}\mbox{.}r$$ 
The notion of folio and $q$-folio extends naurally to normal (unboundaried) graphs as well by treating them as boundaried graphs with an empty boundary. Observe specifically that the $q$-folio of an unboundaried graph $G$ then contains precisely the set of topological minors of $G$ on at most $q$ vertices. 
The following lemma states that the $q$-folio of  $\textsf{forget}(G, \iota, \{i\})$ depends only on the $q$-folio of $(G, \iota)$.

\begin{lemma}\label{lem:qfolioForget}
Let $(G_1, \iota_1)$ and $(G_2, \iota_2)$ be two $\Lambda$-boundaried graphs such that 
$q\textbf{-folio}(G_1, \iota_1) = q\textbf{-folio}(G_2, \iota_2)$. Then, for every $F \subseteq \Lambda$ it holds that 
$$q\textbf{-folio}(\textsf{forget}(G_1, \iota_1, F)) =q\textbf{-folio}(\textsf{forget}(G_2, \iota_2, F)).$$
\end{lemma}

\begin{proof}
We prove the statement only for $F = \{i\}$ for some $i \in \Lambda$. The general statement then follows by repeated application of the statement with $F' = \{i\}$ for each $i$ in $F$. We show that if $(H, \iota)$ is in $q\textbf{-folio}(\textsf{forget}(G_1, \iota_1, \{i\}))$ then it is also in $q\textbf{-folio}(\textsf{forget}(G_2, \iota_2, \{i\}))$. Let $(f, P)$ be a model of $(H, \iota)$ in $q\textbf{-folio}(\textsf{forget}(G_1, \iota_1, \{i\}))$. 
We may assume that there exists a boundary vertex $v$ of $G_1$ labeled $i$, since otherwise such a vertex does not exist for $G_2$ either, and the ${\sf forget}$ operation leaves both $G_1$ and $G_2$ unchanged. Thus, let $v = \iota_1^{-1}(i)$ be the boundary vertex in $G_1$ labeled $i$. We consider three cases. 

First, suppose $v = f(u)$ for a vertex $u \in V(H)$. In that case, define $\iota'(x) = \iota(x)$ for $x \in \bnd(H, \iota)$ and set $\iota'(u) = i$. Then $(f, P)$ is a model of the boundaried graph $(H, \iota')$ in $(G_1, \iota)$ and $\textsf{complexity}(H, \iota') = \textsf{complexity}(H, \iota) - 1 \leq q - 1$ since $u$ counts towards the complexity of $(H, \iota)$ but not towards the complexity of $(H, \iota')$. But $(G_1, \iota_1)$ and $(G_2, \iota_2)$ have the same $q$-folio so $(H, \iota')$ is a topological minor of $(G_2, \iota_2)$. Thus $\textsf{forget}(H, \iota', \{i\}) = (H, \iota)$ is in the folio (and therefore the $q$-folio) of $\textsf{forget}(G_2, \iota_2, \{i\})$.

Next suppose $v \notin f(V(H))$, but that $v$ is an internal vertex of a path $P(xy)$ for $xy \in E(H)$. Define $H'$ to be the graph obtained from $H$ by removing the edge $xy$ and adding the vertex $z$ with edges $zy$ and $xz$. For every vertex $r \in \bnd(H, \iota)$ set $\iota'(r) = \iota(r)$ and define $\iota'(z) = i$.
We have that $(H', \iota')$ is a topological minor of $(G_1, \iota_1)$ and that the complexity of $(H', \iota')$ is at most $|V(H') \setminus \bnd(H', \iota')| = |V(H) \setminus \bnd(H, \iota)| \leq q$. But $(G_1, \iota_1)$ and $(G_2, \iota_2)$ have the same $q$-folio so $(H', \iota')$ is a minor of $(G_2, \iota_2)$. Thus $\textsf{forget}(H', \iota', \{i\})$ is in the folio of $\textsf{forget}(G_2, \iota_2, \{i\})$. Finally,  $(H, \iota)$ is a topological minor of $\textsf{forget}(H', \iota', \{i\})$, and thus   $(H, \iota)$ is in the folio of  $(G_2, \iota_2)$. Since the complexity of $(H, \iota)$ is at most $q$ it is also in the  $q$-folio of  $(G_2, \iota_2)$. 

Finally, consider the case that $v \notin f(V(H))$ and that $v$ is not an internal vertex of a path $P(xy)$ for any $xy \in E(H)$. In this case $(f, P)$ is a model of $(H, \iota)$ also in $(G_1, \iota_1)$. Since $(G_1, \iota_1)$ and $(G_2, \iota_2)$ have the same $q$-folio we have that  $(H, \iota)$ is a topological minor of $(G_2, \iota_2)$ and also in $q\textbf{-folio}(\textsf{forget}(G_2, \iota_2, \{i\}))$, because $i$ is not a label of any boundary vertex of $H$. The proof of the reverse implication is symmetric.
\end{proof}

The following lemma states that the $q$-folios of two boundaried graphs $(G_1, \iota_1)$ and $(G_2, \iota_2)$ completely determine the $q$-folio of  $(G_1, \iota_1) \oplus (G_2, \iota_2)$. 

\begin{lemma}\label{lem:topFolioGlue} 
For every integer $q$ and every triple $(G_1, \iota_1)$, $(G_1', \iota_1')$ and $(G_2, \iota_2)$ of $\Lambda$-boundaried graphs such that $q\textbf{-folio}(G_1, \iota_1) = q\textbf{-folio}(G_1', \iota_1')$ it holds that 
$$q\textbf{-folio}((G_1, \iota_1) \oplus (G_2, \iota_2)) = q\textbf{-folio}((G_1', \iota_1') \oplus (G_2, \iota_2))$$
\end{lemma}

\begin{proof}
It suffices to show that if $(H, \iota)$ is in $q\textbf{-folio}((G_1, \iota_1) \oplus (G_2, \iota_2))$ then it is also a topological minor of $q\textbf{-folio}((G_1', \iota_1') \oplus (G_2, \iota_2))$, since the proof of the reverse direction follows by interchanging $(G_1, \iota_1)$ with $(G_1', \iota_1')$. Consider a model $(f, P)$ of $(H, \iota)$ in $(G_1, \iota_1) \oplus (G_2, \iota_2)$.

We define the boundaried graph $H_1$ as follows. The vertex set of $H_1$ are all the vertices $u$ of $H$ so that $f(u) \in V(G_1)$. Here we consider the vertex set of $(G_1, \iota_1) \oplus (G_2, \iota_2)$ as $V(G_1) \cup V(G_2)$ where the vertices in $\iota_1^{-1}(\Lambda) \cap \iota_2^{-1}(\Lambda)$ are considered as vertices both in $V(G_1)$ and in $V(G_2)$. The edge set of $H_1$ is defined as $\{uv \in E(H) ~:~ E(P(uv)) \subseteq E(G_1)\}$. Similarly $H_2$ has vertex set $\{u \in V(H) ~\mid~ f(u) \in V(G_2)\}$ and edge set $\{uv \in E(G) ~\mid~ E(P(uv)) \subseteq E(G_2)\}$. Note that no path $P(uv)$ has any internal vertices in $\bnd((G_1, \iota_1) \oplus (G_2, \iota_2))$ and therefore every edge $uv$ in $E(H)$ is an edge of $H_1$ or an edge of $H_2$ (it can belong to both $H_1$ and $H_2$ in case $P(uv)$ consists of one edge shared between $G_1$ and $G_2$).

Define $\iota_1^H$ to have domain $\{u \in V(H_1) ~\mid~ f(u) \in \bnd(G_1, \iota_1)\}$ and set $\iota_1^H = \iota_1(f(u))$. Similarly Define $\iota_2^H$ to have domain $\{u \in V(H_2) ~\mid~ f(u) \in \bnd(G_2, \iota_2)\}$ and set $\iota_2^H = \iota_2(f(u))$.  We have that $(H, \iota) = (H, \iota_1^H) \oplus  (H, \iota_2^H)$,  $(H_1, \iota_1^H)$ is in the $q$-folio of $(G_1, \iota_1)$, and $(H_2, \iota_2^H)$ is in the $q$-folio of $(G_2, \iota_2')$. At the same time the complexity of $(H_1, \iota_1^H)$ is at most the complexity of  $(H, \iota)$ since only non-boundary vertices of $(H, \iota)$ in $(G_1, \iota_1)$ are counted, so  $(H_1, \iota_1^H)$ is in the $q$-folio of $(G_1, \iota_1)$. But then $(H_1, \iota_1^H)$ is in the $q$-folio of $(G_1', \iota_1')$ as well and therefore $(H, \iota) = (H, \iota_1^H) \oplus  (H, \iota_2^H)$ is in the  $q$-folio of $(G_1', \iota_1') \oplus (G_2, \iota_2)$. This completes the proof. 
\end{proof}

We summarize the way that we will use Lemmas~\ref{lem:qfolioForget} and \ref{lem:topFolioGlue} in the next lemma. Specifically, if $(A, B)$ is a separation in a graph $G$ then we may define a labeling $\iota \colon A \cap B \rightarrow |[A \cap B]|$ and view $G$ as $\textsf{forget}((G[A], \iota) \oplus (G[B], \iota), [|A \cap B|])$. The next lemma states that the $q$-folio of $G$ depends only on the $q$-folio of $(G[A], \iota)$ and of $(G[B], \iota)$.

\begin{lemma}\label{lem:topFolioGlueAndForget} 
Let $(G_1, \iota_1)$,  $(G_1', \iota_1')$ and $(G_2, \iota_2)$ be $\Lambda$-boundaried graphs such that $q\textbf{-folio}(G_1, \iota_1) = q\textbf{-folio}(G_1', \iota_1')$. Then, for every subset $F \subseteq \Lambda$
$$q\textbf{-folio}(\textsf{forget}((G_1, \iota_1) \oplus (G_2, \iota_2), F))
= q\textbf{-folio}(\textsf{forget}((G_1', \iota_1') \oplus (G_2, \iota_2), F))$$
\end{lemma}

\begin{proof}
From Lemma~\ref{lem:topFolioGlue} we obtain that
$q\textbf{-folio}((G_1, \iota_1) \oplus (G_2, \iota_2)) = q\textbf{-folio}((G_1', \iota_1') \oplus (G_2, \iota_2))$
Then, Lemma~\ref{lem:qfolioForget} immediately yields the statement of the Lemma. 
\end{proof}

%% file: statements.tex
In this section we provide full versions of Theorems~\ref{thm:intro:toUnbreakable} and~\ref{thm:intro:wrapper}
and show how they combine into a proof
of Theorem~\ref{thm:main}. 
We also discuss which subsequent sections treat which parts of proofs of these theorems. 

Let $\mathbb{F}^\star$ be the collection of all finite sets of finite graphs. %
The interpretation here is as follows. We will work on some class of graphs that is defined by a finite forbidden set ${\cal F}$ of topological minors.
However, we will want to focus our attention only on some graph classes that can be defined this way (e.g. minor closed classes of graphs). This is encapsulated in the following way: we will assume that the set ${\cal F}$ of forbidden topological minors is an element of a collection $\mathbb{F} \subseteq \mathbb{F}^\star$. Here different choices for $\mathbb{F}$ allows us to restrict our attention to different classes of classes of graphs closed under topological minors. 
For an example, to restrict attention to minor closed classes we may set $\mathbb{F} \subseteq \mathbb{F}^\star$ to be the collection of all finite lists ${\cal F}$ of graphs so that forbidding all graphs in ${\cal F}$ as topological minors defines a minor closed graph class.

\begin{definition}
Let $\mathbb{F} \subseteq \mathbb{F}^\star$ be a (possibly infinite) collection of finite sets of graphs.
An algorithm ${\cal A}$ is a {\em canonization algorithm for $\mathbb{F}$-free classes} if ${\cal A}$ takes as input a finite list ${\cal F} \in \mathbb{F}$ and a colored graph $G$ that excludes every member of ${\cal F}$ as a topological minor, and outputs a canonical (i.e. weakly isomorphism invariant) labeling $\lambda$ of $G$. 
\end{definition}

Since we aim to reduce the canonization problem to unbreakable graphs, we need a precise definition of what a canonization algorithm for (sufficiently) unbreakable graphs is. 

\begin{definition}\label{def:canonUnbreak}
Let $\mathbb{F} \subseteq \mathbb{F}^\star$ be a (possibly infinite) collection of finite sets of graphs and $q \colon \mathbb{F}^\star \times \mathbb{N} \rightarrow \mathbb{N}$ and  $\kappa \colon \mathbb{F} \rightarrow \mathbb{N}$  be functions.
An algorithm ${\cal A}$ is {\em canonization algorithm for $(q, \kappa)$-unbreakable $\mathbb{F}$-free classes} if it has the following properties:
\begin{itemize}
\item ${\cal A}$ takes as input a finite list ${\cal F} \in \mathbb{F}$, an integer $k$,
a function $q' : [k] \rightarrow \mathbb{N}$ so that $q'(i) \leq q({\cal F}, i)$ for every $i \leq k$, and an ${\cal F}$-topological minor free, $q'$-unbreakable colored graph $G$, and outputs either $\bot$ or a canonical labeling $\lambda$ of $G$. 

\item ${\cal A}$ is weakly isomorphism invariant. %

\item If $k \geq \kappa({\cal F})$ then the algorithm outputs a labeling $\lambda$ of $G$.
\end{itemize}
\end{definition}

We can now state
the full version of Theorem~\ref{thm:intro:toUnbreakable}, that essentially reduces the general case to the $q$-unbreakable case.

\begin{theorem}\label{thm:mainReduction}
There exists a function
$q \colon \mathbb{F}^\star \times \mathbb{N} \rightarrow \mathbb{N}$, such that, for every collection $\mathbb{F} \subseteq \mathbb{F}^\star$ of finite sets of graphs,
if there exists a canonization algorithm ${\cal A}$ for $(q, \kappa)$-unbreakable $\mathbb{F}$-free classes (for some  $\kappa \colon \mathbb{F} \rightarrow \mathbb{N}$) 
then there exists a canonization algorithm for $\mathbb{F}$-free classes.
The running time of this algorithm is upper bounded by $g({\cal F})n^{\cO(1)}$ plus the total time taken by at most $g({\cal F})n^{\cO(1)}$ invocations of ${\cal A}$ on $(G, k)$, where $G$ is an ${\cal F}$-free graph on at most $n$ vertices, and $k$ depends only on ${\cal F}$, $\kappa$ and $q$.
\end{theorem}

We remark here that we are not able to guarantee the computability of the function $q$ of Theorem~\ref{thm:mainReduction}. 
The full version of Theorem~\ref{thm:intro:wrapper}
(deduced from the subordinate paper~\cite{subordinate} in Section~\ref{sec:unbrk-minor}) is as follows.

\begin{theorem}\label{thm:unbrk-minor}
There exists a constant $\ctime$, a computable function $\funrtime$, and 
a weakly isomorphism invariant 
algorithm $\mathcal{A}$ such that the following holds.
\begin{enumerate}
\item Given on input a graph $H$,
an $H$-minor free graph $G$, an integer $k$, and a function $q \colon [k] \to \N$
with a promise that $G$ is $(q(i),i)$-unbreakable for every $1 \leq i \leq k$, the algorithm ${\cal A}$
either returns $\bot$ or computes a 
labeling of $G$. 
The algorithm runs in time bounded by $\funrtime(H, \sum_{i=1}^k q(i)) \cdot |V(G)|^{\ctime}$.
\item For every graph $H$ and for every function $\funq \colon \N \to \N$
there exists $k_{H,\funq} \in \N$
such that the algorithm $\mathcal{A}$ invoked
on input $(H,G,k,q)$ as above never returns $\bot$ provided the following conditions are satisfied: $k \geq k_{H,\funq}$ and $q(i) \leq \funq(i)$ for every $1 \leq i \leq k$.
\end{enumerate}
\end{theorem}

Let us see how our main result --- Theorem~\ref{thm:main} --- follows from combining Theorem~\ref{thm:unbrk-minor} with Theorem~\ref{thm:mainReduction}.

\begin{proof}[Proof of Theorem~\ref{thm:main} assuming Theorems~\ref{thm:mainReduction} and~\ref{thm:unbrk-minor}]
Observe that for every graph $H$ there exists a finite list of graphs ${\cal F}_H$ such that a graph $G$ is $H$-minor-free if and only if $G$ does not admit any member of ${\cal F}_H$ as a topological minor. Define
$$\mathbb{F}\coloneqq \{{\cal F}_H\,\colon\, H\textrm{ is a graph}\}\subseteq \mathbb{F}^\star.$$
Let $q\colon \mathbb{F}^\star\times \mathbb{N}\to \mathbb{N}$ be the function provided by Theorem~\ref{thm:mainReduction}. Now Theorem~\ref{thm:unbrk-minor} applied for $H$ and $\funq=q$ provides a canonization algorithm --- in the sense of Definition~\ref{def:canonUnbreak} --- for $\mathbb{F}$-free graphs that are $(q,\kappa)$-unbreakable, where we define $\kappa({\cal F}_H)=k_{H,q}$ (the latter is the constant provided by Theorem~\ref{thm:unbrk-minor}). Now it remains to apply Theorem~\ref{thm:mainReduction} to conclude that there is a canonization algorithm for $\mathbb{F}$-free graphs that works in time $g({\cal F})\cdot n^{\Oh(1)}$, where $g$ is some function. This is equivalent to the conclusion of Theorem~\ref{thm:main}.
\end{proof}

We now discuss how the proof of Theorem~\ref{thm:mainReduction}
breaks into parts in Sections~\ref{sec:firstDecomp}--\ref{sec:lastDecomp}.
To this end, we need to define canonization algorithms for graphs that are not necessarily unbreakable, but that are {\em{improved-clique-unbreakable}}, as explained in the following definitions.

\begin{definition}
A graph $G$ is {\em{$(q,k)$-clique-unbreakable}} if for every clique separation $(A,B)$ in $G$ of order at most $k$, we have $|A|\leq q$ or $|B|\leq q$. Further, $G$ is {\em{$(q,k,\ell)$-improved-clique-unbreakable}} if for every clique separation $(A,B)$ of order at most $k$ in the $\ell$-improved graph $G^{\langle\ell\rangle}$, we have $|A|\leq q$ and $|B|\leq q$.
\end{definition}

\begin{definition}
Let $\mathbb{F} \subseteq \mathbb{F}^\star$ be a (possibly infinite) collection of finite sets of graphs and  $\kappa \colon \mathbb{F} \rightarrow \mathbb{N}$ be a function.
An algorithm ${\cal B}$ is a {\em canonization algorithm for $\kappa$-improved-clique-unbreakable $\mathbb{F}$-free classes} if it has the following properties:
\begin{itemize}
\item ${\cal B}$ takes as input a finite list ${\cal F} \in \mathbb{F}$, integers $k$ and $s$, and an ${\cal F}$-topological-minor-free, $(s, k, k)$-improved-clique-unbreakable colored graph $G$ and outputs either $\bot$ or a labeling $\lambda$ of $G$. %

\item ${\cal B}$ is weakly isomorphism invariant. 

\item If $k \geq \kappa({\cal F})$  then the algorithm outputs a canonical labeling $\lambda$ of $G$.
\end{itemize}
\end{definition}

Theorem~\ref{thm:mainReduction} reduces the task of designing a canonization algorithm for $\mathbb{F}$-free classes to designing a canonization algorithm for unbreakable $\mathbb{F}$-free classes.
As an intermediate step we first reduce the task of designing a canonization algorithm for $\mathbb{F}$-free classes to designing a canonization algorithm for improved clique-unbreakable $\mathbb{F}$-free classes. This is formally stated as the following Lemma, which is proved in Section~\ref{sec:redToCliqueUnbreakable}. 

\begin{lemma}\label{lem:cliqueToGeneral}
Let $\mathbb{F} \subseteq \mathbb{F}^\star$ be a (possibly infinite) collection of finite sets of graphs.
Suppose there exists a canonization algorithm ${\cal B}$ for $\kappa_{\cal B}$-improved-clique-unbreakable
$\mathbb{F}$-free graphs, for some function $\kappa_{\cal B} \colon \mathbb{F} \rightarrow \mathbb{N}$. Then
there exists a canonization algorithm for $\mathbb{F}$-free classes. 
The running time of this latter algorithm is upper bounded by $g({\cal F})n^{\cO(1)}$ (for some function $g \colon \mathbb{F} \rightarrow \mathbb{N}$) plus the time taken by at most $g({\cal F})n^{\cO(1)}$ invocations of ${\cal B}$ on instances $(G', {\cal F}, k, q)$ where $G'$ is a $(q,k,k)$-improved-clique-unbreakable, ${\cal F}$-topological minor free graph on at most $n$ vertices and $k$ and $q$ are integers upper bounded by a function of ${\cal F}$ and $\kappa_{\cal B}$.
\end{lemma}

The second main ingredient of the proof of Theorem~\ref{thm:mainReduction} is reducing the task of  designing a canonization algorithm for improved clique-unbreakable $\mathbb{F}$-free classes to designing a canonization algorithm for unbreakable $\mathbb{F}$-free classes. This is formally stated as the following Lemma, which is proved in Section~\ref{sec:reductionToUnbreakable}.

\begin{lemma}\label{lem:unbreakToClique}
There exists a function $q^\star: \mathbb{F}^\star \times \mathbb{N} \rightarrow \mathbb{N}$ such that the following holds. 
For every collection of finite sets of graphs $\mathbb{F} \subseteq \mathbb{F}^\star$, 
if there exists a canonization algorithm ${\cal A}$ for $(q^\star, \kappa)$-unbreakable $\mathbb{F}$-free classes (for some  $\kappa : \mathbb{F} \rightarrow \mathbb{N}$) 
then there exists a function $\kappa_{\cal B} : \mathbb{F} \rightarrow \mathbb{N}$ and a canonization algorithm ${\cal B}$ for $\kappa_{\cal B}$-improved-clique-unbreakable $\mathbb{F}$-free classes.
The running time of ${\cal B}$ on an instance $(G,{\cal F},k,s)$ is upper bounded by $g({\cal F},k,s)n^{\cO(1)}$ and the total time taken by at most $h({\cal F},k)n^{\cO(1)}$ invocations of ${\cal A}$ on $({\cal F}, G, k/2, q)$, where ${\cal F} \in \mathbb{F}$, $q : [k/2] \rightarrow \mathbb{N}$ is a function such that $q(i) \leq q^\star(i)$ for $i \leq k/2$, and $G$ is an ${\cal F}$-free, $q$-unbreakable graph on at most $n$ vertices.
\end{lemma}

The proof of Theorem~\ref{thm:mainReduction} follows directly by using the premise of Theorem~\ref{thm:mainReduction} as the premise of Lemma~\ref{lem:unbreakToClique}, the conclusion of Lemma~\ref{lem:unbreakToClique} as the premise of Lemma~\ref{lem:cliqueToGeneral}, and observing that the conclusion of Lemma~\ref{lem:cliqueToGeneral} is the conclusion of Theorem~\ref{thm:mainReduction}. We now give a full proof.

\begin{proof}[Proof of Theorem~\ref{thm:mainReduction}.]
Let $q^\star: \mathbb{F}^\star \times \mathbb{N} \rightarrow \mathbb{N}$ be the function given by Lemma~\ref{lem:unbreakToClique}. We claim that $q^\star$ satisfies the conclusion of Theorem~\ref{thm:mainReduction}

Let $\mathbb{F} \subseteq \mathbb{F}^\star$ be a collection of finite sets of graphs,
$\kappa : \mathbb{F} \rightarrow \mathbb{N}$ be a function and 
${\cal A}$ be a canonization algorithm  for $(q^\star, \kappa)$-unbreakable $\mathbb{F}$-free classes.
By Lemma~\ref{lem:unbreakToClique} there exists a function $\kappa_{\cal B} : \mathbb{F} \rightarrow \mathbb{N}$ and a canonization algorithm ${\cal B}$ for $\kappa_{\cal B}$-improved-clique-unbreakable $\mathbb{F}$-free classes.
By Lemma~\ref{lem:cliqueToGeneral} there exists a canonization algorithm ${\cal C}$ for $\mathbb{F}$-free classes.

The running time of ${\cal C}$ is upper bounded by 
$g({\cal F})n^{\cO(1)}$ (for some function $g \colon \mathbb{F} \rightarrow \mathbb{N}$ as guaranteed by Lemma~\ref{lem:cliqueToGeneral}) plus the time taken by at most $g({\cal F})n^{\cO(1)}$ invocations of ${\cal B}$ on instances $(G', {\cal F}, k, s)$ where $G'$ is a $(s,k,k)$-improved-clique-unbreakable, ${\cal F}$-topological minor free graph on at most $n$ vertices and $k$ and $s$ are integers upper bounded by a function of ${\cal F}$ and $\kappa_{\cal B}$.
However, $\kappa_B$ is a function that depends solely on ${\cal F}$, and therefore $k$ and $s$ are integers upper bounded by a function of ${\cal F}$. We denote these functions by $k({\cal F})$  and $s({\cal F})$ respectively. 
By Lemma~\ref{lem:unbreakToClique}, the running time of each invocation of ${\cal B}$ on an instance $(G', {\cal F}, k({\cal F}), s({\cal F}))$ 
is upper bounded by $g'({\cal F},k({\cal F}),s({\cal F}))n^{\cO(1)}$ and the total time taken by at most $h({\cal F},k({\cal F}))n^{\cO(1)}$ invocations of ${\cal A}$ on $({\cal F}, G, k({\cal F})/2, q)$, where ${\cal F} \in \mathbb{F}$, $q : [k({\cal F})/2] \rightarrow \mathbb{N}$ is a function such that $q(i) \leq q^\star(i)$ for $i \leq k({\cal F})/2$, and $G$ is an ${\cal F}$-free, $q$-unbreakable graph on at most $n$ vertices.
Here the functions $g'$ and $h$ are the functions denoted by $g$ and $h$ in the conclusion of Lemma~\ref{lem:unbreakToClique}, respectively. 
Thus the total running time and total number of invocations of ${\cal A}$ by ${\cal C}$ is upper bounded by $g^\star({\cal F})n^{\cO(1)}$ for some function $g^\star$, as claimed. This concludes the proof.
\end{proof}

%% file: canonizationTools.tex
For technical reasons, we define two types of signatures: strong and weak. Each type has a designated lemma that explains its properties.

\begin{definition}[{\bf Strong $k$-Cut Signature}]\label{def:cutSignature}
Let $\Lambda$ be a set, $(G, \iota)$ be a $\Lambda$-boundaried graph. The {\em strong $k$-cut signature} of  $(G, \iota)$ is the function $\Gamma$ that takes as input two subsets $\hat{A}$ and $\hat{B}$ of $\Lambda$ such that  $\hat{A} \cup \hat{B} = \Lambda$ and outputs $\min(o - |\hat{A} \cap \hat{B}|, k+1)$, where $o$ is the smallest order of a separation $(A, B)$ of $G$ so that $A \cap \bnd G = \iota^{-1}(\hat{A})$ and $B \cap \bnd G = \iota^{-1}(\hat{B})$. If no such separation exists then $\Gamma(\hat{A}, \hat{B}) = \infty$.
\end{definition}

\begin{lemma}\label{lem:sameImproved}
Let $\Lambda$ be a set and $(G, \iota)$, $(G_1, \iota_1)$ and $(G_2, \iota_2)$ be $\Lambda$-boundaried graphs. If $(G_1, \iota_1)$ and $(G_2, \iota_2)$ have the same strong $k$-cut signature $\Gamma$ then
$$\imp{((G, \iota) \oplus (G_1, \iota_1))}{k}[V(G)] = \imp{((G, \iota) \oplus (G_2, \iota_2))}{k}[V(G)]$$
\end{lemma}

Said more plainly, the conclusion of Lemma~\ref{lem:sameImproved} states that the subgraphs of the $k$-improved graph induced by $V(G)$ of $(G, \iota) \oplus (G_1, \iota_1)$ and $(G, \iota) \oplus (G_2, \iota_2)$ are the same. 

\begin{proof}
Let $u$ and $v$ be vertices in $V(G)$ so that $uv \notin E(G)$. We prove that $uv \notin E(\imp{((G, \iota) \oplus (G_1, \iota_1))}{k})$ implies that $uv \notin E(\imp{((G, \iota) \oplus (G_2, \iota_2))}{k})$. Suppose $uv \notin E(\imp{((G, \iota) \oplus (G_1, \iota_1))}{k})$. Then there exists a $u$-$v$ separation $(A, B)$ of $(G, \iota) \oplus (G_1, \iota_1)$ of order less than $k$. Define $\hat{A} = \iota(A \cap \bnd G)$ and $\hat{B} = \iota(B \cap \bnd G)$. It follows that $\Gamma(\hat{A}, \hat{B}) \leq |(A \cap B) \setminus V(G)| < k$. Thus there exists a separation $(A_2, B_2)$ of $G_2$ of order $\Gamma(\hat{A}, \hat{B})$ so that $A_2 \cap \bnd G_2 = \iota^{-1}(\hat{A})$ and $B_2 \cap \bnd G_2 = \iota^{-1}(\hat{B})$. We have that $((A \cap V(G)) \cup A_2, (B \cap V(G)) \cup B_2)$ is a $u$-$v$ separation of $(G, \iota) \oplus (G_2, \iota_2)$ of order 
$$|A \cap B \cap V(G)| + \Gamma(\hat{A}, \hat{B}) \leq |A \cap B \cap V(G)| + |(A \cap B) \setminus V(G)| = |A \cap B| < k$$
witnessing that $uv$ is not an edge in $\imp{((G, \iota) \oplus (G_1, \iota_1))}{k}$. The proof that $uv \notin E(\imp{((G, \iota) \oplus (G_2, \iota_2))}{k})$ implies that $uv \notin E(\imp{((G, \iota) \oplus (G_1, \iota_1))}{k})$ is symmetric.
\end{proof}

\begin{definition}[{\bf Weak Cut Signature}]\label{def:weakcutSignature}
Let $\Lambda$ be a set, $(G, \iota)$ be a $\Lambda$-boundaried graph. The {\em weak cut signature} of  $(G, \iota)$ is a function $\Gamma$ that takes as input two subsets $\hat{A}$ and $\hat{B}$ of $\Lambda$ such that  $\hat{A} \cup \hat{B} = \Lambda$ and outputs $o - |\hat{A} \cap \hat{B}|$, where $o$ is the smallest order of a separation $(A, B)$ of $G$ so that $A \cap \bnd G \supseteq \iota^{-1}(\hat{A})$ and $B \cap \bnd G \supseteq \iota^{-1}(\hat{B})$. 
\end{definition}

The key difference between weak cut signatures and strong cut signatures is that for weak cut signatures we require containment  $A \cap \bnd G \supseteq \iota^{-1}(\hat{A})$ and $B \cap \bnd G \supseteq \iota^{-1}(\hat{B})$, while strong signatures require equality $A \cap \bnd G = \iota^{-1}(\hat{A})$ and $B \cap \bnd G = \iota^{-1}(\hat{B})$.
When equality is required (as in strong $k$-cut signatures) it is possible for the separation $(A, B)$ not to exist, and even if it exists then the smallest order of such a separation can be as large as $|V(G)|-2$. This leads to the possibility for $\Gamma$ from Definition~\ref{def:cutSignature} to take value $k+1$.
On the other hand, when only containment is required (for weak cut signatures), the separation $(A, B)$ with $A = \bnd G$ and $B = V(G)$ satisfies the requirement. Therefore the weak cut signature $\Gamma(\hat{A}, \hat{B}) \leq |\bnd G|$  for every boundaried graph $(G, \iota)$.  We formulate this as an observation

\begin{observation}
For every $\Lambda$-boundaried graph $(G, \iota)$ and $\hat{A},\hat{B}\subseteq \bnd G$ with $\hat{A}\cup \hat{B}=\Lambda$, the weak cut signature $\Gamma$ of $G$ satisfies $\Gamma(\hat{A}, \hat{B}) \leq |\bnd G|$. 
\end{observation}

We now prove a lemma similar in spirit to Lemma~\ref{lem:sameImproved} that shows that gluing onto $G$ two graphs with the same weak cut signature results in the same sets being $(q, r)$-unbreakable for every pair $r,q$ of positive integers. The proof is almost identical to the proof of Lemma~\ref{lem:sameImproved}.

\begin{lemma}\label{lem:sameUnbreakable}
Let $\Lambda$ be a set and $(G, \iota)$, $(G_1, \iota_1)$, $(G_2, \iota_2)$ be $\Lambda$-boundaried graphs, $q$, $r$ be non-negaitive integers and $S$ be a subset of $V(G)$. If $(G_1, \iota_1)$ and $(G_2, \iota_2)$ have the same weak cut signature $\Gamma$ then $S$ is $(q, r)$-unbreakable in  $(G, \iota) \oplus (G_1, \iota_1)$ if and only if $S$ is $(q, r)$-unbreakable in  $(G, \iota) \oplus (G_2, \iota_2)$
\end{lemma}

\begin{proof}
Let $S$ be a vertex set in $G$. We prove that $S$ being $(q, r)$-breakable in $(G, \iota) \oplus (G_1, \iota_1)$ implies that $S$ is $(q, r)$-breakable in $(G, \iota) \oplus (G_2, \iota_2)$. Suppose $S$ is $(q, r)$-breakable in $(G, \iota) \oplus (G_1, \iota_1)$ and let $(A, B)$ be a witnessing separation of $(G, \iota) \oplus (G_1, \iota_1)$. Define $\hat{A} = \iota(A \cap \bnd G)$ and $\hat{B} = \iota(B \cap \bnd G)$. It follows that $\Gamma(\hat{A}, \hat{B}) \leq |(A \cap B) \setminus V(G)|$. 
Thus there exists a separation $(A_2, B_2)$ of $G_2$ of order at most $\Gamma(\hat{A}, \hat{B}) + |\hat{A} \cap \hat{B}|$
so that $A_2 \cap \bnd G_2 \supseteq \iota^{-1}(\hat{A})$ and $B_2 \cap \bnd G_2 \supseteq \iota^{-1}(\hat{B})$. We have that $((A \cap V(G)) \cup A_2, (B \cap V(G)) \cup B_2)$ is a separation of $(G, \iota) \oplus (G_2, \iota_2)$ of order at most
$$|A \cap B \cap V(G)| + \Gamma(\hat{A}, \hat{B}) \leq |A \cap B \cap V(G)| + |(A \cap B) \setminus V(G)| = |A \cap B| \leq r.$$
Since $((A \cap V(G)) \cup A_2) \cap V(G) \supseteq A \cap V(G)$, $((B \cap V(G)) \cup B_2) \cap V(G) \supseteq B \cap V(G)$, and $S \subseteq V(G)$, it follows that $((A \cap V(G)) \cup A_2, (B \cap V(G)) \cup B_2)$ witnesses that $S$ is $(q, r)$-breakable in $(G, \iota) \oplus (G_2, \iota_2)$. The proof that $S$ being $(q, r)$-breakable in $(G, \iota) \oplus (G_2, \iota_2)$ implies that $S$ is $(q, r)$-breakable in $(G, \iota) \oplus (G_1, \iota_1)$ is symmetric.
\end{proof}

Next we show that when gluing two boundaried graphs, the weak cut signature of the resulting graph only depends on the weak cut signatures of the graphs we glue together. In fact we show a slightly stronger statement. 
The proof of the next Lemma is very similar to the proof of Lemma~\ref{lem:sameUnbreakable}.

\begin{lemma}\label{lem:sameWeakCutAfterGluing}
Let $\Lambda$ be a set and $(G, \iota)$, $(G_1, \iota_1)$, $(G_2, \iota_2)$ be $\Lambda$-boundaried graphs, $X,Y$ be subsets of $V(G)$. If $(G_1, \iota_1)$ and $(G_2, \iota_2)$ have the same weak cut signature $\Gamma$ then the minimum order of an $X$-$Y$ separation in $(G, \iota) \oplus (G_1, \iota_1)$ is equal to the  minimum order of an $X$-$Y$ separation in  $(G, \iota) \oplus (G_2, \iota_2)$.
\end{lemma}

\begin{proof}
Let $(A, B)$ be a minimum order $X$-$Y$ separation of $(G, \iota) \oplus (G_1, \iota_1)$.
We prove that $(G, \iota) \oplus (G_2, \iota_2)$ also has an $X$-$Y$ separation of order at most $|A \cap B|$.
Define $\hat{A} = \iota(A \cap \bnd G)$ and $\hat{B} = \iota(B \cap \bnd G)$. It follows that $\Gamma(\hat{A}, \hat{B}) \leq |(A \cap B) \setminus V(G)|$. 
Thus there exists a separation $(A_2, B_2)$ of $G_2$ of order at most $\Gamma(\hat{A}, \hat{B}) + |\hat{A} \cap \hat{B}|$
so that $A_2 \cap \bnd G_2 \supseteq \iota^{-1}(\hat{A})$ and $B_2 \cap \bnd G_2 \supseteq \iota^{-1}(\hat{B})$. We have that $((A \cap V(G)) \cup A_2, (B \cap V(G)) \cup B_2)$ is a separation of $(G, \iota) \oplus (G_2, \iota_2)$ of order at most
$$|A \cap B \cap V(G)| + \Gamma(\hat{A}, \hat{B}) \leq |A \cap B \cap V(G)| + |(A \cap B) \setminus V(G)| = |A \cap B|.$$
Since $((A \cap V(G)) \cup A_2) \cap V(G) \supseteq A \cap V(G)$, $((B \cap V(G)) \cup B_2) \cap V(G) \supseteq B \cap V(G)$, and $X,Y \subseteq V(G)$, it follows that $((A \cap V(G)) \cup A_2, (B \cap V(G)) \cup B_2)$ is an $X$-$Y$ separation in $(G, \iota) \oplus (G_2, \iota_2)$ of order at most $|A \cap B|$. 
The proof that the minimum order of an $X$-$Y$ separation in $(G, \iota) \oplus (G_1, \iota_1)$ is at most the minimum order of an $X$-$Y$ separation in $(G, \iota) \oplus (G_2, \iota_2)$ is symmetric. 
\end{proof}

The weak and strong cut signatures succintly capture the effect that a subgraph of $G$ has on the rest of the graph with respect to separator sizes. On the other hand, Lemma~\ref{lem:topFolioGlueAndForget} succintly captures the effect that a subgraph of $G$ has on the rest of the graph with respect to being ${\cal F}$-free. We now turn to studying the automorphisms of boundaried graphs, we are specifically interested in which permutations of the boundary can be extended to automorphisms of the whole boundaried graph. 

\begin{definition}[{\bf Extendable Permutation}]
Let $\Lambda$ be a set, $(G, \iota)$ be a $\Lambda$-boundaried, colored graph with $\iota(\bnd G) = \Lambda$. A permutation $\pi : \Lambda \rightarrow \Lambda$ is an {\em extendable permutation} for $(G, \iota)$ if there exists an automorphism $\phi$ of $G$ such that for every $v \in \bnd G$ it holds that $\phi(v) = \iota^{-1}(\pi(\iota(v)))$.
\end{definition}

Let $(G, \iota)$ be a colored $\Lambda$-boundaried graph. The {\em set of extendable permutations} of $G$ is simply the set of all functions  $\pi : \Lambda \rightarrow \Lambda$ such that $\pi$ is an extendable permutation of $(G, \iota)$.

We are now ready to prove a key ingredient of the proof of Theorem~\ref{thm:mainReduction}: that for every boundaried graph $G$ there exists a small gadget that preserves all relevant information for solving canonization on ${\cal F}$-free graphs. Here the precise meaning of small depends on whether we want to preserve strong or weak cut signatures.

\begin{lemma}\label{lem:equivalentStrongGadget}
For every three integers $t$, $f$, $k$ there exists an integer $\ell$ so that for every $t$-boundaried colored graph $(G, \iota)$ there exists a $t$-boundaried colored graph $(G', \iota')$ such that
\begin{itemize}\setlength\itemsep{-.7mm}
\item $|V(G')| \leq \ell$,
\item $|V(G')| \leq |V(G)|$,
\item No vertex in $\bnd(G', \iota')$ has the same color as a vertex in $V(G') \setminus \bnd(G', \iota')$.
\item The range of ${\sf col}_{G'}$ on $V(G') \setminus \bnd(G', \iota)$ is $[\ell]$.
\item Every component $C$ of $G - \bnd{G}$ satisfies $N(C) = \bnd{G}$ if and only if every component $C'$ of $G' - \bnd{G'}$ satisfies $N(C') = \bnd{G'}$,
\item $(G, \iota)$ and $(G', \iota')$ are boundary-consistent,
\item $(G, \iota)$ and $(G', \iota')$ have the same strong $k$-cut signature,
\item $(G, \iota)$ and $(G', \iota')$ have the same $f$-folio
\item $(G, \iota)$ and $(G', \iota')$ have the same set of extendable permutations.
\end{itemize}
\end{lemma}

\begin{proof}
Let $(G, \iota)$ be a $t$-boundaried, colored graph. We can assume without loss of generality that the coloring  ${\sf col}_{G}$ of $(G, \iota)$ satisfies the following property: no vertex in $\bnd(G, \iota)$ has the same color as a vertex outside of $\bnd(G, \iota)$.
If ${\sf col}_{G}$ does not have this property we can modify it as follows. Let $c$ be the highest color used by   ${\sf col}_{G}$. Add $c+1$ to the color of every vertex in $V(G) \setminus \bnd(G, \iota)$. This modification does not affect the coloring of $\bnd(G, \iota)$ and it does not change the set of extendable permutations of $(G, \iota)$. However after the modification no vertex in $\bnd(G, \iota)$ has the same color as a vertex outside of $\bnd(G, \iota)$.

Consider the following equivalence relation $\equiv$ defined on colored $t$-boundaried graphs such that no vertex on the boundary has the same color as a vertex outside of the boundary. We say that $(G_1, \iota_1) \equiv (G_2, \iota_2)$ if
\begin{itemize}\setlength\itemsep{-.7mm}
\item Every component $C$ of $G - \bnd{G}$ satisfies $N(C) = \bnd{G}$ if and only if every component $C'$ of $G' - \bnd{G'}$ satisfies $N(C') = \bnd{G'}$,
\item $(G_1, \iota_1)$ and $(G_2, \iota_2)$ are weakly boundary-consistent. 
\item $(G_1, \iota_1)$ and $(G_2, \iota_2)$ have the same strong $k$-cut signature,
\item $(G_1, \iota_1)$ and $(G_2, \iota_2)$ have the same $f$-folio,
\item $(G_1, \iota_1)$ and $(G_2, \iota_2)$ have the same set of extendable permutations.
\end{itemize}
Clearly $\equiv$ is an equivalence relation where every equivalence class is completely defined by a single bit (whether or not every component $C$ of $G - \bnd{G}$ satisfies $N(C) = \bnd{G}$), a boundary-color-partition, a labeled graph with vertex labels from $[t]$, a $k$-cut signature, an $f$-folio, and a set of extendable permutations of $[t]$. Thus the number of equivalence classes of $\equiv$, as well as the size of the smallest member of each equivalence class, is upper bounded by a function of $k$, $t$ and $f$. Let $\ell$ be the smallest integer so that for every equivalence class of $\equiv$ there is at least one (boundaried) graph on at most $\ell$ vertices. We claim that $\ell$ satisfies the conclusion of the lemma.

Let $(G', \iota')$ be a colored $t$-boundaried graph such that $|V(G')| \leq \ell$, $|V(G')| \leq |V(G)|$ and  $(G, \iota) \equiv (G', \iota')$.
By the definition of $\equiv$, we have that  $(G, \iota)$ and $(G', \iota')$ have the same  strong $k$-cut signature, $f$-folio and set of extendable permutations. Further, every component $C$ of $G - \bnd{G}$ satisfies $N(C) = \bnd{G}$ if and only if every component $C'$ of $G' - \bnd{G'}$ satisfies $N(C') = \bnd{G'}$

Let $\textsf{col}_G'$ be the coloring of $(G', \iota')$ and let ${\cal P}$ be the color partition of  $\textsf{col}_G'$. We construct the coloring $\textsf{col}_G''$ of $(G', \iota')$ by first iterating over all parts $P \in {\cal P}$ so that $P \subseteq \bnd(G', \iota')$ and assigning to all vertices of $P$ the color $\textsf{col}_G(\iota^{-1}\iota'(v))$ for some $v \in \bnd G' \cap P$. Notice that because $(G, \iota)$ and $(G', \iota')$ are weakly boundary-consistent all choices of $v \in \bnd G' \cap P$ lead to the same value of $\textsf{col}_G(\iota^{-1}\iota'(v))$. Further, since no vertex in $\bnd(G, \iota')$ has the same color as a vertex outside $\bnd(G, \iota')$, this defines $\textsf{col}_G''$ for all vertices in $\bnd(G, \iota')$ and no other vertices.

To complete the construction of $\textsf{col}_G''$ we iterate over the remaining parts $P$ in ${\cal P}$ in an arbitrary order, and assign to all vertices in $P$ the lowest color from $[\ell]$ not yet assigned to any other vertices. A color from  $[\ell]$ will always be available because $G'$ has at most $\ell$ colors. 
From Observation~\ref{obs:samePartition}, and because $\textsf{col}_G''$ was constructed to have the same color partition as $\textsf{col}_G'$ while also being boundary-consistent with $(G, \iota)$ and coloring $\textsf{col}_G$, it follows that $(G', \iota')$ with coloring  $\textsf{col}_G''$ satisfies the conclusion of the lemma. 
\end{proof}

We will also need a lemma similar to Lemma~\ref{lem:equivalentStrongGadget}, but preserving weak cut signatures rather than strong $k$-cut signatures. Notice that here the size of the representative is independent of the cutsize.

\begin{lemma}\label{lem:equivalentWeakGadget}
For every pair of integers $t$, $f$ there exists an integer $\ell$, so that for every $t$-boundaried colored graph $(G, \iota)$ there exists a $t$-boundaried colored graph $(G', \iota')$ such that
\begin{itemize}\setlength\itemsep{-.7mm}
\item $|V(G')| \leq \ell$,
\item $|V(G')| \leq |V(G)|$,
\item No vertex in $\bnd(G', \iota')$ has the same color as a vertex in $V(G') \setminus \bnd(G', \iota')$.
\item The range of ${\sf col}_{G'}$ on $V(G') \setminus \bnd(G', \iota)$ is $[\ell]$.
\item Every component $C$ of $G - \bnd{G}$ satisfies $N(C) = \bnd{G}$ if and only if every component $C'$ of $G' - \bnd{G'}$ satisfies $N(C') = \bnd{G'}$,
\item $(G, \iota)$ and $(G', \iota')$ are boundary-consistent,
\item $(G, \iota)$ and $(G', \iota')$ have the same weak  cut signature,
\item $(G, \iota)$ and $(G', \iota')$ have the same $f$-folio
\item $(G, \iota)$ and $(G', \iota')$ have the same set of extendable permutations.
\end{itemize}
\end{lemma}

The proof of Lemma~\ref{lem:equivalentWeakGadget} is identical to the proof of Lemma~\ref{lem:equivalentStrongGadget} with the following exceptions: the equivalence relation $\equiv$ requires $(G_1, \iota_1)$ and $(G_2, \iota_2)$ to have the same weak cut signature rather than the same strong $k$-cut signature, and therefore the (number of) equivalence classes of $\equiv$ are upper bounded by a function of $t$ and $f$.

\subsubsection{Strong And Weak Representatives and their Size Bounds ($\eta_s$ and $\eta_w$)}\label{sec:representatives}
Let $(G, \iota)$ be a boundaried (colored) graph. The $(k,f)$-{\em strong representative} of $(G, \iota)$ is the lexicographically smallest
boundaried colored graph $(G^R, \iota^R)$ that satisfies the conclusion of Lemma~\ref{lem:equivalentStrongGadget}. By Lemma~\ref{lem:equivalentStrongGadget} such that the $(k,f)$-strong representative $(G^R, \iota^R)$ exists and has at most $\ell$ vertices for some integer $\ell$ depending only on $k$, $f$ and $|\bnd G|$. We will denote by ${\eta}_s(k, f, t)$ the smallest $\ell$ such that every $t$-boundaried graph has a $(k,f)$-strong representative on at most $\ell$ vertices.

Similarly, the $f$-{\em weak representative} of $(G, \iota)$ is the lexicographically smallest boundaried colored graph $(G^R, \iota^R)$ that satisfies the conclusion of Lemma~\ref{lem:equivalentWeakGadget}. By Lemma~\ref{lem:equivalentWeakGadget}, the $f$-weak representative $(G^R, \iota^R)$ exists and has at most $\ell$ vertices for some integer $\ell$ depending only on $k$, $f$ and $|\bnd G|$. We will denote by ${\eta}_w(f, t)$ the smallest $\ell$ such that every $t$-boundaried graph has an $f$-weak representative on at most $\ell$ vertices.

The set of boundaried colored graphs $(G^R, \iota^R)$ that satisfy the conclusion of Lemma~\ref{lem:equivalentStrongGadget} (Lemma~\ref{lem:equivalentWeakGadget}) is isomorphism invariant in $(G, \iota)$, and selecting a lexicographically first element from such a set is a weakly isomorphism invariant operation (by Lemma~\ref{lem:selectWeaklyIsoInvariant} and Lemma~\ref{lem:compactColorEncoding}). We conclude the following.

\begin{observation}\label{obs:repIsWeakIso} The  $(k,f)$-strong representative of $(G, \iota)$ and the $f$-weak representative of $(G, \iota)$ are weakly isomorphism invariant in $(G, \iota)$. 
\end{observation}

We now show that if we have sufficiently analyzed a boundaried graph $(G, \iota)$, then we can efficiently compute a representative for it. We remind the reader that a {\em proper} labeling of a graph $G$ is a bijection from $V(G)$ to $[|V(G)|]$.

\begin{lemma}\label{lem:computeWeakRep}
There exists an algorithm that takes as input an integer $h$, a boundaried graph $(G, \iota)$ together with a set
$\{ \lambda_\pi ~:~ \pi : \bnd(G) \rightarrow [|\bnd(G)|] \}$ that contains a weakly isomorphism invariant proper labeling $\lambda_\pi$ of $(G, \pi)$ for every $\pi : \bnd(G) \rightarrow [|\bnd(G)|]$ such that $\lambda_\pi(v) = \pi(v)$ for all $v \in \bnd(G)$.
The algorithm runs in time $f(h,|\bnd(G)|)n^{\cO(1)}$ for some function $f$, and outputs the $h$-weak representative $(G_R, \iota_R)$ of $(G, \iota)$.
\end{lemma}

\begin{proof}
The algorithm computes the weak signature of $(G, \iota)$ by computing using Lemma~\ref{lem:f-f} the minimum order of an $A$-$B$ separation for every $A, B \subseteq \bnd(G)$ such that $A \cup B = \bnd(G)$. This takes time $3^kn^{\cO(1)}$.
The algorithm computes the $h$-folio of $G$ in time $g(h, |\bnd{}(G)|)n^{\cO(1)}$ using the algorithm of \cite{GroheKMW11}.
Then the algorithm computes the set of extendable permutations of $(G, \iota)$ using the following observation: A permutation $\pi : [|\bnd{}(G)|] \rightarrow [|\bnd{}(G)|]$ is an extendable permutation of $(G, \iota)$ if and only if $(G, \iota)$ and $(G, \iota')$ are isomorphic, where $\iota' = \iota^{-1}\pi\iota$. However $(G, \iota)$ and $(G, \iota')$ are isomorphic if and only if $(G, \iota)$ with proper labeling $\lambda_\iota$ is equal to $(G, \iota)$ with proper labeling $\lambda_{\iota'}$.
For each $\pi : [|\bnd{}(G)|] \rightarrow [|\bnd{}(G)|]$ the algorithm sets $\iota' = \iota^{-1}\pi\iota$ and checks whether  $(G, \iota)$ with proper labeling $\lambda_\iota$ is equal to $(G, \iota)$ with proper labeling $\lambda_{\iota'}$. If yes, then $\pi$ is an extendable permutation, otherwise it is not. In total this takes time $\cO(|\bnd(G)|!(n+m))$.

Having computed the weak signature, $h$-folio and the set of extendable permutations of $(G, \iota)$, the algorithm enumerates all colored properly labeled graphs $(G', \iota')$ such that
\begin{itemize}\setlength\itemsep{-.7pt}
\item No vertex in $\bnd(G', \iota')$ has the same color as a vertex in $V(G') \setminus \bnd(G', \iota')$.
\item The range of ${\sf col}_{G'}$ on $V(G') \setminus \bnd(G', \iota)$ is $[|V(G')|]$.
\item $(G, \iota)$ and $(G', \iota')$ are boundary-consistent,
\item Every component $C$ of $G - \bnd{G}$ satisfies $N(C) = \bnd{G}$ if and only if every component $C'$ of $G' - \bnd{G'}$ satisfies $N(C') = \bnd{G'}$
\end{itemize}
in lexicographic order, and checks whether
\begin{itemize}\setlength\itemsep{-.7pt}
\item $(G, \iota)$ and $(G', \iota')$ have the same weak cut signature,
\item $(G, \iota)$ and $(G', \iota')$ have the same $f$-folio,
\item $(G, \iota)$ and $(G', \iota')$ have the same set of extendable permutations.
\end{itemize}
The algorithm outputs the first $(G', \iota')$ that passes this test. 

To determine whether $(G, \iota)$ and $(G', \iota')$ have the same weak cut signature, and $f$-folio  set of extendable permutations the algorithm computes the weak cut signature and $f$-folio of $(G', \iota')$ in the same way as it did for $(G, \iota)$. It computes the set of extendable permutations of $(G', \iota')$ by brute force - by iterating over all bijections from $V(G')$ to $V(G')$ that map $\bnd(G')$ to $\bnd(G')$ and checking whether they are automorphisms of $G'$. For each candidate graph $(G', \iota')$ the time taken to perform these tests is upper bounded by a computable function of $|V(G')|$.

By Lemma~\ref{lem:equivalentWeakGadget} this algorithm will find a graph $(G', \iota')$ on $\eta_w(h, |\delta(G)|)$ vertices that passes the test. Since $(G', \iota')$ is the lexicographically first boundaried graph that satisfies the conclusion of Lemma~\ref{lem:equivalentWeakGadget} it is the $h$-weak representative of $(G, \iota)$

The algorithm iterates over a computable function of  $\eta_w(h, |\delta(G)|)$ many graphs before identifying the representative $(G', \iota')$. For each such graph the time it spends is upper bounded by a computable function of $\eta_w(h, |\delta(G)|)$. Thus the algorithm terminates within the claimed running time. 
\end{proof}

An identical argument yields the corresponding lemma for computing strong representatives. We omit the proof.

\begin{lemma}\label{lem:computeStrongRep}
There exists an algorithm that takes as input integers $h$, $k$, a boundaried graph $(G, \iota)$ together with a set
$\{ \lambda_\pi ~:~ \pi : \bnd(G) \rightarrow [|\bnd(G)|] \}$ that contains a weakly isomorphism invariant proper labeling $\lambda_\pi$ of $(G, \pi)$ for every $\pi : \bnd(G) \rightarrow [|\bnd(G)|]$ such that $\lambda_\pi(v) = \pi(v)$ for all $v \in \bnd(G)$.
The algorithm runs in time $f(h,k,|\bnd(G)|)n^{\cO(1)}$ for some function $f$, and outputs the $(k,h)$-strong representative $(G_R, \iota_R)$ of $(G, \iota)$.
\end{lemma}

\paragraph{A Note about Computability.}
We remark that the functions $\eta_s(k, h)$ and $\eta_w(h)$ are not guaranteed to be computable. Indeed, they are defined from a purely existential argument in the proofs of Lemmata~\ref{lem:equivalentStrongGadget} and~\ref{lem:equivalentWeakGadget} that in an equivalence relation with a finite number of equivalence classes of graphs, the size of the largest (when taken over all equivalence classes) smallest graph in the equivalence class is finite.  The fact that we do not know whether $\eta_s(k, h)$ and $\eta_w(h)$ are computable is the only reason why the dependence $f$ on $H$ in the running time of our algorithm of Theorem~\ref{thm:main} might not be computable. 
For an example, the running time of the algorithm of Lemmata~\ref{lem:computeWeakRep} and~\ref{lem:computeStrongRep} depend on $\eta_w$ and $\eta_s$ respectively.

We speculate that it might be possible (saying we {\em believe} it is possible would probably be too strong) to turn the algorithm of Theorem~\ref{thm:main} into a proof that $\eta_s(k, h)$ and $\eta_w(h)$ are computable for families ${\cal F}$ that are defined by excluding a single graph $H$ as a minor. This would in turn yield a computable upper bound on the running time dependence $f$ on $H$ in Theorem~\ref{thm:main}.

\paragraph{Isomorphisms Between Boundaried Graphs and Between Their Representatives.}
For the next lemma, let us recall the convention that all graphs can be possibly colored. 

\begin{lemma}\label{lem:replacementKeepsPerm}
Let $(G, \iota)$ and $(\hat{G}, \hat{\iota})$ be isomorphic $t$-boundaried graphs and let $(G^R, \iota^R)$ and $(\hat{G}^R, \hat{\iota}^R)$ be their $(k,f)$-strong representatives (or $f$-weak representatives). 
For every permutation $\pi \colon [t] \rightarrow [t]$
there exists an isomorphism $\psi$ from $G$ to $\hat{G}$ mapping $\bnd G$ to $\bnd \hat{G}$ such that $\pi = \hat{\iota}\psi\iota^{-1}$
if and only if
there exists an isomorphism $\psi^R$ from $G^R$ to $\hat{G}^R$ mapping $\bnd G^R$ to $\bnd \hat{G}^R$ such that $\pi = \hat{\iota}^R\psi^R(\iota^R)^{-1}$.
\end{lemma}

\begin{proof}
We start by proving the forward direction, namely that if there exists an isomorphism $\psi$ from $G$ to $\hat{G}$ mapping $\bnd G$ to $\bnd \hat{G}$ such that $\pi = \hat{\iota}\psi\iota^{-1}$
then
there exists an isomorphism $\psi^R$ from $G^R$ to $\hat{G}^R$ mapping $\bnd G^R$ to $\bnd \hat{G}^R$ such that $\pi = \hat{\iota}^R\psi^R(\iota^R)^{-1}$. Since $(G, \iota)$ and $(\hat{G}, \hat{\iota})$ are isomorphic there exists an isomorphism $\phi$ from $G$ to $\hat{G}$ that maps  $\bnd G$ to $\bnd \hat{G}$ such that  $\hat{\iota}\phi\iota^{-1}$ is the identity permutation of $[t]$.
Further, by assumption there is an isomorphism $\psi$ from $G$ to $\hat{G}$ mapping $\bnd G$ to $\bnd \hat{G}$ such that $\pi = \hat{\iota}\psi\iota^{-1}$
Thus $\phi^{-1}\psi$ is an automorophism of $G$ mapping $\bnd G$ to $\bnd G$ such that
$$\iota \phi^{-1}\psi \iota^{-1} = \iota \phi^{-1} \hat{\iota}^{-1} \hat{\iota}  \psi \iota^{-1} = \mathsf{id}\, \pi = \pi$$
Hence $\pi$ is an extendable permutation of $G$. Thus, since $G^R$ and $G$ have the same set of extendable permutations, $\pi$ is also an extendable permutation of $G^R$. By the definition of extendable permutations there exists an automorphism $\mu$ of $G^R$ that maps $\bnd G^R$ to $\bnd G^R$ such that $\pi = \iota^R \mu(\iota^R)^{-1}$.
Furthermore, by Observation~\ref{obs:repIsWeakIso} there also is an isomorphism $\phi^R$ from $G^R$ to $\hat{G}^R$ such that $\hat{\iota}^R\phi^R(\iota^R)^{-1}$ is the identity permutation of $[t]$.
Let $\psi^R = \phi^R\mu$. Since it is a composition of an automorphism of $G^R$ mapping $\bnd G^R$ to $\bnd G^R$ with an isomorphism from $G^R$ to $\hat{G}^R$ that maps $\bnd G^R$ to $\bnd \hat{G}^R$ we have that $\psi^R$ is an isomorphism from $G^R$ to $\hat{G}^R$ mapping $\bnd G^R$ to $\bnd \hat{G}^R$ as required.
Finally we have that 
$$\hat{\iota}^R \psi^R (\iota^R)^{-1} = \hat{\iota}^R \phi^R\mu (\iota^R)^{-1} = \hat{\iota}^R \phi^R  (\iota^R)^{-1} \iota^R \mu (\iota^R)^{-1} = \mathsf{id}\,\pi = \pi$$
This concludes the proof of the forward direction.

We now prove the reverse direction, namely that if there exists an isomorphism $\psi^R$ from $G^R$ to $\hat{G}^R$ mapping $\bnd G^R$ to $\bnd \hat{G}^R$ such that $\pi = \hat{\iota}^R\psi^R(\iota^R)^{-1}$ then there exists an isomorphism $\psi$ from $G$ to $\hat{G}$ mapping $\bnd G$ to $\bnd \hat{G}$ such that $\pi = \hat{\iota}\psi\iota^{-1}$. The proof is nearly identical to that of the forward direction, we nevertheless include it for completeness. 
Since $(G, \iota)$ and $(\hat{G}, \hat{\iota})$ are isomorphic there exists an isomorphism $\phi$ from $G$ to $\hat{G}$ that maps  $\bnd G$ to $\bnd \hat{G}$ such that  $\hat{\iota}\phi\iota^{-1}$ is the identity permutation of $[t]$. Furthermore, by Observation~\ref{obs:repIsWeakIso} there also is an isomorphism $\phi^R$ from $G^R$ to $\hat{G}^R$ such that $\hat{\iota}^R\phi^R(\iota^R)^{-1}$ is the identity permutation of $[t]$.

Further, by assumption there is an isomorphism $\psi^R$ from $G^R$ to $\hat{G}^R$ mapping $\bnd G^R$ to $\bnd \hat{G}^R$ such that $\pi = \hat{\iota}^R\psi^R(\iota^R)^{-1}$.
Thus $(\phi^R)^{-1}\psi^R$ is an automorophism of $G^R$ mapping $\bnd G^R$ to $\bnd G^R$ such that
$$\iota^R (\phi^R)^{-1}\psi^R (\iota^R)^{-1} = \iota^R (\phi^R)^{-1} (\hat{\iota}^R)^{-1} \hat{\iota}^R  \psi^R (\iota^R)^{-1} = \mathsf{id}\, \pi = \pi$$
Hence $\pi$ is an extendable permutation of $G^R$. Thus, since $G^R$ and $G$ have the same set of extendable permutations $\pi$ is also an extendable permutation of $G$. By the definition of extendable permutations there exists an automorphism $\mu$ of $G$ that maps $\bnd G$ to $\bnd G$ such that $\pi = \iota \mu \iota^{-1}$.

Let $\psi = \phi\mu$. Since it is a composition of an automorphism of $G$ mapping $\bnd G$ to $\bnd G$ with an isomorphism from $G$ to $\hat{G}$ that maps $\bnd G$ to $\bnd \hat{G}$ we have that $\psi$ is an isomorphism from $G$ to $\hat{G}$ mapping $\bnd G$ to $\bnd \hat{G}$ as required.
Finally we have that 
$$\hat{\iota} \psi (\iota)^{-1} = \hat{\iota} \phi\mu (\iota)^{-1} = \hat{\iota} \phi  \iota^{-1} \iota \mu \iota^{-1} = \mathsf{id}\,\pi = \pi$$
This concludes the proof. 
\end{proof}

\subsubsection{Star Decompositions and Leaf-Labelings}
Given a graph $G$ and a star decomposition $(T, \chi)$  of $G$ with central node $b$, a {\em{leaf-labeling}} of $(T, \chi)$ is a set of labelings $\{ \lambda_{\pi, \ell} \}$ with a proper labeling  $\lambda_{\pi, \ell}$ of $G[\chi(\ell)]$ for every leaf $\ell$ of $T$ and every permutation $\pi \colon \sigma(\ell) \rightarrow [|\sigma(\ell)|]$. For every $\ell$ and $\pi$ the labeling $\lambda_{\pi, \ell}$ is required to satisfy $\lambda_{\pi, \ell}(v) = \pi(v)$ for every $v \in \sigma(\ell)$. 
If all the labelings $\lambda_{\pi, \ell}$ are compact we say that  $\{ \lambda_{\pi, \ell} \}$ is a compact leaf labeling.  We will say that a star decomposition  $(T, \chi)$  of $G$ is (compactly) {\em leaf-labeled} if it comes with a (compact) leaf labeling $\{ \lambda_{\pi, \ell} \}$. A leaf labeling is called {\em locally weakly isomorphism invariant} if each labeling $\lambda_{\pi, \ell}$ is weakly isomorphism invariant in $G[\chi(\ell)]$ and $\pi$.

\paragraph{Lexicographically First Labelings of a Leaf Labeling}
Given a leaf-labeling  $\{ \lambda_{\pi, \ell} \}$ of $(T, \chi)$, for each leaf $\ell$ of $T$ let $\pi\langle\ell\rangle$ be a permutation $\pi(\ell) : \sigma(\ell) \rightarrow [|\sigma(\ell)|]$ such that the labeled graph $G[\chi(\ell)]$ with labeling $\lambda_{\pi\langle\ell\rangle}$ is lexicographically smallest. We shorten notation by writing $\lambda\langle\ell\rangle = \lambda_{\pi\langle\ell\rangle, \ell}$. 

\smallskip
\noindent
{\bf Remark:} We note that there can be multiple permutations $\pi$ such that $G[\chi(\ell)]$ with labeling $\lambda_\pi$ are lexicographically smallest. In this case we select an arbitrary one of such permutations $\pi$ as $\pi\langle\ell\rangle$. %
When considering the boundaried and labeled graph $(G[\chi(\ell), \lambda\langle\ell\rangle])$ as a stand alone object, disconnected from $G$, all candidate choices of $\pi\langle\ell\rangle$ result in the same $(G[\chi(\ell), \lambda\langle\ell\rangle])$.
It is therefore very tempting to conjecture that this selection of $\pi\langle\ell\rangle$ is weakly isomorphism invariant in $G$, despite the arbitrary choice. Unfortunately, this is not true. %

\begin{lemma} \label{lem:equalSetsOfLabeled}
Let $G$ be a graph, $(T, \chi)$ be a weakly isomorphism invariant star decomposition of $G$, and $\{ \lambda_{\pi, \ell} \}$ be a locally weakly isomorphism invariant leaf-labeling of $(T, \chi)$.
Let $\hat{G}$ be a graph, $(\hat{T}, \hat{\chi})$ be a weakly isomorphism invariant star decomposition of $\hat{G}$, and $\{ \hat{\lambda}_{\pi, \hat{\ell}} \}$ be a locally weakly isomorphism invariant leaf-labeling of $(\hat{T}, \hat{\chi})$.
If $G$ and $\hat{G}$ are isomorphic then the multi-sets of properly labeled graphs
$$\{ (G[\chi(\ell)], \lambda\langle\ell\rangle ~:~ \ell \in V(T) \setminus \{b\} \} \mbox{ and } \{ (\hat{G}[\hat{\chi}(\hat{\ell})], \hat{\lambda}\langle\hat{\ell}\rangle ~:~ \hat{\ell} \in V(\hat{T}) \setminus \{\hat{b}\} \}$$
are equal. 
\end{lemma}

The proof of Lemma~\ref{lem:equalSetsOfLabeled} (omitted) follows directly from weak isomorphism invariance of $(T, \chi)$, and that the labeling $\lambda\langle\ell\rangle$ is weakly isomorphism invariant in $G[\chi(\ell)], \sigma(\ell)$.

\paragraph{Rank.} 
Let $G$ be a colored graph $G$, $(T, \chi)$ be a star decomposition of $G$ with central bag $b \in V(T)$,  $\{ \lambda_{\pi, \ell} \}$ be a locally weakly isomorphism invariant leaf-labeling of $(T, \chi)$.
For each leaf $\ell$, the {\em rank} of $\ell$ (with respect to $(T, \chi)$ and the leaf labelings $\{\lambda_{\pi, \ell}\}$) is the number of leaves $\ell'$ in $T$ so that $G[\chi(\ell')]$ with labeling $\lambda\langle\ell'\rangle$ is lexicographically (strictly) smaller than $G[\chi(\ell')]$ with labeling $\lambda\langle\ell'\rangle$. From Lemma~\ref{lem:equalSetsOfLabeled} and the definition of rank we can immediately conclude the following. 

\begin{lemma} \label{lem:equalRankIso}%
Let $G$ be a graph, $D \subseteq V(G)$ be a vertex set, $(T, \chi)$ be an isomorphism invariant (with respect to $G$ and $D$) star decomposition of $G$ with central bag $b \in V(T)$, and $\{ \lambda_{\pi, \ell} \}$ be a locally weakly isomorphism invariant leaf-labeling of $(T, \chi)$.
For every pair $\ell$, $\ell'$, the following are equivalent:
\begin{itemize}
\item There is an isomorphism from $G[\chi(\ell)]$ to $G[\chi(\ell')]$ that maps $\sigma(\ell)$ to $\sigma(\ell')$,
\item $G[\chi(\ell)]$ with labeling $\lambda\langle\ell\rangle$ and $G[\chi(\ell')]$ with labeling $\lambda\langle\ell'\rangle$ are equal,
\item The leaves $\ell$ and $\ell'$ have the same rank. 
\end{itemize}
\end{lemma}

\paragraph{Leaf-Labeled Star Decompositions With Small Centers.}
We give an algorithm for handling the case when we have a star decompostion with a small center.

\begin{lemma}\label{lem:smallCentersAlgorithmWrapper}
There exists an algorithm that takes as input a compactly colored graph $G$, a vertex set $D \subseteq V(G)$, a star decomposition $(T, \chi)$ of $G$ with central bag $b$ such that $D \subseteq \chi(b)$, and a compact locally weakly isomorphism invariant leaf-labeling $\{ \lambda_{\pi, \ell} \}$ of $(T, \chi)$. The algorithm runs in time $(|\chi(b)|)! \cdot n^{\cO(1)}$ and outputs for every permutation $\iota \colon D \rightarrow [|D|]$ a proper labeling of $(G, \iota)$ which is weakly isomporphism invariant in $(G, \iota)$ and  $(T, \chi)$. 
\end{lemma}

\begin{proof}
The algorithm iterates through all choices of a permutation $\gamma \colon \chi(b) \rightarrow [|\chi(b)|]$. For each choice of $\gamma$ the algorithm produces a proper labeling $\lambda$ of $G$ that coincides with $\gamma$ on $\chi(b)$. For each $\iota : D \rightarrow [|D|]$ the algorithm then outputs a  $\lambda$ that yields a lexicographically smallest labeled graph $(G, \lambda)$ over all the choices of $\gamma$ that coincide with $\iota$ on $D$. 

Since picking a lexicographically first labeling is a weakly isomorphism invariant function, to prove correctness it suffices to show the construction of $\lambda$ is weakly isomorphism invariant with respect to $G$, $(T, \chi)$ the leaf labeling $\{ \lambda_{\pi, \ell} \}$ as well as with respect to $\gamma$. We now describe how the algorithm constructs $\lambda$ given $\gamma$.

The algorithm proceeds as follows. For each leaf $\ell$ define the permutation $\pi^\ell\colon \sigma(\ell)\to [|\sigma(\ell)|]$ (not to be confused with $\pi\langle\ell\rangle$) to be the unique permutation of $\sigma(\ell)$ so that $\pi(u) < \pi(v)$ if and only if $\gamma(u) < \gamma(v)$. We define $\lambda^\ell$ to be $\lambda_{\pi^\ell, \ell}$.
The algorithm sorts the leaves $\ell$ of $T$ according to the lexicographic ordering of the set of labels $\gamma( \sigma(\ell) )$. For leaves $\ell$ and $\ell'$ that are tied according to this order (so $\sigma(\ell) = \sigma(\ell')$) it breaks ties by picking first the lexicographically smallest of $G[\chi(\ell)]$ with labeling $\lambda^\ell$, and $G[\chi(\ell')]$ with labeling $\lambda^{\ell'}$. Note that if $\ell$ and $\ell'$ are still tied then $(\lambda^{\ell'})^{-1}\lambda^\ell$ is an isomorhism from  $G[\chi(\ell)]$ to $G[\chi(\ell')]$ that maps each vertex $v$ of $\sigma(\ell)$ to itself. 

The labeling $\lambda$ is defined as follows. For every vertex $v \in \chi(b)$ where $b$ is the central bag of $T$ we set $\lambda(v) = \gamma(v)$. For every leaf $\ell$ and every vertex $v$ in $\chi(\ell) \setminus \chi(b)$ we set $\lambda(v) = \lambda^ \ell(v) + ((n+1) \cdot p)$ where $p$ is the position of $\ell$ in the sorted order of leaves defined above.  The algorithm outputs a compactification of $\lambda$ as the weakly isomorphism invariant (with respect to $G$, $(T, \chi)$, and $\gamma$) labeling of $G$. The algorithm clearly runs in time $(|\chi(b)|)! \cdot n^{\cO(1)}$ and outputs a  labeling of $G$. Thus to complete the proof it suffices to show that $\lambda$ is in fact weakly isomorphism invariant in $G$, $(T, \chi)$,  $\{ \lambda_{\pi, \ell} \}$ and $\gamma$.

Consider now a graph $\hat{G}$, with star decomposition $(\hat{T}, \hat{\chi})$ with central bag $\hat{b}$, leaf labeling $\{ \hat{\lambda}_{\hat{\pi}, \hat{\ell}} \}$ of  $(\hat{T}, \hat{\chi})$ and permutation $\hat{\gamma} : \hat{\chi}(\hat{b}) \rightarrow [|\hat{\chi}(\hat{b})|]$.
Suppose that there exists an isomorphism $\phi$ from $G$ to $\hat{G}$ that respects $(T, \chi)$,  $\{ \lambda_{\pi, \ell} \}$ and $\gamma$.
Specifically we assume that there exists an isomorphism $\phi_T$ from $T$ to $\hat{T}$ such that the following holds:
\begin{itemize}\setlength\itemsep{-.7pt}
\item For every $u \in V(G)$ and every $q \in V(T)$ we have $u \in \chi(q)$ if and only if $\phi(u) \in \phi_T(q)$. 
\item For every $u \in \chi(b)$ we have that $\gamma(u) = \hat{\gamma}\phi(u)$
\item For every leaf $\ell$ of $T$ and every permutation $\pi \colon \sigma(\ell) \rightarrow [|\sigma(\ell)|]$, let $\hat{\ell} = \phi_T(\ell)$ and let $\hat{\pi} \colon \hat{\sigma}(\hat{\ell}) \rightarrow [|\hat{\sigma}(\hat{\ell})|]$
be defined as $\hat{\pi}(u) = \pi\phi^{-1}(u)$ for every $u \in \hat{\chi}(\hat{\sigma}(\hat{\ell}))$.
Then $\hat{\lambda}_{\hat{\pi}, \hat{\ell}}^{-1}\lambda_{\pi, \ell}$ is an isomorphism from $G[\chi(\ell)]$  to $\hat{G}[\hat{\chi}(\hat{\ell})]$.
\end{itemize}

Let $\hat{\lambda}$ be the labeling produced by the algorithm on input $\hat{G}$, $(\hat{T}, \hat{\chi})$, $\{ \hat{\lambda}_{\pi, \hat{\ell}} \}$ when considering the permutation $\hat{\gamma}$. To prove that $\lambda$ is weakly isomorphism invariant it suffices to show that $\hat{\lambda}^{-1}\lambda$ is an isomorphism from $G$ to $\hat{G}$.
We have that  $\hat{\lambda}^{-1}\lambda$ coincides with $\hat{\gamma}^{-1}\gamma$, which in turn coincides with $\phi$, and therefore $\hat{\lambda}^{-1}\lambda$ is an isomorphism from $G[\chi(b)]$ to $\hat{G}[\hat{\chi}(\hat{b})]$.

Let $\ell$ be an arbitrarily chosen leaf of $T$ and $p$ be the position of $\ell$ in the sorted list of leaves produced by the algorithm. Let $\hat{\ell}$ be the leaf in position $p$ in the sorted list of the leaves of $\hat{T}$. To prove that  $\hat{\lambda}^{-1}\lambda$ is an isomorphism from $G$ to $\hat{G}$ it is sufficient to prove that $\hat{\lambda}^{-1}\lambda$ is an isomorphism from $G[\chi(\ell)]$ to $\hat{G}[\hat{\chi}(\hat{\ell})]$.

Since $\phi$ respects $\gamma$ we have that for every leaf $\ell'$ in $T$,  $\gamma \sigma (\ell') = \hat{\gamma}\hat{\sigma}\phi_T(\ell')$. 
Further, since $\phi$ is an isomorphism and the labelings $\{ \hat{\lambda}_{\hat{\pi}, \hat{\ell}} \}$ are locally weakly isomorphism invariant it follows that for every leaf $\ell'$ in $T$ the graph $G[\chi(\ell')]$ with labeling $\lambda^{\ell'}$ and the graph $\hat{G}[\hat{\chi}(\phi_T(\ell'))]$ with labeling $\hat{\lambda}^{\phi_T(\ell')}$ are equal.

Let now $\ell' = \phi_T^{-1}(\hat{\ell})$. Suppose $\gamma(\sigma(\ell'))$ is lexicographically larger than $\gamma(\sigma(\ell))$, or $\gamma(\sigma(\ell')) = \gamma(\sigma(\ell))$ and the graph $G[\chi(\ell')]$ with labeling $\lambda^{\ell'}$ is lexicographically larger than $G[\chi(\ell)]$ with labeling $\lambda^{\ell}$. Then, for every leaf $\ell''$ at position $p$ or smaller in the ordering of the leaves of $T$ we have that $\phi_T(\ell'')$ comes before $\phi_T(\ell') = \hat{\ell}$ in the ordering of the leaves of $\hat{T}$. But this contradicts that $\hat{\ell}$ is at position $p$ in this ordering. Similarly if $\gamma(\sigma(\ell'))$ is lexicographically smaller than $\gamma(\sigma(\ell))$, or $\gamma(\sigma(\ell')) = \gamma(\sigma(\ell))$ and the graph $G[\chi(\ell')]$ with labeling $\lambda^{\ell'}$ is lexicographically smaller than $G[\chi(\ell)]$ with labeling $\lambda^{\ell}$ then for every leaf $\ell''$ at position $p$ or larger in the ordering of the leaves of $T$ we have that $\phi_T(\ell'')$ comes before $\phi_T(\ell') = \hat{\ell}$ in the ordering of the leaves of $\hat{T}$. Again this  contradicts that $\hat{\ell}$ is at position $p$ in this ordering. We conclude that $\gamma(\sigma(\ell')) = \gamma(\sigma(\ell))$ and that $G[\chi(\ell')]$ with labeling $\lambda^{\ell'}$ is equal to $G[\chi(\ell)]$ with labeling $\lambda^{\ell}$. 

Hence (by choice of $\ell' = \phi_T^{-1}(\hat{\ell})$) we have that $\hat{G}[\hat{\chi}(\hat{\ell})]$ with labeling $\hat{\lambda}^{\hat{\ell}}$ is equal to $G[\chi(\ell')]$ with labeling $\lambda^{\ell'}$ and therefore also to $G[\chi(\ell)]$ with labeling $\lambda^{\ell}$. 
Thus $(\hat{\lambda}^{\hat{\ell}})^{-1}\lambda^{\ell}$ is an isomorphism from $G[\chi(\ell)]$ to $\hat{G}[\hat{\chi}(\hat{\ell})]$. 
However, for vertices in $\sigma(\ell)$ we have $(\hat{\lambda}^{\hat{\ell}})^{-1}\lambda^{\ell}(v) = \hat{\gamma}^{-1}\gamma(v) = \hat{\lambda}^{-1}\lambda(v)$. On the other hand both $\ell$ and $\hat{\ell}$ are at position $p$ in their respective orderings. Therefore, for vertices in $\chi(\ell) \setminus \sigma(\ell)$ we have that $\lambda(v) = \lambda^{\ell}(v) + ((n+1) \cdot p)$, while for vertices in  $\hat{\chi}(\hat{\ell}) \setminus \hat{\sigma}(\hat{\ell})$ we have that $\hat{\lambda}(v) = \hat{\lambda}^{\hat{\ell}}(v) + ((n+1) \cdot p)$. Thus for vertices in $\chi(\ell) \setminus \sigma(\ell)$  it holds that $(\hat{\lambda}^{\hat{\ell}})^{-1}\lambda^{\ell}(v) =  \hat{\lambda}^{-1}\lambda(v)$. And so we can conclude that $\hat{\lambda}^{-1}\lambda(v)$  is an isomorphism from $G[\chi(\ell)]$ to $\hat{G}[\hat{\chi}(\hat{\ell})]$, completing the proof. 
\end{proof}

\begin{lemma}\label{lem:smallCentersAlgorithm}
There exists an algorithm that takes as input a compactly colored graph $G$, a vertex set $D \subseteq V(G)$, an isomorphism invariant (with respect to $G$ and $D$) star decomposition $(T, \chi)$ of $G$ with central bag $b$ such that $D \subseteq \chi(b)$, and a compact locally weakly isomorphism invariant leaf-labeling $\{ \lambda_{\pi, \ell} \}$ of $(T, \chi)$. The algorithm runs in time $(|\chi(b)|)! \cdot n^{\cO(1)}$ and outputs for every permutation $\iota \colon D \rightarrow [|D|]$ a canonical labeling of $(G, \iota)$.
\end{lemma}

\begin{proof}
The algorithm applies Lemma~\ref{lem:smallCentersAlgorithmWrapper}
and obtains for every permutation $\iota \colon D \rightarrow [|D|]$ a labeling $\lambda$ of $(G, \iota)$ which is weakly isomorphism invariant in $(G, \iota)$ and $(T, \chi)$. Since  $(T, \chi)$ is isomorphism invariant in $(G, \iota)$, $\lambda$ is weakly isomorphism invariant in $(G, \iota)$.
\end{proof}

%% file: unpumping.tex
We now define the {\em unpumping} operation which is at the heart of our recursive understanding procedure. The intuition behind the name ``unpumping'' is that if we think of gluing graphs as string concatenation, then the implementation of the operation mimics the classic pumping lemma from formal languages, except that we ''pump`` to make the graph smaller, rather than larger. Hence {\em un}pumping, rather than pumping.
The operation has two versions --- the ``strong'' and the ``weak'' --- depending on what type of representatives are used in the replacement. 
The unpumping operation takes as input a compactly colored graph $G$, a vertex set $D \subseteq V(G)$, a connectivity sensitive and isomorphism invariant (in $G$ and $D$) star decomposition $(T, \chi)$ of $G$ with central bag $b \in V(T)$, such that $D \subseteq \chi(b)$, and a compact locally weakly isomorphism invariant leaf-labeling $\{\lambda_{\pi, \ell}\}$ of $(T, \chi)$.
The procedure also depends on parameters $k$ and $h$ that are part of the input.
The parameter $k$ is only used for strong unpumping.
The operation produces a colored graph $G^ \star$ and a connectivity sensitive star decomposition $(T, \chi^\star)$ as follows.

The graph $G^\star$ is initialized as $G$. Then the unpumping procedure iterates over every leaf $\ell$ of $T$, selects the permutation $\pi\langle\ell\rangle$ and labeling $\lambda \langle \ell \rangle$, and computes using Lemma~\ref{lem:computeStrongRep} or~\ref{lem:computeWeakRep} the ($(k,h)$-strong or $h$-weak, depending on the type of unpumping) representative $(G_R^\ell, \pi_R^\ell)$ of $(G[\chi(\ell)], \pi\langle\ell\rangle)$.
Then $G^\star$ is updated to:
$$G^\star \leftarrow (G_R^\ell, \pi_R^\ell) \oplus (G^\star \setminus (\chi(\ell) \setminus \sigma(\ell)), \pi\langle\ell\rangle)$$
We say that the unpumping {\em replaces} every leaf of $T$ by its representative. 

The star decomposition $(T, \chi^\star)$ of $G^\star$ has the same decomposition star $T$ as $G$, for the central bag $b$ we set $\chi^\star(b) = \chi(b)$ and for every leaf $\ell$ of $T$ we set $\chi^\star(\ell)$ to be the vertex set of the representative $G_R^\ell$ used to replace $G[\chi(\ell)]$.

The unpumping procedure then modifies the coloring of $G^\star$ by going over every leaf $\ell$ in $T$. Let $r$ be the rank of $\ell$ in $T$ (with respect to $G$). The procedure adds $n \cdot (r+1)$ to the color of every vertex $v \in \chi^\star(\ell) \setminus \chi^\star(b)$. Thus, if the color of $v$ used to be $x$ it is changed to $x + n(r+1)$. This concludes the description of the unpumping procedure. We call $G^\star$ the {\em $(k,h)$-strongly unpumped version of $G$}, or the {\em $h$-weakly unpumped version of $G$} depending on whether strong or weak representatives are used in the replacement. 
We collect a few observations about the unpumped version $G^\star$.

\begin{lemma}\label{lem:unpumpProperties} The unpumped graph $G^\star$ and star decomposition $(T, \chi^\star)$ has the following properties. 
\begin{enumerate}\setlength\itemsep{-3pt}

\item\label{itm:gstarCenterUnchanged} $(G^\star[\chi^\star(b)], D) = (G[\chi(b)], D)$, in terms of equality of colored, labeled graphs. 

\item\label{itm:gstarfolio} For every labeling $\iota : D \rightarrow [|D|]$ the $h$-folio of $(G, \iota)$ and the $h$-folio of $(G^\star, \iota)$ are equal. 

\item\label{itm:colorRank} In $G^\star$ every vertex in $\chi^\star(b)$ has a color between $1$ and $n$. For every leaf $\ell$ of $T$, every vertex in $\chi^\star(\ell) \setminus \chi^\star(b)$ has a color between $n(r+1) + 1$ and $n(r+2)$, where $r$ is the rank of $\ell$.

\item\label{itm:sameExtPerm} For every leaf $\ell$ of $T$ and $\iota : \sigma(\ell) \rightarrow [|\sigma(\ell)|]$, the boundaried graphs $(G[\chi(\ell)], \iota)$ and $(G^\star[\chi(\ell)], \iota)$ have the same set of extendable permutations.

\item\label{itm:sameImprovedCenter} For strong unpumping: For every $\iota : D \rightarrow [|D|]$, and every boundaried graph $(H, \iota)$ boundary-consistent with $(G, \iota)$, the subgraph of the $k$-improved graph of $(G, \iota) \oplus (H, \iota)$ induced by $\chi(b)$ and the subgraph of the $k$-improved graph of $(G^\star, \iota) \oplus (H, \iota)$ induced by $\chi^\star(b)$ are the same. 

\item\label{itm:unpumpKeepsUnbreak} For weak unpumping: For every function $q$, if $\chi(b)$ is $q$-unbreakable in $G$ then $\chi^\star(b)$ is $q$-unbreakable in $G^\star$. 

\item\label{itm:weakCutSignature} For weak unpumping: for every $\iota : D \rightarrow [|D|]$, the weak cut signature of $(G, \iota)$ is equal to the weak cut signature of $(G^\star, \iota)$.

\item\label{itm:connSensitiveGstar} The decomposition $(T, \chi^\star)$ is connectivity-sensitive.

\end{enumerate}
\end{lemma}

\begin{proof}
Property~\ref{itm:gstarCenterUnchanged} follows directly from the construction, since only the leaves are replaced and every leaf $G[\chi(\ell)]$ is replaced by a representative, which is boundary-consistent with $G[\chi(\ell)]$. 
Property~\ref{itm:gstarfolio} follows directly by repeated applications of Lemma~\ref{lem:topFolioGlueAndForget}.
Property~\ref{itm:colorRank} follows directly from the definition of unpumping together with the assumption that $G$ is compactly colored. 
Property~\ref{itm:sameExtPerm} follows from Lemma~\ref{lem:equivalentStrongGadget} (for strong unpumping) and Lemma~\ref{lem:equivalentWeakGadget} for weak unpumping. 

Property~\ref{itm:sameImprovedCenter} follows from repeated applications of Lemma~\ref{lem:sameImproved}.
In particular, let $\iota : D \rightarrow [|D|]$ be a bijection and $(H, \iota)$ be boundary-consistent with $(G, \iota)$.
Consider the step in the unpumping algorithm where $(G^\star([\chi(\ell)], \pi\langle \ell \rangle)$ is replaced by $(G_R^\ell, \pi_R^\ell)$.
Consider the graph of $F = (G^\star[\chi(\ell)], \iota) \oplus (H, \iota)$. Strictly speaking $F$ is a boundaried graph with boundary $D$, we consider only the underlying un-boundaried graph. 
We now consider $(G^\star[\chi(\ell)], \pi\langle\ell\rangle)$ and $(F \setminus (\chi(\ell) \setminus \sigma(\ell)), \pi\langle\ell\rangle)$ as boundaried graphs with boundary $\sigma(\ell)$.
We have that $F = (G^\star[\chi(\ell)], \lambda\langle\ell\rangle) \oplus (F \setminus (\chi(\ell) \setminus \sigma(\ell)), \pi\langle\ell\rangle)$.
When the unpumping replaces $(G^\star([\chi(\ell)], \pi\langle \ell \rangle)$ by $(G_R^\ell, \pi_R^\ell)$, the graph $F$ changes from
$(G^\star[\chi(\ell)], \lambda\langle\ell\rangle) \oplus (F \setminus (\chi(\ell) \setminus \sigma(\ell)), \pi\langle\ell\rangle)$.
to
$(G_R^\ell, \pi_R^\ell) \oplus (F \setminus (\chi(\ell) \setminus \sigma(\ell)), \pi\langle\ell\rangle)$.
By Lemma~\ref{lem:sameImproved} the $\imp{F}{k}[\chi(\ell)]$ remains unchanged. Since this holds for the replacement of every leaf $\ell$, this concludes the proof of Property~\ref{itm:sameImprovedCenter}.

Property~\ref{itm:unpumpKeepsUnbreak} follows from repeated applications of Lemma~\ref{lem:sameUnbreakable} in much the same way Property~\ref{itm:sameImprovedCenter} follows from repeated applications of Lemma~\ref{lem:sameImproved}.
Consider the step in the unpumping algorithm where $(G^\star([\chi(\ell)], \pi\langle \ell \rangle)$ is replaced by $(G_R^\ell, \pi_R^\ell)$.
In this step $G^\star$ changes from
$(G^\star([\chi(\ell)], \pi\langle \ell \rangle) \oplus 
(G^\star \setminus (\chi(\ell) \setminus \sigma(\ell)), \pi\langle\ell\rangle)$
to
$G^\star \leftarrow (G_R^\ell, \pi_R^\ell) \oplus (G^\star \setminus (\chi(\ell) \setminus \sigma(\ell))$.
Since $(G^\star([\chi(\ell)], \pi\langle \ell \rangle)$ and $(G_R^\ell, \pi_R^\ell)$ have the same weak cut signature, Lemma~\ref{lem:sameUnbreakable} implies that the family of $q$-unbreakable sets in $G^\star \setminus (\chi(\ell) \setminus \sigma(\ell))$ remains unchanged.

For Property~\ref{itm:weakCutSignature}, observe that repeated applications of Lemma~\ref{lem:sameWeakCutAfterGluing} (in exactly the same way Property~\ref{itm:unpumpKeepsUnbreak} follows from repeated applications of Lemma~\ref{lem:sameUnbreakable}) yields that for every $X, Y \subseteq \chi(b)$, the minimum order of an $X$-$Y$ separation in $G$ and $G^\star$ are the same. Applying this observation to all pairs $X$, $Y$ such that $X \cup Y = D$ yields that the weak cut signatures of  $G$ and $G^\star$ are the same, proving Property~\ref{itm:weakCutSignature}.

Finally, for Property~\ref{itm:connSensitiveGstar}, observe that $(T, \chi)$ is is connectivity sensitive. Thus every leaf $\ell$ of $T$ and every component $C$ of $G[\chi(\ell) \setminus \sigma(\ell)]$ satisfies $N(C) = \sigma(\ell)$. Then, by Lemma~\ref{lem:equivalentStrongGadget} or~\ref{lem:equivalentWeakGadget} (for strong or weak unpumping respectively), every 
leaf $\ell$ of $T$ and every component $C$ of $G^\star[\chi^\star(\ell) \setminus \sigma^\star(\ell)]$ satisfies $N(C) = \sigma^\star(\ell)$. Therefore $(T, \chi^\star)$ is a connectivity sensitive star decomposition of $(G, \chi^\star)$, proving Property~\ref{itm:connSensitiveGstar}.
\end{proof}

The next lemma essentially states that $G^\star$ is weakly isomorphism invariant in $G$. However we need a stronger version - morally we would like to say that for every bijection $\iota \colon D \rightarrow [|D|]$ the boundaried graph $(G^\star, \iota)$ is weakly isomorphism invariant in $(G, \iota)$. However, since the computation of $G^\star$ does {\em not} take any labeling of $D$ as input, saying that $(G^\star, \iota)$ is weakly isomorphism invariant in $(G, \iota)$ is not a well-formed statement. This leads to a rather cumbersome lemma statement that essentially ``unpacks'' the definition of weak isomorphism invariance and adapts it to the setting where the map remains weakly isomorphism invariant for every $\iota : D \rightarrow [|D|]$.

\begin{lemma}\label{lem:weakIsoStar} 
Let $\hat{G}$ be a graph and $\hat{D} \subseteq V(\hat{G})$ be a vertex subset such that $(G, D)$ and $(\hat{G}, \hat{D})$ are isomorphic. 
Let $(\hat{T}, \hat{\chi})$ be a connectivity sensitive, isomorphism invariant (in $\hat{G}$ and $\hat{D}$) star decomposition of $\hat{G}$ with central bag $\hat{b} \in V(\hat{T})$ and let $\{\hat{\lambda}_{\hat{\pi}, \hat{\ell}}\}$ be a compact locally weakly isomorphism invariant leaf-labeling of $(\hat{T}, \hat{\chi})$.
Finally, let $\hat{G}^\star$ be the $(k,h)$-strongly unpumped version of $\hat{G}$  (or the $h$-weakly unpumped version of $\hat{G}$), and let $(\hat{T}^\star, \hat{\chi}^\star)$ be the corresponding star decomposition of $\hat{G}^\star$.
Then, for every pair $\iota \colon D \rightarrow [|D|]$ and $\hat{\iota} \colon \hat{D} \rightarrow [|\hat{D}|]$ of bijections, if the boundaried graphs $(G, \iota)$ and $(\hat{G}, \hat{\iota})$ are isomorphic then there is an isomorphism from $(G^\star, \iota)$ to $(\hat{G}^\star, \hat{\iota})$ that maps $(T^\star, \chi^\star)$ to $(\hat{T}^\star, \hat{\chi}^\star)$.
\end{lemma}

\begin{proof}
Let $\psi : V(G) \rightarrow V(\hat{G})$ be an isomorphism from $(G, \iota)$ to $(\hat{G}, \hat{\iota})$.
Since $(T, \chi)$ is isomorphism invariant in $(G, D)$ (and $(\hat{T}, \hat{\chi})$ is isomorphism invariant in $(\hat{G}, \hat{D})$) there is an isomorphism $\psi_T : V(T) \rightarrow V(\hat{T})$ from $T$ to $\hat{T}$ such that $\psi_T(b) = \hat{b}$ and for every $v \in V(G)$ and node $u \in V(T)$ it holds that $v \in \chi(u)$ if and only if $\psi(v) \in \hat{\chi}(\psi_T(u))$.
Let $\ell$ be a leaf of $T$ and $\hat{\ell} = \psi_T(\ell)$.
We have that $\psi$ restricted to $\chi(\ell)$ is an isomorphism from $G[\chi(\ell)]$ to $\hat{G}[\hat{\chi}(\hat{\ell})]$ that maps $\sigma(\ell)$ to $\hat{\sigma}(\hat{\ell})$. It follows from Lemma~\ref{lem:equalRankIso} (and the isomorphism of $(G,D)$ with $(\hat{G}, \hat{D})$) that the rank of $\ell$ and $\hat{\ell}$ are equal, and that
the boundaried graph $(G[\chi(\ell)], \pi\langle\ell\rangle)$ with labeling $\lambda\langle\ell\rangle$
and the boundaried graph $(\hat{G}[\hat{\chi}(\hat{\ell})], \hat{\pi}\langle\hat{\ell}\rangle)$ with labeling $\hat{\lambda}\langle\hat{\ell}\rangle$ are equal.

Since unpumping replaces $(G[\chi(\ell)], \pi\langle\ell\rangle)$ and $(\hat{G}[\hat{\chi}(\hat{\ell})], \hat{\pi}\langle\hat{\ell}\rangle)$ by their representatives, Lemma~\ref{lem:replacementKeepsPerm} applied to  $(G[\chi(\ell)], \pi\langle\ell\rangle)$, $(\hat{G}[\hat{\chi}(\hat{\ell})], \hat{\pi}\langle\hat{\ell}\rangle)$ and the restriction of $\psi$ to $\chi(\ell)$ implies that there exists an isomorphism $\psi_\ell^R$ from $G^\star[\chi^\star(\ell)]$ to $\hat{G}^\star[\hat{\chi}^\star(\hat{\ell})]$ such that $\psi(v) = \psi_\ell^R(v)$ for every $v \in \sigma(\ell)$. For every leaf $\ell$ of $T$, let  $\psi_\ell^R$ be such an isomorphism. 
But then, since every vertex $v$ in $G^\star$ is either in $\chi^\star(b)$ or in $\chi^\star(\ell) \setminus \sigma(\ell)$ for precisely one leaf $\ell$ of $T$, the mapping $\psi^\star : V(G^\star) \rightarrow V(\hat{G}^\star)$ defined as 
$$
\psi^\star(v) =
\begin{cases}
\psi(v), \mbox{ if } v \in \chi^\star(b) \\
\psi_\ell^\star(v) \mbox{ if } v \in \chi^\star(\ell) \setminus \sigma(\ell)
\end{cases}
$$
is an isomorphism from $(G^\star, \iota)$ to $(\hat{G}^\star, \hat{\iota})$ that maps $(T^\star, \chi^\star)$ to $(\hat{T}^\star, \hat{\chi}^\star)$.
\end{proof}

\paragraph{Lifting Canonical Labelings of $G^\star$ (Specification).}
We will shortly describe an algorithm, that we call the {\em lifting procedure}. We first specify what it takes as input, and what it produces as output. 

The lifting procedure takes as input 
a compactly colored graph $G$,
a vertex set $D \subseteq V(G)$,
a connectivity sensitive and isomorphism invariant (in $G$ and $D$) star decomposition $(T, \chi)$ of $G$ with central bag $b \in V(T)$ such that $D \subseteq \chi(b)$, 
a compact locally weakly isomorphism invariant leaf-labeling $\{\lambda_{\pi, \ell}\}$ of $(T, \chi)$, together with
the (weakly or strongly) unpumped version $G^\star$ of $G$ and the corresponding star decomposition $(T, \chi^\star)$ of $G^\star$, 
a bijection $\iota : D \rightarrow [|D|]$, and a
compact weakly isomorphism invariant proper labeling $\lambda^\star$ of the boundaried graph $(G^\star, \iota)$.
The algorithm outputs a proper labeling $\lambda$ of $(G, \iota)$. We will call an input $(G, D, (T, \chi), b, \{\lambda_{\pi, \ell}\}, G^\star, (T, \chi^\star), \iota, \lambda^\star)$ to the lifting procedure adhering to the above specifications a {\em lifting bundle}.

\paragraph{Lifting Canonical Labelings of $G^\star$ (Description).}
The lifting procedure works as follows. For each leaf $\ell$ define $\pi^\ell$ (not to be confused with $\pi\langle\ell\rangle$) to be the unique permutation of $\sigma(\ell)$ such that $\pi(u) < \pi(v)$ if and only if $\lambda^\star(u) < \lambda^\star(v)$. We define $\lambda^\ell$ to be $\lambda_{\pi^\ell, \ell}$.
The algorithm sorts the leaves $\ell$ of $T$ according to the lexicographic ordering of the set of labels $\lambda^\star( \sigma^\star(\ell) )$. Because $(T, \chi)$ is connectivity-sensitive, this yields a total order on the leaves of $T$.
The labeling $\lambda$ is defined as follows. For every vertex $v \in \chi(b)$ where $b$ is the central bag of $T$ we set $\lambda(v) = \lambda^\star(v)$. For every leaf $\ell$ and every vertex $v$ in $\chi(\ell) \setminus \chi(b)$ we set $\lambda(v) = \lambda^ \ell(v) + ((n+1) \cdot p)$ where $p$ is the position of $\ell$ in the sorted order of leaves defined above.  The algorithm outputs a compactification of the labeling $\lambda$ of $(G, \iota)$. It is clear from the definition of the lifting procedure that it runs in time polynomial in the size of the objects given to it as input. 

\begin{observation}\label{obs:liftRunTime}
The lifting procedure runs in time polynomial in its input. 
\end{observation}

Our next goal is to show that the labeling $\lambda$ computed by the lifting procedure is in fact weakly isomorphism invariant with respect to $(G, \iota)$. The next lemma is the main step towards that goal. 

\begin{lemma}
\label{lem:liftingIsoHelper}
Let 
$(G, D, (T, \chi), b, \{\lambda_{\pi, \ell}\}, G^\star, (T, \chi^\star), \iota, \lambda^\star)$ and 
$(\hat{G}, \hat{D}, (\hat{T}, \hat{\chi}), \hat{b}, \{\hat{\lambda}_{\ell,\pi}\}, \hat{G}^\star, (\hat{T}, \hat{\chi}^\star), \hat{\iota}, \hat{\lambda}^\star)$ be lifting bundles, and let $\lambda$ and $\hat{\lambda}$ be the output of the lifting procedure on them. 
If $(G, D)$ is isomorphic with $(\hat{G}, \hat{D})$, and $(G^\star, \iota)$ is isomorphic with $(\hat{G}^\star, \hat{\iota})$ then $\hat{\lambda}^{(-1)}\lambda$ is an isomorphism from $(G, \iota)$ to $(\hat{G}, \hat{\iota})$.
\end{lemma}

\begin{proof}
Since $\lambda^\star$ and $\hat{\lambda}^\star$ are weakly isomorphism invariant proper labelings of $(G^\star, \iota)$ and $(\hat{G}^\star, \hat{\iota})$ respectively it follows that  $\phi^\star = (\hat{\lambda}^\star)^{(-1)}\lambda^\star$ is an isomorphism from $(G^\star, \iota)$ to $(\hat{G}^\star, \hat{\iota})$.
By Property~\ref{itm:colorRank} of Lemma~\ref{lem:unpumpProperties} every vertex of $\chi^\star(b)$ is mapped by $\phi^\star$ to a vertex of  $\hat{\chi}^\star(\hat{b})$, where $\hat{b}$ is the center of $\hat{T}$. Therefore, the restriction of $\hat{\lambda}^{(-1)}\lambda$ to $\chi(b)$, which is equal to the restriction of $\phi^\star$ to $\chi^\star(b)$ (which is equal to $\chi(b)$) is an isomorphism from $(G[\chi(b)], \iota)$ to $(\hat{G}[\hat{\chi}(\hat{b})], \hat{\iota})$. We now turn our attention to the leaves of $T$.
Since $(T, \chi^\star)$ and $(\hat{T}, \hat{\chi}^\star)$ are both connectivity sensitive it follows by Observation~\ref{obs:uniqueStarDec} that there exists an isomorphism $\phi_T^\star$ from $T$ to $\hat{T}$ so that $v \in \chi^\star(q)$ if and only if $\phi^\star(v) \in \hat{\chi}^\star(\phi_T^\star(q))$.
Let $\ell$ be a leaf of $T$ and  $\hat{\ell} = \phi_T^\star(\ell)$. Note that the sets of labels $\lambda^\star( \sigma^\star(\ell))$ and $\hat{\lambda}^\star( \hat{\sigma}^\star(\hat{\ell}))$ are the same. Thus the position $p$ of $\ell$ in the sorted order of leaves in $T$ and the position $\hat{p}$ of $\hat{\ell}$ in the sorted order of leaves in $\hat{T}$ are the same. To complete the proof of the lemma it suffices to show that $\hat{\lambda}^{(-1)}\lambda$ restricted to $\chi(\ell)$ is an isomorphism from $G[\ell]$ to $\hat{G}[\hat{\ell}]$. 

By Property~\ref{itm:colorRank} of Lemma~\ref{lem:unpumpProperties}, the ranks of $\ell$ and $\hat{\ell}$ are equal. The remainder of the proof considers the relationships between the following four boundaried graphs:
\begin{itemize} \setlength\itemsep{-3pt}
\item $(G[\chi(\ell)], \pi\langle\ell\rangle )$,
\item $(G^\star[\chi^\star(\ell)], \pi\langle\ell\rangle)$,
\item $(\hat{G}[\hat{\chi}(\hat{\ell})], \hat{\pi}\langle\hat{\ell}\rangle )$,
\item $(\hat{G}^\star[\hat{\chi}^\star(\hat{\ell})], \hat{\pi}\langle\hat{\ell}\rangle )$.
\end{itemize}

Since $(G, D)$ and $(\hat{G}, \hat{D})$ are isomorphic and $(T^\star, \chi^\star)$ and $(\hat{T}^\star, \hat{\chi}^\star)$ are isomorphism invariant in $G$, $D$ and $\hat{G}, \hat{D}$ respectively, Lemma~\ref{lem:equalRankIso} yields that $(\hat{\lambda}\langle\hat{\ell}\rangle)^{-1}\lambda\langle\ell\rangle$ is an isomorphism from $G[\chi(\ell)]$ to  $\hat{G}[\hat{\chi}(\hat{\ell})]$. 
The construction of $G^\star$ replaced $G[\chi(\ell)]$ by $G^\star[\chi^\star(\ell)]$, and the replacement is isomorphism invariant with respect to $G[\chi(\ell)]$ and $\lambda\langle\ell\rangle$ (since every deterministic algorithm is isomorphism invariant when the input is a properly labeled graph). It follows that there exists an isomorphism $\psi^\star$ from $G^\star[\chi^\star(\ell)]$ to $\hat{G}^\star[\hat{\chi}^\star(\hat{\ell})]$ so that $\psi^\star(v) = (\hat{\lambda}\langle\hat{\ell}\rangle)^{-1}\lambda\langle\ell\rangle(v)$ for every $v \in \sigma(\ell)$.
Now recall that by definition $\pi\langle\ell\rangle$ and $\lambda\langle\ell\rangle$ coincide on the domain of $\pi\langle\ell\rangle$ (namely $\sigma(\ell)$), and similarly $\hat{\pi}\langle\hat{\ell}\rangle$ and $\hat{\lambda}\langle\hat{\ell}\rangle$ coincide on the domain of $\hat{\pi}\langle\hat{\ell}\rangle$ (namely $\hat{\sigma}(\hat{\ell})$).
Thus the restriction of $\psi^\star$ to $\sigma^\star(\ell)$ is equal to  $(\hat{\pi}\langle\hat{\ell}\rangle)^{-1}\pi\langle\ell\rangle$.
At the same time $\phi^\star$ restricted to $\chi^\star(\ell)$ is also an isomorhism from $G^\star[\chi^\star(\ell)]$ to $\hat{G}^\star[\hat{\chi}^\star(\hat{\ell})]$ that takes $\sigma^\star(\ell)$ to $\hat{\sigma}^\star(\hat{\ell})$.
It follows that $(\psi^\star)^{-1} \phi^\star$ is an automorhism of $G^\star[\chi^\star(\ell)]$.
The restriction of $(\psi^\star)^{-1}$ to $\hat{\sigma}(\hat{\ell})$ is equal to $(\pi\langle\ell\rangle)^{-1}\hat{\pi}\langle\hat{\ell}\rangle$, and therefore the automorphism $(\psi^\star)^{-1} \phi^\star$ takes each vertex $v\in \sigma(\ell)$ labeled $\pi\langle\ell\rangle(v)$ to the vertex $(\psi^\star)^{-1}\phi^\star(v)\in \sigma(\ell)$, which has label 
$\pi\langle\ell\rangle(\pi\langle\ell\rangle)^{-1}\hat{\pi}\langle\hat{\ell}\rangle\phi^\star(v) = \hat{\pi}\langle\hat{\ell}\rangle\phi^\star(v)$.
Therefore  $\hat{\pi}\langle\hat{\ell}\rangle\phi^\star(\pi\langle\ell\rangle)^{-1}$ is an extendable permutation of the boundaried graph $(G^\star[\chi^\star(\ell)], \pi\langle\ell\rangle)$. 

But $(G^\star[\chi^\star(\ell)], \pi\langle\ell\rangle)$ and $(G[\chi(\ell)], \pi\langle\ell\rangle )$ have the same set of extendable permutations, and hence $\hat{\pi}\langle\hat{\ell}\rangle\phi^\star(\pi\langle\ell\rangle)^{-1}$ is an extendable permutation also of $(G[\chi(\ell)], \pi\langle\ell\rangle )$.
Thus there exists an automorphism $\gamma$ of $G[\chi(\ell)]$ whose restriction to $\sigma(\ell)$ is 
$$(\pi\langle\ell\rangle)^{-1}\hat{\pi}\langle\hat{\ell}\rangle\phi^\star(\pi\langle\ell\rangle)^{-1}\pi\langle\ell\rangle = (\pi\langle\ell\rangle)^{-1}\hat{\pi}\langle\hat{\ell}\rangle\phi^\star$$

Hence $(\hat{\lambda}\langle\hat{\ell}\rangle)^{-1}\lambda\langle\ell\rangle \gamma$ is an isomorhism from $G[\chi(\ell)]$ to  $\hat{G}[\hat{\chi}(\hat{\ell})]$ whose restriction to $\sigma(\ell)$ equals
$$(\hat{\pi}\langle\hat{\ell}\rangle)^{-1} \pi\langle\ell\rangle (\pi\langle\ell\rangle)^{-1} \hat{\pi}\langle\hat{\ell}\rangle \phi^\star = \phi^\star = (\hat{\lambda}^\star)^{(-1)}\lambda^\star$$

From the definition of $\pi^\ell$ and $\hat{\pi}^{\hat{\ell}}$ it follows that for every vertex $v \in \sigma(\ell)$ we have that $\pi^\ell(v) = \hat{\pi}^{\hat{\ell}} (\hat{\lambda}^\star)^{-1} \lambda^\star(v)$. Therefore, $(\hat{\lambda}\langle\hat{\ell}\rangle)^{-1}\lambda\langle\ell\rangle \gamma$ is not only an isomorhism from $G[\chi(\ell)]$ to $\hat{G}[\hat{\chi}(\hat{\ell})]$, it is also in isomorhism from the boundaried graph $(G[\chi(\ell)], \pi^\ell)$ to the boundaried graph $(\hat{G}[\hat{\chi}(\hat{\ell})], \hat{\pi}^{\hat{\ell}})$.
This means that the weakly isomorphism invariant proper labelings $\lambda^\ell$ of $(G[\chi(\ell)], \pi^\ell)$ and  $\hat{\lambda}^{\hat{\ell}}$ of $(\hat{G}[\hat{\chi}(\hat{\ell})], \hat{\pi}^{\hat{\ell}})$ also yield an isomorhism $(\hat{\lambda}^{\hat{\ell}})^{(-1)}\lambda^\ell$ from $(G[\chi(\ell)], \pi^\ell)$ to $(\hat{G}[\hat{\chi}(\hat{\ell})], \hat{\pi}^{\hat{\ell}})$, and the restriction of $(\hat{\lambda}^{\hat{\ell}})^{(-1)}\lambda^\ell$ to $\sigma(\ell)$ equals $(\hat{\lambda}^\star)^{-1}\lambda^\star$.

However, for every $v \in \chi(\ell)$ we have that $\lambda(v) = \lambda^\star(v)$ if $v \in \sigma(\ell)$ and  $\lambda(v) = \lambda^\ell(v) + ((n + 1) \cdot p)$ otherewise, where $p$ is the position of $\ell$ in the sorted order of $T$. Similarly  for every $v \in \hat{\chi}(\hat{\ell})$ we have that $\hat{\lambda}(v) = \hat{\lambda}^\star(v)$ if $v \in \hat{\sigma}(\hat{\ell})$ and  $\hat{\lambda}(v) = \hat{\lambda}^{\hat{\ell}}(v) + ((n + 1) \cdot \hat{p})$ otherwise, where $\hat{p}$ is the position of $\hat{\ell}$ in the sorted order of $\hat{T}$.

Since $p = \hat{p}$ it follows that restricted to $\chi(\ell)$ it holds that $\hat{\lambda}^{-1}\lambda = (\hat{\lambda}^{\hat{\ell}})^{(-1)}\lambda^\ell$, which is an isomorhism from $G[\chi(\ell)]$ to $\hat{G}[\hat{\chi}(\hat{\ell})]$. But then $\hat{\lambda}^{-1}\lambda$ is an isomorphism from $(G, \iota)$ to $(\hat{G}, \hat{\iota})$ completing the proof. 
\end{proof}

We are now ready to summarize the results regarding the lifting procedure in the following lemma. 

\begin{lemma}[Lifting Procedure]
\label{lem:lem:caNlifting}
The lifting procedure takes as input a lifting bundle 
$$(G, D, (T, \chi), b, \{\lambda_{\pi, \ell}\}, G^\star, (T, \chi^\star), \iota, \lambda^\star),$$ 
runs in polynomial time in the size of the input and produces a weakly isomorphism invariant (with respect to $(G, \iota)$) proper labeling $\lambda$ of $(G, \iota)$.
\end{lemma} 

\begin{proof}
The running time bound follows from Observation~\ref{obs:liftRunTime}. Consider now a boundaried graph $(\hat{G}, \hat{\iota})$ isomorhpic to $(G, \iota)$ and let $(\hat{G}, \hat{D}, (\hat{T}, \hat{\chi}), \hat{b}, \{\hat{\lambda}_{\ell,\pi}\}, \hat{G}^\star, (\hat{T}, \hat{\chi}^\star), \hat{\iota}, \hat{\lambda}^\star)$ be the corresponding lifting bundle. Let $\hat{\lambda}$ be the output of the lifting procedure on this bundle. To prove the lemma it suffices to show that $\hat{\lambda}^{(-1)}\lambda$ is an isomorphism from $(G, \iota)$ to $(\hat{G}, \hat{\iota})$.

Since $(\hat{G}, \hat{\iota})$ is isomorphic to $(G, \iota)$, we have that $(\hat{G}, \hat{D})$ is isomorphic to $(G, D)$. Further, Lemma~\ref{lem:weakIsoStar} yields that $(\hat{G}^\star, \hat{\iota})$ is isomorphic to $(G^\star, \iota)$. Then Lemma~\ref{lem:liftingIsoHelper} implies that $\hat{\lambda}^{(-1)}\lambda$ is an isomorphism from $(G, \iota)$ to $(\hat{G}, \hat{\iota})$, completing the proof.
\end{proof}

On the way to proving the main properties of the lifting procedure (namely Lemma~\ref{lem:lem:caNlifting}) we also developed tools (namely Lemmas~\ref{lem:weakIsoStar} and~\ref{lem:liftingIsoHelper}) to relate the sets of extendable permutations of $(G, \iota)$ and $(G^\star, \iota)$ for any $\iota : D \rightarrow [|D|]$.

\begin{lemma}\label{lem:unpumpSameExtPerm}
Let $G$ be a graph, 
$D \subseteq V(G)$ be a vertex subset, 
$(T, \chi)$ be a connectivity sensitive, isomorphism invariant (in $(G, D)$) star decomposition of $G$ with central bag $b \in V(T)$ such that $D \subseteq \chi(b)$,
$\{\lambda_{\pi}, \ell\}$ be a compact locally weakly isomorphism invariant leaf-labeling of $(T, \chi)$,
and let $G^\star$ be the $(k,h)$-strongly unpumped version of $G$  (or the {\em $h$-weakly unpumped version of $G$}).
Then for every $\iota \colon D \rightarrow [|D|]$, the sets of extendable permutations of $(G, \iota)$ and $(G^\star, \iota)$ are the same. 
\end{lemma}

\begin{proof}
It suffices to show that for every $\hat{\iota} \colon D \rightarrow [|D|]$ we have that $(G, \iota)$ is isomorphic with $(G, \hat{\iota})$ if and only if $(G^\star, \iota)$ is isomorphic with $(G^\star, \hat{\iota})$. Suppose first that  $(G, \iota)$ is isomorphic with $(G, \hat{\iota})$. Lemma~\ref{lem:weakIsoStar} applied with $(\hat{G}, \hat{D}) = (G, D)$ yields an isomorphism from $(G^\star, \iota)$ to $(G^\star, \hat{\iota})$. For the reverse direction suppose that $(G^\star, \iota)$ is isomorphic with $(G^\star, \hat{\iota})$. Then, Lemma~\ref{lem:liftingIsoHelper} applied to the lifting bundles $(G, D, (T, \chi), b, \{\lambda_{\pi, \ell}\}, G^\star, (T, \chi^\star), \iota, \lambda^\star)$ and $(G, D, (T, \chi), b, \{\lambda_{\pi, \ell}\}, G^\star, (T, \chi^\star), \hat{\iota,} \lambda^\star)$ (the only difference between the two lifting bundles is that the first one has $\iota$, the second has $\hat{\iota}$) yields that there is an isomorphism from $(G, \iota)$ to $(G, \hat{\iota})$. This concludes the proof. 
\end{proof}

%% file: imp-clique-unbreakable.tex
In this section, we design an algorithm for {\sc Canonization} in graphs, excluding a fixed family $\cal F$ of graphs as topological minors, assuming  we have such an algorithm for 
graphs that, in addition to excluding $\cal F$, are also 
$(q,k,k)$-improved-clique-unbreakable. Towards this, we recall a result of Elberfeld and Schweitzer~\cite{ElberfeldS17}, who proved that for a given graph one can compute an isomorphism invariant tree decomposition, where each bag is either a clique of size at most $k$ or a $k$-atom (a maximal subgraph that does not contain clique separators of size at most $k$), and adhesions have sizes bounded by $k$; call this decomposition the {\em{ES-decomposition}}. Elberfeld and Schweitzer proposed an algorithm that computes the ES-decompistion in logspace for every fixed $k$. We revisit this algorithm and show that it can be implemented so that it runs in polynomial time, even if $k$ is a part of the input (i.e., the degree of the polynomial governing the running time is independent of $k$).

Using the ES-decomposition we prove a structural result that says that if a  connected graph $G$ is $(q,k)$-clique-unbreakable  then there exists an isomorphism invariant star decomposition $(T,\chi)$ satisfying ``some useful properties''. Conversely, if $G$ has  a  star decomposition $(T,\chi)$  with ``some useful properties'', such that for every leaf $\ell \in V(T)$ we have $|\chi(\ell)|\leq q$, then $G$ is $(2^kq,k)$-clique 
unbreakable. This structural result is used crucially for several arguments. 

Finally, in the last subsection we design the stated algorithm for {\sc Canonization}.  
The algorithm takes as input  a set ${\cal F} \in \mathbb{F}$, an integer $k$, and a colored graph $G$ 
with a distinguished set $D$ such that $|D| \leq k$ with the property that there exist a labelling $\iota \colon D \rightarrow [|D|]$ and a  $[|D|]$-boundaried graph $(H,\iota)$ boundary-consistent with $G$ so that $(G,\iota) \oplus (H,\iota)$ is ${\cal F}$-free and $D$ is a clique in the $k$-improved graph of $(G,\iota) \oplus (H,\iota)$. The algorithm either fails (i.e. outputs~$\bot$) or succeeds. If $k \geq \kappa_{\cal B}({\cal F})$, where $\kappa_{\cal B}({\cal F})$ is a threshold depending only on ${\cal F}$, then the algorithm succeeds and  outputs for every labeling $\pi \colon D \rightarrow [|D|]$ a proper labeling $\lambda_\pi$ of $(G, \pi)$ such that $\lambda_\pi(v) = \pi(v)$. The labeling $\lambda_\pi$ is weakly isomorphism invariant in $(G, \pi)$.  The algorithm itself is recursive  and consider several cases based on the next ``canonical bag''.

\subsubsection{Canonical Decomposition into Atoms} 
Let $c$ be a nonnegative integer. A {\em $c$-atom} is a graph that does not contain clique separators of size 
at most $c$. A graph $G$ is called an {\em atom} if $G$ does not have any clique separators. A maximal $c$-atom of a graph $G$ is a maximal induced subgraph $G[X]$ for some $X\subseteq  V (G)$ that is a $c$-atom. In particular, if $X\subsetneq V(G)$, then for every proper superset of $X$, $G[X]$ contains a clique separator of size at most $c$. We follow the procedure described by Elberfeld and Schweitzer~\cite{ElberfeldS17} to compute an isomorphism-invariant tree decomposition of a graph into its $c$-atoms. Their procedure runs in logspace for every fixed $c$, which implies that it runs in polynomial time with the degree of the polynomial governing the running time possibly depending on $c$. Since we cannot afford this, in what follows we repeat their arguments and show that the decomposition can be computed in polynomial time even if $c$ is given on input.

For every $c \in \Nat$ and graph~$G$, we define the graph $T_c = T_c(G)$ whose node set consists of all $c$-atoms of $G$ and all (inclusion wise) minimal
clique separators of size at
most~$c$. An edge is inserted between every $c$-atom~$G[X]$ and every minimal clique
separator $C$ with $C \subseteq X$. Next we define the bag function $\beta_c\colon V(T_c) \to 2^{V(G)}$ 
as follows. If  $t \in V(T_c)$ is identified with a $c$-atom $G[X]$, then $\beta_c(t) \coloneqq X$ and if  $t \in V(T_c)$ is identified with a  minimal clique separator $C$, then $\beta_c(t) \coloneqq C$.

\begin{proposition}{\rm \cite[Proposition 3.7]{ElberfeldS17}}
  For every $c \in \Nat$, the mapping $G \mapsto (T_{c}(G),\beta_c)$ is
  isomorphism-invariant for every $c \in \Nat$. 
\end{proposition}

The graph~$T_c(G)$ is typically not a tree. However, as stated
next, provided $G$ itself is a $(c-1)$-atom, $T_c(G)$ is a tree and, moreover,
$(T_c(G),\beta_c)$ is a tree decomposition of $G$.
We will also need the notion of a {\em center} in a tree decomposition and a tree. Let $T$ be an unrooted tree. Then a {\em center} of $T$ is a vertex $r$ in $T$ that minimizes the height of the tree $T$ if rooted at that vertex. It is well known (see e.g.~\cite{ElberfeldS17}) that every tree has at least one and at most two center vertices, and if there are two center vertices then they are adjacent. A center of a tree decomposition is a center of its decomposition tree. 

\begin{lemma}{\rm \cite[Lemma 3.8]{ElberfeldS17}}
  \label{lem:c-atom-tree}
  For every positive $c \in \Nat$ and a $(c-1)$-atom $G$,
  $(T_c(G),\beta_c)$ is a tree decomposition for
  $G$. Moreover, $T_c(G)$ has a unique center.
  Furthermore, given $G$ and $c$, we can construct 
  $(T_{c}(G),\beta_c)$ in polynomial time. 
\end{lemma}
 
To construct $(T_{c}(G),\beta_c)$ in polynomial time, all we need is an algorithm to compute all minimal clique separators of size at most~$c$ in $G$ and a way to compute all maximal $c$-atoms of $G$. It is known that in any graph $G$ and any minimal triangulation 
$H$ of $G$ (that is, a minimal chordal completion of $G$), the minimal clique  separators of $G$ are exactly the minimal separators of $H$ that are cliques in $G$~\cite{BerryPS10}. Since a chordal graph on $n$ vertices has at most $n$ minimal separators~\cite{brandstadt1999graph}, and they can be enumerated in polynomial time~\cite{BerryBC00,KloksK94}, we can find all minimal clique separators of size at most $c$ in $G$ by taking any minimal triangulation $H$ and listing all minimal separators of $H$ that are cliques in $G$ and have size at most $c$. Next, we describe how we can compute all maximal $c$-atoms of $G$. Given $G$ and the family $\cal C$ comprising all minimal clique separators of size at most $c$, we construct an auxiliary graph $\widehat{G}$ from $G$ as follows. Add an edge between any pair of non-adjacent vertices, say $u$ and $v$, provided there is no clique separator of size at most $c$ that separates them. It is known that $G[A]$ is a maximal $c$-atom in $G$ if and only if $\widehat{G}[A]$ is a maximal $c$-atom in $\widehat{G}$~\cite[Lemma 3.6]{ElberfeldS17}. Furthermore, $\widehat{G}$ is chordal~\cite{ElberfeldS17}, so the maximal $c$-atoms of $\widehat{G}$ are precisely the maximal cliques of this $\widehat{G}$. A chordal graph on $n$ vertices has at most $n$ maximal cliques and they can be enumerated in polynomial time~\cite{BerryPS10}, so in this way we can also enumerate all maximal $c$-atoms in $G$ in polynomial time. All in all, we can construct $(T_c(G),\beta_c)$ in polynomial time. 

From now on, we assume that the tree decomposition $(T_c(G),\beta_c)$ is rooted at its (unique) center.

 \begin{lemma}{\rm \cite[Lemma 3.9]{ElberfeldS17}}
  \label{lem:c-to-c++-atom-decomposition} 
There exists a polynomial time algorithm that, given $d,c \in \Nat$, with~$d \leq c$,  and  a
  $d$-atom~$G$, outputs an isomorphism-invariant rooted tree decomposition $(T,\beta)$ and a partition of 
 $V(T)$ into two sets $A$ and $C$ such that  the following holds: 
  \begin{enumerate}  \setlength{\itemsep}{-1pt}
  \item There is no edge between any pair of nodes in $A$. 
 \item For every  $t\in C$, $\beta(t)$ is a clique of size at most $c$.
 \item For every  $t\in A$,  $\beta(t)$ is a maximal $c$-atom of $G$. 
 \item If $s \in C$, $t \in A$ and $st \in E(T)$ then $\beta(s) \subseteq \beta(t)$. 
 \item Each adhesion of $(T,\beta)$ is a clique of size at most $c$.  
 \end{enumerate}
\end{lemma}
\begin{proof}
  We show the lemma by induction on~$c-d$. If~$c-d = 0$, then the
  $c$-atom $G$ is the unique bag of the tree decomposition $T$, which
  satisfies all requirements of the lemma. If~$c-d > 0$, we construct
  and prove the correctness of the constructed tree decomposition as
  follows.

\medskip

\noindent 
{\bf (Construction.)} For this proof given a tree decomposition $(T',\beta')$, we will identify 
  a node $t\in V(T')$ by the vertex subset $\beta(t)$. That is, rather than calling it node $t$, we will call it $\beta(t)$. This will help us ease notations as well as make it consistent with the notations used by 
  Elberfeld and Schweitzer~\cite{ElberfeldS17}, which we will refer to for the correctness of our construction.  
  
  We use Lemma~\ref{lem:c-atom-tree} to construct a tree
  decomposition~$D' = ( T_{d+1}(G),\beta_{d+1})$ whose bags are
  the graph's $(d+1)$-atoms and minimal clique separators of size at most $d+1$. The vertex set of $T_{d+1}(G)$ is partitioned into $A'$ and $C'$, where $A'$ consists of all $d+1$-atoms of $G$ and $C'$ consists of all minimum clique separators of size at most~$d+1$. Next we apply 
  induction hypothesis, and compute for each $(d+1)$-atom~$Y$ an
  isomorphism-invariant tree decomposition~$D_Y = (T_Y,\beta_Y)$ into
  its~$c$-atoms.  Let $A_Y$ and $C_Y$ be the corresponding vertex partition of $T_Y$ satisfying all the properties listed in the statement of the lemma. 
  We  combine~$D'$ with the decompositions $D_A$ to
  construct $D = (T,\beta)$.  Let $r_Y$ denote the root of $T_Y$. Informally, we obtain 
  $(T,\beta)$, by taking disjoint union of $D_Y = (T_Y,\beta_Y)$ and node set $C'$ (where each node in $C'$ 
   can be regarded as a single-node tree decomposition), and then for every edge
  $Y_1Y_2$ in $E(T_{d+1}(G))$, add an edge between a node, say ${\sf top}_{Y_1}(Y_1\cap Y_2)$, in $T_{Y_1}$, that contains $Y_1\cap Y_2$, 
   closest to $r_{Y_1}$ and a node, say ${\sf top}_{Y_2}(Y_1\cap Y_2)$, in $T_{Y_2}$, that contains $Y_1\cap Y_2$, closest to $r_{Y_2}$. Observe that since $Y_1\cap Y_2$ is a clique, we have that  ${\sf top}_{Y_1}(Y_1\cap Y_2)$, and  ${\sf top}_{Y_2}(Y_1\cap Y_2)$ are unique. 
  
  Formally, we use nodes $V(T) \coloneqq \{(X,Y) \mid
  X \in V(T_Y) \text{ and } Y \in V(T_{d+1}(G)) \}$ 
  for $T$. Two nodes~$(X_1,Y_1)$ and~$(X_2,Y_2)$ of~$T$ are adjacent
  if (1) $Y_1 = Y_2$ and~$X_1$ and~$X_2$ are adjacent in~$D_{Y_1}$, or
  (2) $Y_1$ and~$Y_2$ are adjacent in~$T_{d+1}(G)$ and for each~$i \in \{1,2\}$, $X_i$
  contains $Y_1 \cap Y_2$ and is closest to the root with this property in
  $D_{Y_i}$. To each node of $(X,Y)$ of $T$, $\beta$ assigns the
  bag $X$. That is, $\beta((X,Y))=X$. We construct the partition of $V(T)$ as follows: 
  $$A=\bigcup\{A_Y\colon Y \textrm{ is a }(d+1)\textrm{-atom}\}\qquad\textrm{and}\qquad C=C'\cup \bigcup\{C_Y\colon Y \textrm{ is a }(d+1)\textrm{-atom}\}.$$ Since, there is no edge between two vertices in $A'$, we have that there is no edge between any pair of vertices in $A$. The rest of the proof that it is indeed a tree decomposition with the desired properties can be found in~\cite[Lemma 3.9]{ElberfeldS17}.

The running time of the algorithm can easily seen to be polynomial. Indeed, we have at most $c-d$ levels of recursion, and in each level we apply Lemma~\ref{lem:c-atom-tree} on at most $n$ instances of size at most $n$. Each invocation of Lemma~\ref{lem:c-atom-tree} takes polynomial time yielding a polynomial upper bound on the running time. Finally, to convert  an isomorphism-invariant un-rooted tree decomposition into an isomorphism-invariant rooted tree decomposition, we use the folklore trick mentioned in~\cite[Fact~2.1]{ElberfeldS17}. In particular we can select the center $c$ of the decomposition tree $T$ as root. If $T$ has two centers then they are adjacent in $T$, and subdividing the edge between them yields a tree with a unique center. The bag of this new node is chosen to be the intersection of the bags of its two neighbors. This concludes the proof. 
\end{proof}

Using Lemma~\ref{lem:c-to-c++-atom-decomposition} we prove the next structural claim which is crucially used in the proof later. 

\begin{lemma}
\label{lem:isoinvBag}
Let $G$ be a graph, $h$ be a positive integer, and $C$ be a clique of size at most $h$ in $G$. Suppose further that $G-C$ is connected and that every vertex in $C$ has a neighbor outside of $C$. Then, in polynomial time, we can find an isomorphism invariant (in $G$ and $C$) set $B$ such that
\begin{itemize}
  \setlength{\itemsep}{-1pt}
    \item $C\subsetneq B$ ($B$ is a strict superset of $C$);
    \item for every component $X$ in $G-B$, we have that  $|N(X)\cap B|\leq h$; and
    \item $B$ is either a clique separator of size at most $h$ or an $h$-atom.
\end{itemize}
\end{lemma}
\begin{proof}
Recall that graph $G$ is connected, so it is a $0$-atom. We apply Lemma~\ref{lem:c-to-c++-atom-decomposition} to $G$ with $d = 0$ and $c = h $, and in polynomial time obtain an isomorphism invariant rooted tree decomposition $(T,\beta)$, and a partition of 
 $V(T)$ into two sets ${\sf Atoms}$ and ${\sf CliqueSep}$ such that  the following holds: 
  \begin{enumerate}
  \setlength{\itemsep}{-1pt}
  \item \label{es:pr:1}There is no edge between any pair of vertices in ${\sf Atoms}$. 
 \item \label{es:pr:2} For every  $t\in {\sf CliqueSep}$, $\beta(t)$ is a clique of size at most $h$.
 \item \label{es:pr:3} For every  $t\in {\sf Atoms}$,  $\beta(t)$ is a maximal $h$-atom of $G$. 
 \item \label{es:pr:4} If $s \in  {\sf CliqueSep}$, $t \in {\sf Atoms}$ and $st \in E(T)$ then $\beta(s) \subseteq \beta(t)$. 
 \end{enumerate}
Let $r^\star$ be the root of $T$. 
Since  $C$  is a clique in $G$, there exists a node $t\in V(T)$ such that $C \subseteq \beta(t)$. Let $r$ be the unique (topmost  in $T$) node closest to $r^\star$  such that $C \subseteq \beta(r)$, and let $B=\beta(r)$.  Clearly, the definition of $B$ is isomorphism invariant in $G$ and $C$. Further, by construction $B$ is either a clique separator of size at most $h$ or an $h$-atom.  Finally, note that since every adhesion of  $(T,\beta)$ has size at most $h$ (the last property of Lemma~\ref{lem:c-to-c++-atom-decomposition}),
 we have that $|N(X)\cap B|\leq h$ for every component $X$ in $G-B$.

All that remains to show is that $C\subsetneq B$. 
We divide the proof into several cases and argue each case individually. 
\begin{description}
\item[Case I: ($B=V(G)$)] Since, every vertex in $C$ has a neighbor outside of $C$, we have that there exists a vertex $w \in V(G)\setminus C$ and hence $C\subsetneq B$. 
\item[Case II: ($B$ is a clique separator)] We know that $G-C$ is connected and furthermore, $G-B$ is not connected. This together with the fact that $C\subseteq B$ implies that in fact, $C\subsetneq B$. 
\item[Case II: ($B$ is an $h$-atom)] In this case we know that $B\neq V(G)$. Thus, by Properties \ref{es:pr:1} and \ref{es:pr:4} of $(T,\beta)$,   there exists a clique separator $C'\subseteq B$ of $G$, of size at most $h$.  
If $C'\nsubseteq C$, that is, $|C'\setminus C|\geq 1$, then it is clear that $C\subsetneq B$. So we assume that 
$C'\subseteq C$. Further, we assume that $B=C$, else we are done. However, $C'$ is a clique separator in $G$ and hence there exist two vertices $x$ and $y$ such that $x$ and $y$ belong to two different connected components in $G-C'$. Clearly, $x$ and $y$ do not both belong to $V(G)\setminus B$ as 
$G-B=G-C$ and $G-C$ is connected. Similarly,  $x$ and $y$ do not both belong to $C$, as $C$ is a clique. Therefore, w.l.o.g. $x\in C$ and $y\in V(G)\setminus B$. Since $C'\subseteq C$, we have that the connected component of $G-C'$ containing $y$ contains all the vertices in  $V(G)\setminus B$. This implies that $x$ does not have any neighbor in $V(G)\setminus B$, which contradicts our assumption that  every vertex in $C$ has a neighbor outside of $C$. So the case $B=C$ does not occur.
\end{description}
This concludes the proof. 
\end{proof}

\subsubsection{Structural Properties of Clique Separators and Improved Graphs}
In this subsection we give couple of standalone structural results that will be useful later. We start with a simple observation. 
\begin{lemma}
\label{lem:indImprovement}
Let $G$ be a graph and $(A,B)$ be a clique separator of $G$. Then,
$\imp{G}{k}[A]=\imp{G[A]}{k}$. 
\end{lemma}
\begin{proof}
We show that for every pair of vertices $u,v\in A$, we have $\mu_G(u,v)=\mu_{G[A]}(u,v)$. It is clear that $\mu_G(u,v)\geq \mu_{G[A]}(u,v)$, so let us focus on the converse inequality. Also, assume that $u$ and $v$ are non-adjacent, as otherwise both sides of the postulated equality are equal to $+\infty$.

Let $\cal P$ be a family of internally disjoint $u$-$v$ paths that witnesses the value of $\mu_G(u,v)$. By shortcutting if necessary, we may assume that each path in $\cal P$ is induced. Now observe that since $(A,B)$ is a clique separator and $u,v\in A$, every induced $u$-$v$ path in $G$ is entirely contained in $G[A]$. Thus $\mu_{G[A]}(u,v)\geq |{\cal P}|=\mu_G(u,v)$ and we are done.
\end{proof}

In the next statement we turn the decomposition provided by Lemma~\ref{lem:c-to-c++-atom-decomposition} into an isomorphism invariant star decomposition.

\begin{lemma}
\label{lem:unbreakbleAndstarDeco}
If a graph $G$ is $(q,k)$-clique-unbreakable then there exists an isomorphism invariant star decomposition $(T,\chi)$, computable in polynomial time, with the following properties. 
\begin{enumerate}
\setlength{\itemsep}{-2pt}
\item If $r$ is the central (root) node of $(T,\chi)$, then $\chi(r)$ is either a clique of size at most $k$ or a $k$-atom. 
\item For every $t\in V(T)$, $|\sigma(t)|\leq k$ (the adhesion sizes are bounded by $k$). 
\item For every leaf $\ell \in V(T)$, $|\chi(\ell)|\leq q$. 
\item $(T,\chi)$ is connectivity-sensitive. %
\end{enumerate}
Furthermore, suppose $G$ has a star decomposition $(T,\chi)$ with properties as above. 
Then, $G$ is $(2^kq,k)$-clique-unbreakable.
\end{lemma}
\begin{proof}
First we assume that $G$ is connected, hence it is a $0$-atom. We apply Lemma~\ref{lem:c-to-c++-atom-decomposition} to $G$ with $d = 0$ and $c = k $, and obtain an isomorphism invariant rooted tree decomposition $(T',\beta)$, and a partition of 
 $V(T')$ into two sets ${\sf Atoms}$ and ${\sf CliqueSep}$ such that  the following holds: 
  \begin{enumerate}
  \setlength{\itemsep}{-1pt}
  \item \label{nes:pr:1}There is no edge between any pair of vertices in ${\sf Atoms}$. 
 \item \label{nes:pr:2} For every  $t\in {\sf CliqueSep}$, $\beta(t)$ is a clique of size at most $k$.
 \item \label{nes:pr:3} For every  $t\in {\sf Atoms}$,  $\beta(t)$ is a maximal $k$-atom of $G$. 
 \item \label{nes:pr:4} If $s \in  {\sf CliqueSep}$, $t \in {\sf Atoms}$ and $st \in E(T')$ then $\beta(s) \subseteq 
 \beta(t)$. 
  \item Each adhesion is a clique.  That is, for every node $t$, $\sigma(t)$ is a clique of size at most $k$.
 \end{enumerate}
 
Let $r^\star$ be the root of $T'$.  We now describe a procedure to orient the edges of $T'$.   Let $t_1t_2$ be any edge of $T'$ and let $T_1,T_2$ be the components of $T' -t_1t_2$, 
with $t_1 \in  V(T_1)$ and $t_2 \in  V(T_2)$. Then we know that $\beta(t_1) \cap \beta(t_2)$  separates
$U_1\coloneqq\bigcup_{t\in V(T_1)} \beta(t)$ from $U_2\coloneqq\bigcup_{t\in V(T_2)} \beta(t)$  in $G$.  Since  
$\beta(t_1) \cap \beta(t_2)$ is a clique of size at most $k$ and  $G$ is $(q,k)$-clique-unbreakable, we have that either $|U_1|\leq q$ or $|U_2|\leq q$. We find an $i\in \{1,2\}$ such that $|U_i| >q$ (if exists), and orient $t_1t_2$ towards $t_i$. We call such edges {\em fixed}. However, we could have an edge where both  $|U_1|\leq q$ and $|U_2|\leq q$ (this may happen when $G$ has fewer than $2q$ vertices). We call such 
edges {\em ambiguous}. We orient the ambiguous edges in the direction of $r^\star$. That is, we find an $i\in \{1,2\}$ such that $t_i$ is strictly closer to~$r^\star$, and orient $t_1t_2$ towards $t_i$. This completes the description of a  procedure to orient the edges of $T'$. 

Next we show that the orientation is consistent in the sense that there is precisely one vertex of out-degree $0$ and that all edges are oriented towards this vertex. Towards this goal we observe the following. 
\begin{enumerate}
\setlength{\itemsep}{-2pt}
\item  For every $t_1\in V(T')$, the number of fixed edges oriented away from $t_1$ is upper bounded by $1$. Indeed, consider a fixed edge oriented away from $t_1$, say $t_1t_2$, and consider the clique separation $(U_1,U_2)$ of order $k$ defined above. Now observe that if we had more than one fixed edge oriented away from $t_1$  then $|U_1|>q$, and  $|U_2|>q$, contradicting that $G$ is $(q,k)$-clique-unbreakable. 
\item If $t_1\in V(T')$ is incident to a fixed edge oriented away from $t_1$, then there are no ambiguous edges oriented away from $t_1$. Indeed, let the fixed edge be $t_1t$, and consider another edge incident to $t_1$, say $t_1t_2$ ($t_2\neq t$). Let $(U_1,U_2)$ be the clique separation  of order $k$ corresponding to the edge $t_1t_2$. Now observe that  since we have a fixed edge $t_1t$ oriented away from $t_1$,   we have that  $|U_1|>q$,  so $t_1t_2$ must be a fixed edge that is oriented towards 
$t_1$.  
\item  For every $t_1\in V(T')$, the number of ambiguous edges oriented away from $t_1$ is upper bounded by $1$. This follows from the fact that $T'$ is a tree and we have a unique root $r^\star$. 
\end{enumerate}
The above observations imply that after orientation, every vertex in $T'$ has out-degree at most $1$.  Further, every edge contributes to the in-degree of a single vertex and to the out-degree of a single vertex.  However, the tree has has at most $n-1$ edge and $n$ vertices and hence, $n-1$ vertices have out-degree $1$ and there exists exactly one vertex that has out-degree $0$. Let $t^\star$ be the {\em unique node} such that the out-degree of $t^\star$ is $0$. Clearly, the definition of $t^\star$ is isomorphism invariant.  

Given $(T',\beta)$, we construct the star decomposition $(T,\chi)$ satisfying the properties mentioned in the statement 
as follows.     Let $B=\beta(t^\star)$.  The central node is $r$ and we set $\chi(r)\coloneqq B$. Furthermore, for each subset $X$ of $B$ such that $X = N(C)$ for some component $C$ of $G-B$, let $C_X$ be the union of all such components. For each  $C_X$ we construct a leaf $\ell_X\in V(T)$ with $\chi(\ell_X)=N[C_X]$.  Thus, by construction we have that  $\chi(\ell) \cap \chi(r) \neq \chi(\ell') \cap \chi(r)$ for every pair of distinct leaves $\ell, \ell'$ in $T$, and for every leaf $\ell$ of $T$ and every connected component $C$ of $G[\chi(\ell) \setminus \chi(r)]$ it holds that $N(C) = \chi(\ell) \setminus \chi(r)$. Thus, the constructed star decomposition $(T, \chi)$ with central node $r$ is connectivity-sensitive.  
Finally, we prove properties $3$ and $4$, encapsulated in the following observation. 
\begin{observation}
\label{obs:CSSsizebound}
For every leaf $\ell \in V(T)$, we have  $|\sigma(\ell)|\leq k$ and  $|\chi(\ell)|\leq 3q$. 
\end{observation}
\begin{proof}
Let $C$  be a connected component $G-B$. We  first show that $|C|\leq q$. Indeed, consider a child $t'$ of $t^\star$ in $T'$, such that $C\cap (U_{t'}\setminus B)\neq \emptyset$. Here, $U_{t'}\coloneqq \bigcup_{t\in V(T_{t'})} \beta(t)$ and $T_{t'}$ is the rooted subtree of $T'$. However, since $C$ is connected, it fully lies in  
$(U_{t'}\setminus B)$ and $N(C) \subseteq \sigma(t')\leq k$. Further, $|U_{t'}|\leq q$, for otherwise the edge $t' t^\star$ would be oriented towards $t'$. Thus, $|C|\leq q$. Further,  $|N(C)|\leq k$ and  $(T, \chi)$ is connectivity-sensitive. Thus, we have that for every  $\ell \in V(T)$, $|\sigma(\ell)|\leq k$. 

Next we show that $|\chi(\ell)|\leq 3q$. Let $\hat{C}$ be the union of all components such that $\chi(\ell)=N[\hat{C}]$.  Suppose $|\hat{C}|>3q$. Let $C^\star$ be a subset of components of  $\hat{C}$ such that $q< \sum_{C \in C^\star} |C| <2q$. Now consider the clique separation $(A, B)$ of order at most $k$, where $A=N[C^\star]$, and $B=V(G)\setminus C^\star$. Since, $|\hat{C}|>3q$, we have that $|A|>q$ and $|B|>q$, contradicting that $G$ is $(q,k)$-clique-unbreakable. 
\end{proof}

This concludes the proof for the forward direction for the case that $G$ is connected. Now we give the desired construction for the general case, when $G$ may be disconnected. Since, $G$ is $(q,k)$-clique-unbreakable, there exists at most one component $C$ such that $|C|>q$. Further, if the total number of vertices appearing in connected components of size at most $q$ exceeds $3q$, then similarly to the proof of Observation~\ref{obs:CSSsizebound} we can contradict that $G$ is $(q,k)$-clique-unbreakable. Thus, if there is {\em no} component of size more than $q$ then the star decomposition $(T, \chi)$ with central node $r$ with $\chi(r) = \emptyset$ and a single leaf $\ell$ with $\chi(\ell) = V(G)$ satisfies the conclusion of the lemma. From now on we assume that there exists a component $C$ of size more than $q$.

Let $C_{\sf small}$ be the union of all components of size at most $q$. We first apply the construction of $G$ connected on $C$ and obtain an isomorphism invariant star decomposition $(T,\chi)$ of $C$. Finally, we add a leaf $\ell$ to $T$ and assign 
$\chi(\ell)=C_{\sf small}$. Thus, we obtain an isomorphism invariant star decomposition $(T,\chi)$ of $G$ satisfying all the required properties. 
To compute $(T,\chi)$ from $(T',\beta)$ in polynomial time it suffices to execute the steps of the construction. 
This concludes the proof for the forward direction.

We are left with the last claim of the lemma statement. Suppose $G$ admits a star decomposition $(T,\chi)$ as asserted, and let $(A,B)$ be a clique separation of $G$ of order at most $k$. Since $\chi(r)$ is either a clique of size at most $k$ or a $k$-atom, we have that 
$\chi(r) \subseteq A$ or $\chi(r) \subseteq B$. Without loss of generality assume that 
$\chi(r) \subseteq A$. Call a leaf $\ell$ {\em{affected}} if $\chi(\ell)$ contains a vertex of $B$. Observe that since $(T,\chi)$ is connectivity-sensitive, for every affected leaf $\ell$ we either have that $\chi(\ell)\setminus \chi(r)$ contains a vertex of $A\cap B$, or $\chi(\ell)\setminus \chi(r)\subseteq B\setminus A$ and $\sigma(\ell)\subseteq (A\cap B\cap \chi(r))$. The number of affected leaves of the first kind is bounded by $|(A\cap B)\setminus \chi(r)|$, while the number of leaves of the second kind is bounded by $2^{|A\cap B\cap \chi(r)|}$ due to connectivity sensitivity of $(T,\chi)$. So all in all the number of affected leaves is at most
$$|(A\cap B)\setminus \chi(r)|+2^{|A\cap B\cap \chi(r)|}\leq 2^{|A\cap B|}\leq 2^k.$$
As the bag of every affected leaf is of size at most $q$, we conclude that $|B|\leq 2^k q$. This concludes the proof.
\end{proof}

\subsubsection{Reduction Algorithm}
In this subsection we combine all the tools prepared so far and prove Lemma~\ref{lem:cliqueToGeneral}. Recall that, intuitively, Lemma~\ref{lem:cliqueToGeneral} provides a reduction to the setting of graphs whose improved graphs are clique-unbreakable. We re-state it here for convenience.

\medskip
\noindent
{\bf Lemma~\ref{lem:cliqueToGeneral} (restated). }{\em
Let $\mathbb{F} \subseteq \mathbb{F}^\star$ be a (possibly infinite) collection of finite sets of graphs.
Suppose there exists a canonization algorithm ${\cal B}$ for $\kappa_{\cal B}$-improved-clique-unbreakable
$\mathbb{F}$-free graphs, for some function $\kappa_{\cal B} \colon \mathbb{F} \rightarrow \mathbb{N}$. Then
there exists a canonization algorithm for $\mathbb{F}$-free classes. 
The running time of this latter algorithm is upper bounded by $g({\cal F})n^{\cO(1)}$ (for some function $g \colon \mathbb{F} \rightarrow \mathbb{N}$) plus the time taken by at most $g({\cal F})n^{\cO(1)}$ invocations of ${\cal B}$ on instances $(G', {\cal F}, k, q)$ where $G'$ is a $(q,k,k)$-improved-clique-unbreakable, ${\cal F}$-topological minor free graph on at most $n$ vertices and $k$ and $q$ are integers upper bounded by a function of ${\cal F}$ and $\kappa_{\cal B}$. 
}
\smallskip

The remainder of this section is devoted to the proof of Lemma~\ref{lem:cliqueToGeneral}. Thus, throughout this section we assume that the premise of Lemma~\ref{lem:cliqueToGeneral} holds.
That is, let  $\mathbb{F} \subseteq \mathbb{F}^\star$ be a collection of finite sets of graphs, let  $\kappa_{\cal B}  \colon \mathbb{F} \rightarrow \mathbb{N}$ be a function, and  ${\cal B}$ be a canonization algorithm for $\kappa_{\cal B}$-improved-clique-unbreakable $\mathbb{F}$-free graphs. In order to prove Lemma~\ref{lem:cliqueToGeneral} we will prove the following stronger statement.

\begin{lemma}\label{lem:cliqueToGeneralBoundaried} 
There exists a weakly isomorphism invariant algorithm with the following specifications. The algorithm takes as input 
\begin{itemize}
\setlength{\itemsep}{-2pt}
\item a set ${\cal F} \in \mathbb{F}$,
\item an integer $k$, and 
\item a colored graph $G$ with a distinguished vertex subset $D$ satisfying $|D| \leq k$ such that
there exists a labelling $\iota \colon D \rightarrow [|D|]$ and a $[|D|]$-boundaried graph $(H,\iota)$, boundary-consistent with $G$, so that $(G,\iota) \oplus (H,\iota)$ is ${\cal F}$-free and $D$ is a clique in the $k$-improvement of $(G,\iota) \oplus (H,\iota)$.
\end{itemize}
The algorithm either fails (i.e. outputs $\bot$) or succeeds. If $k \geq \kappa_{\cal B}({\cal F})$ then the algorithm succeeds and 
\begin{itemize}
\item
it outputs for every labeling $\pi \colon D \rightarrow [|D|]$ a proper labeling $\lambda_\pi$ of $(G, \pi)$ such that $\lambda_\pi(v) = \pi(v)$ for all $v\in D$. The labeling $\lambda_\pi$ is weakly isomorphism invariant in $(G, \pi)$ and ${\cal F}$ and $k$. 
\end{itemize}
The running time of the algorithm is upper bounded by $g({\cal F}, k)n^{\cO(1)}$ plus the time taken by $g({\cal F},k)n^{\cO(1)}$ invocations of ${\cal B}$ on instances $(G', {\cal F}, k, q)$ where $G'$ is a $(q,k,k)$-improved-clique-unbreakable, ${\cal F}$-topological minor free graph on at most $n$ vertices and $q$ is upper bounded by a function of ${\cal F}$ and $k$.
\end{lemma}

We remark that the labeling $\iota$ and boundaried graph $(H, \iota)$ are not provided to the algorithm as input. 
Before proceeding to the proof of Lemma~\ref{lem:cliqueToGeneralBoundaried}, we show how Lemma~\ref{lem:cliqueToGeneralBoundaried} implies the goal of this section - namely Lemma~\ref{lem:cliqueToGeneral}

\begin{proof}[Proof of Lemma~\ref{lem:cliqueToGeneral} assuming Lemma~\ref{lem:cliqueToGeneralBoundaried}]
We assume that the premise of the Lemma holds, and design a canonization algorithm for $\mathbb{F}$-free classes.
Given as input ${\cal F} \in \mathbb{F}$ and an ${\cal F}$-free graph $G$, the algorithm sets $k = 1$ and invokes the algorithm of Lemma~\ref{lem:cliqueToGeneralBoundaried} on ${\cal F}$, $k$, $G$, $D=\emptyset$.
Since $G$ is ${\cal F}$-free and $D=\emptyset$ is already a clique in the $k$-improvement of $G$, this is a well-formed input to the algorithm of Lemma~\ref{lem:cliqueToGeneralBoundaried}.
The algorithm of Lemma~\ref{lem:cliqueToGeneralBoundaried} either returns a canonical labeling $\lambda$ of $G$, or it fails and returns $\bot$. In that case our algorithm increments $k$ and calls the algorithm of Lemma~\ref{lem:cliqueToGeneralBoundaried} on ${\cal F}$, $k$, $G$, $D=\emptyset$ again.
After at most $\kappa_{\beta}({\cal F})$ iterations we have $k \geq \kappa_{\beta}({\cal F})$ and so the algorithm succeeds and outputs a canonical labeling of $G$ (since $D=\emptyset$ there is a unique choice for labeling $\pi : D \rightarrow [|D|]$ as the empty function which does not assign anything to anything).
The claimed canonization algorithm for $\mathbb{F}$-free is weakly isomorphism invariant because the algorithm of Lemma~\ref{lem:cliqueToGeneralBoundaried} is.
The claimed running time follows from the running time upper bound of Lemma~\ref{lem:cliqueToGeneralBoundaried} together with the fact that we invoke this algorithm at most $\kappa_{\beta}({\cal F})$ times, each time with $k \leq \kappa_{\beta}({\cal F})$.
\end{proof}

\begin{proof}[Proof of Lemma~\ref{lem:cliqueToGeneralBoundaried}]
We set $|D|=h\leq k$, throughout the algorithm.  
We ensure that the algorithm does not loop using a potential $\pot(G,D)=2(|V(G)|-|D|)+|D|=2|V(G)|-|D|$. Formally, every call of the algorithm will use only calls for inputs with a strictly smaller potential, and
calls with nonpositive potential will be resolved without invoking any further calls (i.e. they are base cases).

\paragraph{Base Case ($\pot(G,D) \leq 0$):} Consider such a base case when $2|V(G)|-|D|\leq 0$.
This implies that $|V(G)|\leq |D|/2=h/2$, so the overall vertex count is bounded by $h/2\leq k/2$. So in this case we can iterate, for every labeling $\pi$ of $D$, through all proper labelings $\lambda_\pi$ of $(G,\pi)$ that extend $\pi$, and choose the one that minimizes the properly labeled graph $(G,\lambda_\pi)$ lexicographically. Clearly, this choice is weakly isomorphism invariant.

\smallskip
From now on we may assume that we work with the regular case when $2|V(G)|-|D|>0$. Then we design and prove the correctness of the algorithm as follows. We first deal with two simple cases. We apply these cases in the order of appearance, that is, in every subsequent case we assume that the previous ones do not apply.

\paragraph{Case I: $G-D$ is disconnected.} In this case we recurse on $\widetilde{G}=G[N[C]]$ for each connected component $C$ of $G-D$. However, to do this we first need to show that $\widetilde{G}$ satisfies the premise of Lemma~\ref{lem:cliqueToGeneralBoundaried}. We take the restriction of ${\sf col}_G$ to $\widetilde{G}$ as the coloring ${\sf col}_{\widetilde{G}}$. 
The set $\widetilde{D}=N(C)$ is defined to be the distinguished set.  Since $|D|\leq h$ and $N(C)\subseteq D$, we have that 
$|\widetilde{D}|=|N(C)|\leq h$.
Let $\tilde{\iota} \colon \widetilde{D} \to [|\widetilde{D}|]$ be an arbitrary bijection from $\widetilde{D}$ to $[|\widetilde{D}|]$.

We define $\widetilde{H}$ as $(G-C,\iota)\oplus (H,\iota)$.
We view $(\widetilde{H},\tilde{\iota})$ as a $[|\widetilde{D}|]$-boundaried graph with boundary $\widetilde{D}$ and the labelling $\tilde{\iota} \colon \widetilde{D} \to [|\widetilde{D}|]$, defined above. By construction $(\widetilde{H},\tilde{\iota})$ is boundary-consistent with $(\widetilde{G},\tilde{\iota})$, and $(\widetilde{G},\tilde{\iota})\oplus (\widetilde{H},\tilde{\iota})$ stripped of the boundary is the same graph as $(G,\iota)\oplus (H,\iota)$ stripped of the boundary.  This implies that $(\widetilde{G},\tilde{\iota})\oplus (\widetilde{H},\tilde{\iota})$ is ${\cal F}$-free, as $(G,\iota)\oplus (H,\iota)$ is ${\cal F}$-free. Finally, we need to show that 
$\widetilde{D}$ is a clique in the $k$-improved graph of $(\widetilde{G},\tilde{\iota}) \oplus (\widetilde{H},\tilde{\iota})$. However, this follows from the fact that $\widetilde{D}\subseteq D$, $D$ is a clique in the $k$-improved graph of $(G,\iota)\oplus (H,\iota)$, and  $(\widetilde{G},\tilde{\iota}) \oplus (\widetilde{H},\tilde{\iota})$ and  $(G,\iota) \oplus (H,\iota)$ are the same graphs with possibly different boundaries.
To be able to apply induction, it remains to show that 
 $\pot(\widetilde{G},\widetilde{D})<\pot(G,D)$. However, this follows from the fact that $G-D$ is not connected and hence $\widetilde{G}$ does not contain all the vertices of $G-D$. 

Having justified the possibility of invoking the algorithm on $\widetilde{G}$ and $\widetilde{D}$, this invocation either outputs $\bot$ or it outputs for every labeling $\pi \colon \widetilde{D} \rightarrow [|\widetilde{D}|]$ a proper labeling $\tilde{\lambda}_\pi$ of $(\widetilde{G}, \pi)$ such that $\tilde{\lambda}_\pi(v) = \pi(v)$. The labeling $\tilde{\lambda}_\pi$ is weakly isomorphism invariant in $(\widetilde{G}, \pi)$. If the algorithm outputs $\bot$ for any component $C$, we output $\bot$ for the original instance $({\cal F},k,G,D)$ and terminate (note that this is weakly isomorphism invariant).  Thus, we assume that the algorithm succeeds on each recursive call made on the children.

Our goal now is to aggregate labelings $\tilde{\lambda}_\pi$ obtained from the calls into labelings $\lambda_\pi$ for $G$ and $D$.
We first apply Lemma~\ref{lem:compact} on $G$ and ${\sf col}_{G}$  and obtain an isomorphism invariant  compact coloring of $G$. We also need to compactify the labelings obtained from applying recursive steps. To do this we again apply Lemma~\ref{lem:compact}   and obtain compact labelings.

Next, we define a star decomposition $(T, \chi)$ of $G$ as follows. The bag of the central node $r$ is $\chi(r)=D$. Next, for every connected component 
$C$ of $G-D$, construct a leaf $\ell$ with associated bag $\chi(\ell)=N[C]$. Clearly, $(T,\chi)$ is isomorphism invariant in $G$ and $D$.

Using the output of the recursive 
calls, we define a leaf-labeling $\{\lambda_{\pi,\ell}\}$ of $(T, \chi)$ as follows. Precisely, for every leaf $\ell$ of $T$, the invocation on the corresponding instance $(\widetilde{G},\widetilde{D})$ returns a set of labelings $\{\tilde{\lambda}_{\pi}\}$ for all $\pi\colon \sigma(\ell)\to [|\sigma(\ell)|]$, and we set $\lambda_{\pi,\ell}=\tilde{\lambda}_\pi$. By construction, the obtained leaf labeling is locally weakly isomorphism invariant. Now we apply Lemma~\ref{lem:smallCentersAlgorithm} to $(G, D,  (T,\chi,b,\{ \lambda_{\pi, \ell} \}))$  
and every labeling $\pi \colon D \rightarrow [|D|]$, thus obtaining 
 proper labelings $\lambda_\pi$ of $(G, \pi)$ such that $\lambda_\pi(v) = \pi(v)$. 
 The labeling $\lambda_\pi$ is weakly isomorphism invariant in $(G, \pi)$, ${\cal F}$ and $k$ hence it can be returned by the algorithm.

 The running time of this step is upper bounded by $|D|!n^{\cO(1)}=2^{\cO(h \log h)} n^{\cO(1)} $.

\paragraph{Case II: There exists a vertex in $D$ that does not have a neighbor in $G-D$.} In this case 
we will recurse on $({\cal F},k,\widetilde{G},\widetilde{D})$ defined as follows. $\widetilde{D}$ is the subset of $D$ that comprises of all the vertices of $D$ that have a neighbor in
 $V(G)\setminus D$, and $\widetilde{G}=G-(D\setminus \widetilde{D})$. With arguments identical to Case I, we can 
show that $({\cal F},k,\widetilde{G},\widetilde{D})$  satisfies the premise of Lemma~\ref{lem:cliqueToGeneralBoundaried}.  
 To apply the inductive argument we need to show that 
 $\pot(\widetilde{G},\widetilde{D})<\pot(G,D)$. Indeed, 
 $$\pot(\widetilde{G},\widetilde{D})=2(n-|D\setminus \widetilde{D}|)- |\widetilde{D}|=(2n-|D|)-|D\setminus \widetilde{D}|= \pot(G,D)-|D\setminus \widetilde{D}|< \pot(G,D).$$
 
 Having justified the possibility of invoking the algorithm on $({\cal F},k,\widetilde{G},\widetilde{D})$,
 the invocation either outputs $\bot$ or it outputs for every labeling $\pi \colon \widetilde{D} \rightarrow [|\widetilde{D}|]$ a proper labeling $\widetilde{\lambda}_\pi$ of $(\widetilde{G}, \pi)$ such that $\widetilde{\lambda}_\pi(v) = \pi(v)$ for all $v\in \widetilde{D}$. The labeling $\widetilde{\lambda}_\pi$ is weakly isomorphism invariant in $(\widetilde{G}, \pi)$. If the algorithm outputs $\bot$, we output $\bot$ on the original instance $({\cal F},k,G,D)$ (and this is weakly isomorphism invariant).  Thus, we assume that the algorithm succeeds on the recursive call made on  the instance $({\cal F},k,\widetilde{G},\widetilde{D})$.   

We first apply Lemma~\ref{lem:compact}  on $G$ and ${\sf col}_{G}$  and obtain an isomorphism invariant  compact coloring of $G$. We also compact the labelings obtained from applying recursive steps. To do this we again apply Lemma~\ref{lem:compact}   and obtain compact labelings.

Next, we define $(T, \chi)$: an isomorphism invariant (with respect to $G$ and $D$) star decomposition of 
$G$. $(T,\chi)$ has a central node $r \in V(T)$ and associated bag $\chi(r)=D$. Further, we have a single leaf $\ell$ 
such that  $\chi(\ell)=V(\widetilde{G})$.  Using the output of the recursive call, we define a leaf-labeling $\{ \lambda_{\pi,\ell} \}$ of $(T, \chi)$ as follows: for every permutation $\pi \colon \sigma(\ell) \rightarrow [|\sigma(\ell)|]$ we set $\lambda_{\pi,\ell}=\widetilde{\lambda}_\pi$. By construction, this leaf labeling is locally weakly isomorphism invariant.  Now we apply Lemma~\ref{lem:smallCentersAlgorithm}, with $(G, D,  (T,\chi,b,\{ \lambda_{\pi, \ell} \}))$   
and for every labeling $\pi \colon D \rightarrow [|D|]$,  obtain 
 a labeling $\lambda_\pi$ of $(G, \pi)$ such that $\lambda_\pi(v) = \pi(v)$. 
 The labeling $\lambda_\pi$ is weakly isomorphism invariant in $(G, \pi)$.  The running time of this step is upper bounded by $|D|!n^{\cO(1)}=2^{\cO(h \log h)} n^{\cO(1)} $.

\medskip

\noindent 
{\bf Non-Trivial Cases.}  Because of the cases considered previously, we assume that  $G-D$ is connected, and every vertex in $D$ has a neighbor in   $V(G)\setminus D$.  Now, we first compute a $k$-improved graph $\imp{\widehat{G}}{k}$ of $\widehat{G}$, where $\widehat{G}$ has been obtained from $G$ by making $D$ into a clique. Now we  apply Lemma~\ref{lem:isoinvBag} on $\imp{\widehat{G}}{k}$, $k$ and $D$, and in polynomial time, find an isomorphism invariant set $B$ such that (a) $D \subsetneq B$ ($B$ is a strict superset of $D$); and (b) $B$ is either a clique separator of size at most $k$ or a $k$-atom.

\paragraph{Case III: $B = V(G)$ in $\imp{\widehat{G}}{k}$.} In this case we know that  $\imp{\widehat{G}}{k}$ is a $k$-atom and thus it does not have any clique separator of size at most $k$. On the other hand $G$ may not be a $k$-atom.  
In this case for every labeling $\pi \colon D \rightarrow [|D|]$, we will produce  a proper  labeling $\lambda_\pi$ of $(G, \pi)$ such that $\lambda_\pi(v) = \pi(v)$ for all $v\in D$. The labeling $\lambda_\pi$ will be weakly isomorphism invariant in $(G, \pi)$.  

Recall that  there exist a labelling $\pi \colon D \rightarrow [|D|]$ and a  $[|D|]$-boundaried graph $(H,\pi)$ so that $(G,\pi) \oplus (H,\pi)$ is ${\cal F}$-free and $D$ is a clique in the $k$-improved graph of $(G,\pi) \oplus (H,\pi)$. This implies that {\em for every labelling} $\pi \colon D \rightarrow [|D|]$, there exists a  $[|D|]$-boundaried graph $(H,\pi)$ so that $(G,\pi) \oplus (H,\pi)$ is ${\cal F}$-free and $D$ is a clique in the $k$-improvement of $(G,\pi) \oplus (H,\pi)$. Indeed,  we could use the same $H$ but label the boundary vertices so that it is consistent with $\pi$. Let $n_\pi$ 
denote the cardinality of the vertex set of a smallest graph $H$ such that  $(H,\pi)$ is a $[|D|]$-boundaried graph and $(G,\pi) \oplus (H,\pi)$ is ${\cal F}$-free and $D$ is a clique in the $k$-improved graph of $(G,\pi) \oplus (H,\pi)$. Define ${\cal H}_\pi$ to be the set of all boundaried graphs $(H, \pi)$ on $n_\pi$ vertices such that $(G,\pi) \oplus (H,\pi)$ is ${\cal F}$-free and $D$ is a clique in the $k$-improved graph of $(G,\pi) \oplus (H,\pi)$. 

The algorithm constructs  ${\cal H}_\pi$ from $(G, \pi)$ by enumerating all boundaried graphs $(H, \pi)$ in increasing order of $|V(H)|$, testing for each one whether $(G,\pi) \oplus (H,\pi)$ is ${\cal F}$-free, using 
the algorithm of \cite{GroheKMW11} for testing a topological minor in a graph,   
and $D$ is a clique in the $k$-improved graph of $(G,\pi) \oplus (H,\pi)$, and stopping after enumerating all the graphs $H$ on $x$ vertices, where $x=n_\pi$ is the smallest number for which the algorithm finds an $H$ for which this test succeeds.

\begin{claim}\label{clm:caseIIIisCliqueUnbreakable}
For every  $(H,\pi) \in {\cal H}_\pi$ we have that $(G,\pi) \oplus (H,\pi)$ is $(n_\pi+k,k,k)$-improved-clique-unbreakable. 
\end{claim}
\begin{proof}
Let $\tilde{G}=(G,\pi) \oplus (H,\pi)$ and $G^\star=\imp{\tilde{G}}{k}$. Note that since $D$ is a 
clique in $G^\star$,  by Lemma~\ref{lem:indImprovement} applied to $\tilde{G}$ with $D$ turned into a clique, we have that $\imp{\widehat{G}}{k}=G^\star[V(G)]$. Consider a clique separation $(L,R)$ in $G^\star$ of order at most $k$. We will show that $V(G)$ is fully contained in either $L$ or $R$. Towards a contradiction, consider the separation $(L\cap V(G),R\cap V(G))$. Clearly this is a clique separation of order at most $k$ of $G^\star[V(G)]$. Since, $\imp{\widehat{G}}{k}=G^\star[V(G)]$, we have that it is also a clique separation of order at most $k$ of 
$\imp{\widehat{G}}{k}$. This contradicts the fact that $\imp{\widehat{G}}{k}$ is a $k$-atom. That $V(G)\subseteq L$ or $V(G)\subseteq R$ implies that either $L$ or $R$ is of size at most $n_\pi+k$. This concludes that $(G,\pi) \oplus (H,\pi)$ is $(n_\pi+k,k,k)$-improved-clique-unbreakable. \end{proof}

We first apply Lemma~\ref{lem:compact}  on $(G,\pi) \oplus (H,\pi)$ and ${\sf col}_{(G,\pi) \oplus (H,\pi)}$  and obtain an isomorphism invariant  compact coloring of $(G,\pi) \oplus (H,\pi)$. 
Then, we call ${\cal B}$ on $(G,\pi) \oplus (H,\pi)$, ${\cal F}$, $k$ and $n_\pi+k$ for every $(H,\pi)\in {\cal H}_\pi$ and obtain for each  $(H,\pi)$ a weakly isomorphism invariant proper labeling $\lambda_{(H,\pi)}$ of  $(G,\pi) \oplus (H,\pi)$. 
By Claim~\ref{clm:caseIIIisCliqueUnbreakable} the graph $(G,\pi) \oplus (H,\pi)$ is $(n_\pi+k, k, k)$-improved clique unbreakable and therefore this is a well-formed input for the algorithm ${\cal B}$.
If ${\cal B}$ fails (returns $\bot$) then our algorithm fails as well (and this is weakly isomorphism invariant). However, assuming that $k \geq \kappa_{\cal B}({\cal F})$ none of the calls to ${\cal B}$ fail. 
For each $H$ let $\lambda_{(H,\pi)}'$ be the restriction of $\lambda_{(H,\pi)}$ to the domain $V(G)$.
The algorithm sets $\lambda_\pi$ to be the labeling from the set $\{\lambda_{(H,\pi)}' ~\colon~ H \in {\cal H}_\pi\}$  that makes the properly labeled graph $G$ with this labeling the lexicographically smallest. 
The output labeling $\lambda_\pi$  is weakly isomorphism invariant in 
$(G,\pi)$ (and ${\cal F}$ and $k$) because the set ${\cal H}_\pi$ is isomorphism invariant in $G$ and $\pi$ (and ${\cal F}$ and $k$), each labeling $\lambda_{(H,\pi)}'$ is weakly isomorphism invariant in $G$, $H$ and $\pi$ (and ${\cal F}$ and $k$), and selecting a lexicographically first labeled graph from a set of labeled graphs is a weakly isomorphism invariant operation. 
The size of ${\cal H}_\pi$ and the time to compute it is upper bounded by a function of $k$ and ${\cal F}$. Thus, the total time of the algorithm us upper bounded by the time taken by  $|{\cal H}_\pi|$ invocations of ${\cal B}$.

\medskip

Now we recurse in two different ways based on whether the set $B$ is a clique separator of size at most $k$ or a $k$-atom. 
\paragraph{Case IV: $B$ is a clique separator in $\imp{\widehat{G}}{k}$.} The proof of this case is almost identical to Case I, except  that we use the fact that $D \subsetneq B$ to conclude that the instances we recurse on have a strictly smaller potential than the original instance. We repeat the argument for completeness. 

We recurse on $\widetilde{G}=G[N[C]]$ for each of the connected component $C$ of $G-B$. However, to do this we need to show that $\widetilde{G}$ satisfies the premise of Lemma~\ref{lem:cliqueToGeneralBoundaried}. We take the restriction of ${\sf col}_G$ to $\widetilde{G}$ as the coloring ${\sf col}_{\widetilde{G}}$. The set $\widetilde{D}=N(C)$ is defined to be the distinguished  set. The boundary of $\widetilde{G}$ is $N(C)$. By Lemma~\ref{lem:isoinvBag} we have that $|N(C)|\leq k$.  %
Let $\tilde{\iota} \colon \widetilde{D} \to [|\widetilde{D}|]$ be an arbitrary bijection.

We define $\widetilde{H}$ as $(G-C,\iota)\oplus (H,\iota)$.   We view $(\widetilde{H},\tilde{\iota})$ as a $[|\widetilde{D}|]$-boundaried graph with boundary $\widetilde{D}$ and the labelling $\tilde{\iota} \colon \widetilde{D} \to [|\widetilde{D}|]$, defined above. Observe that $(\widetilde{G},\tilde{\iota})\oplus (\widetilde{H},\tilde{\iota})$ and $(G,\iota)\oplus  (H,\iota)$ are the same graphs, just with different boundaries. This implies that $(\widetilde{G},\tilde{\iota})\oplus (\widetilde{H},\tilde{\iota})$ is ${\cal F}$-free, as $(G,\iota)\oplus (H,\iota)$  is ${\cal F}$-free. Finally, since $\widetilde{D}\subseteq D$ and $D$ is a clique in the $k$-improved graph of $(G,\iota)\oplus  (H,\iota)$, it follows that $\widetilde{D}$ is a clique in the $k$-improved graph of $(\widetilde{G},\tilde{\iota}) \oplus (\widetilde{H},\tilde{\iota})$. To apply the inductive argument it remains to show
 $\pot(\widetilde{G},\widetilde{D})<\pot(G,D)$. 
 However this follows directly from the fact that $B\supsetneq D$ and $V(\widetilde{G})\subseteq V(G)$.

Having justified the possibility of invoking the algorithm on $({\cal F},k,\widetilde{G},\widetilde{D})$,
 the invocation either outputs $\bot$ or it outputs for every labeling $\pi \colon \widetilde{D} \rightarrow [|\widetilde{D}|]$ a proper labeling $\widetilde{\lambda}_\pi$ of $(\widetilde{G}, \pi)$ such that $\widetilde{\lambda}_\pi(v) = \pi(v)$ for each $v\in \widetilde{D}$. The labeling $\widetilde{\lambda}_\pi$ is weakly isomorphism invariant in $(\widetilde{G}, \pi)$. If the algorithm outputs $\bot$, we output $\bot$ for the original instance $({\cal F},k,G,D)$ (and this is weakly isomorphism invariant).  Thus, we assume that the algorithm succeeds on each recursive call made on the connected component $C$ of $G'-B$. 

We first apply Lemma~\ref{lem:compact}  on $G$ and ${\sf col}_{G}$  and obtain an isomorphism invariant  compact coloring of $G$. We also compact the labelings obtained from applying recursive steps. To do this we again apply Lemma~\ref{lem:compact}   and obtain compact labelings. 

Next, we define an isomorphism invariant (with respect to $G$ and $D$) star decomposition $(T, \chi)$ of 
$G$. The central node is $r \in V(T)$ and the associated bag is $\chi(r)=B$. Further, for every connected component 
$C$ of $G-B$, there is a leaf $\ell\in V(T)$, such that $\chi(\ell)=N[C]$.  Using the output of the recursive 
calls, we define a leaf-labeling $\{\lambda_{\pi,\ell}\}$ of $(T, \chi)$. 
Precisely, for every component $C$ of $G-B$, say corresponding to a leaf $\ell$, and for each bijection $\pi\colon \sigma(\ell)\to [|\sigma(\ell)|]$, we set $\lambda_{\pi,\ell}=\widehat{\lambda}_\pi$, where $\widehat{\lambda}_\pi$ is a labeling produced by a recursive call on $C$.
By construction, this leaf labeling is locally weakly isomorphism invariant.  Now we apply Lemma~\ref{lem:smallCentersAlgorithm} to $(G, B,  (T,\chi,b,\{ \lambda_{\pi, \ell} \}))$  
and for every bijection $\pi \colon B \rightarrow [|B|]$,  obtain 
 a proper labeling $\lambda_\pi$ of $(G, \pi)$ such that $\lambda_\pi(v) = \pi(v)$ for all $v\in B$.
For each labeling $\pi' : D \rightarrow [|D|]$ the algorithm 
iterates over all bijections $\pi : B \rightarrow [|B|]$ that coincide with $\pi$ on $D$ and sets $\lambda_{\pi'}$ to be the labeling $\lambda_\pi$ of $(G, \pi')$ that makes the labeled graph $(G, \lambda_{\pi'})$ lexicographically smallest. 
The labeling $\lambda_{\pi'}$ is weakly isomorphism invariant in $(G, \pi')$ (and ${\cal F}$ and $k$), hence it can be output by the algorithm.
 The running time of this step is upper bounded by $|B|!n^{\cO(1)}=2^{\cO(k \log k)} n^{\cO(1)} $.

\paragraph{Case V: $B$ is a $k$-atom in $\imp{\widehat{G}}{k}$.}  In this case we recurse as follows. Let ${\cal N}$ be the family of all $X\subseteq B$ for which there exists a component $C$ of $G-B$ with $X=N(C)$. For each subset $X\in {\cal N}$, let $C_X$ be the union of all such components $C$. In this case we recurse on $\widetilde{G}=G[N[C_X]]$ for each $X\in {\cal N}$. However, to do this we need to show that $\widetilde{G}$ satisfies the premise of Lemma~\ref{lem:cliqueToGeneralBoundaried}. We take restriction of ${\sf col}_G$ to $\widetilde{G}$ as the coloring ${\sf col}_{\widetilde{G}}$.  The set $\widetilde{D}=X=N(C)=N(C_X)$ is defined to be the distinguished set for $\widetilde{G}$.  By Lemma~\ref{lem:isoinvBag} we have that $|X|\leq k$ for each $X\in {\cal N}$, and thus   $|\widetilde{D}|\leq k$.  
Let $\tilde{\pi} \colon \widetilde{D} \to [|\widetilde{D}|]$ be an arbitrary bijection from $\widetilde{D} \to [|\widetilde{D}|]$.

We define $\widetilde{H}$ as $(G-C_X,\pi)\oplus (H,\pi)$.   We view $(\widetilde{H},\tilde{\pi})$ as a $[|\widetilde{D}|]$-boundaried graph with boundary $\widetilde{D}$ and the labelling $\tilde{\pi} \colon \widetilde{D} \to [|\widetilde{D}|]$, defined above. Observe that $(\widetilde{G},\tilde{\pi})\oplus (\widetilde{H},\tilde{\pi})$ and $(G,\pi)\oplus (H,\pi)$ are the same graphs with possibly different boundaries. This implies that $(\widetilde{G},\tilde{\pi})\oplus (\widetilde{H},\tilde{\pi})$ is ${\cal F}$-free, as $(G,\pi)\oplus (H,\pi)$  is ${\cal F}$-free. Also, as $\widetilde{D}\subseteq D$ and $D$ is a clique in the $k$-improved graph of $(G,\pi)\oplus (H,\pi)$, it follows that $\widetilde{D}$ is a clique in the $k$-improved graph of $(\widetilde{G},\tilde{\pi})\oplus (\widetilde{H},\tilde{\pi})$. To apply the induction we need to show that 
 $\pot(\widetilde{G},\widetilde{D})<\pot(G,D)$.  
 Observe that, 
\begin{eqnarray*}
 \pot(\widetilde{G},\widetilde{D})   &\leq & 2( n-|B|)+ |\widetilde{D}|\\
&=&  2n-2|B|+|D|-|D|+|\widetilde{D}| \\
& = & 2n-|D|-(2|B|-|D|- |\widetilde{D}|) \\
& < & \pot(G,D)~~~~~~~~~~~~~~~~~~~~~~~~~~~~~~~~~~~~~~~~~~~~(\widetilde{D} \subseteq B \mbox{ and } D \subsetneq B )
 \end{eqnarray*}
 
Having justified the possibility of invoking the algorithm on $({\cal F},k,\widetilde{G},\widetilde{D})$, this invocation either outputs $\bot$ or it outputs for every labeling $\pi \colon \widetilde{D} \rightarrow [|\widetilde{D}|]$ a proper labeling $\widetilde{\lambda}_\pi$ of $(\widetilde{G}, \pi)$ such that $\widetilde{\lambda}_\pi(v) = \pi(v)$. The labeling $\widetilde{\lambda}_\pi$ is weakly isomorphism invariant in $(\widetilde{G}, \pi)$. If the algorithm outputs $\bot$, we output $\bot$ on the original instance $({\cal F},k,G,D)$ (and this is weakly isomorphism invariant).  Thus, we assume that the algorithm succeeds on each recursive call made on $C_X$ for $X\in {\cal N}$.

We first apply Lemma~\ref{lem:compact}  on $G$ and ${\sf col}_{G}$  and obtain an isomorphism invariant  compact coloring of $G$. We also compact the leaf labelings obtained from applying recursive steps. To do this we again apply Lemma~\ref{lem:compact}   and obtain compact labelings.

Next, we define an isomorphism invariant (with respect to $G$ and $D$) star decomposition $(T, \chi)$ of 
$G$. The central node is $r \in V(T)$ and the associated bag is $\chi(r)=B$.  Furthermore, for each $X\in {\cal N}$, we construct a leaf $\ell_X\in V(T)$ with $\chi(\ell_X)=N[C_X]$.  Using the output of the recursive 
calls, we define a leaf-labeling $\{ \lambda_{\pi, \ell} \}$ of $(T, \chi)$ as follows. For each $X\in {\cal N}$ and bijection $\pi\colon \sigma(\ell_X)\to [|\sigma(\ell_X)|]$, set $\lambda_{\pi,\ell}=\widetilde{\lambda}_{\pi}$, where $\widetilde{\lambda}_{\pi}$ was obtained from the recursive call on $C_X$. By construction, this leaf labeling is locally weakly isomorphism invariant.  Further, by construction we have that  $\chi(\ell) \cap \chi(r) \neq \chi(\ell') \cap \chi(r)$ for every pair of distinct leaves $\ell, \ell'$ in $T$, and for every leaf $\ell$ of $T$ and every connected component $C$ of $G[\chi(\ell) \setminus \chi(r)]$ it holds that $N(C) = \chi(\ell) \setminus \chi(r)$. Thus, the constructed star decomposition $(T, \chi)$ with central node $r$ is connectivity-sensitive.

Let $n_{\cal F}$ denote the maximum size of a graph (in terms of the number of vertices) in the given family of graphs  ${\cal F}$ (excluded as a topological minor).   
Next we apply the strong unpumping operation. It takes as input a colored graph $G$, a connectivity sensitive, weakly isomorphism invariant (in $G$) star decomposition $(T, \chi)$ of $G$ with central bag $b \in V(T)$, and a compact locally weakly isomorphism invariant leaf-labeling $\{\lambda_{\pi, \ell}\}$ of $(T, \chi)$. We apply the strong unpumping  operation  on  $(G, B,  (T,\chi),r,\{ \lambda_{\pi, \ell} \})$ and parameters $k$ and $f= n_{\cal F}$.

The operation produces a colored graph $G^ \star$ and a connectivity sensitive star decomposition $(T, \chi^\star)$ as follows.
$G^\star$ is obtained from $G$ by replacing 
for every leaf $\ell$ of $T$ the boundaried properly labeled graph $(G[\chi(\ell)], \pi(\ell))$ with labeling $\lambda\langle\ell\rangle$ by its $(k,n_{\cal F})$-strong representative. 
The star decomposition $(T, \chi^\star)$ of $G^\star$ has the same decomposition star $T$ as $G$, for the central node $r$ we set $\chi^\star(r) = \chi(r)$ and for every leaf $\ell$ of $T$ we set $\chi^\star(\ell)$ to be the vertex set of the representative used to replace $G[\chi(\ell)]$. Let $q={\eta}_s(k, n_{\cal F}, k)$. Let  ${\eta}_s(k, f, t)$ denote the smallest $\ell$ such that every $t$-boundaried graph has a $(k,f)$-strong representative on at most $\rho$ vertices. In other words $q$ denotes the size of a largest representative used to replace $G[\chi(\ell)]$.

By Property $5$ of \ref{lem:unpumpProperties}, we have that for every $\pi \colon D \rightarrow [|D|]$, and every boundaried graph $(H,\pi)$, the subgraph of the $k$-improved graph of $(G, \pi) \oplus (H,\pi)$ induced by $\chi(r)$ and the subgraph of the $k$-improved graph of $(G^\star, \pi) \oplus (H,\pi)$ induced by $\chi^\star(r)$ are the same.  This together with the fact that  there exist a labelling $\pi \colon D \rightarrow [|D|]$ and a  $[|D|]$-boundaried graph $(H,\pi)$ so that $(G,\pi) \oplus (H,\pi)$ is ${\cal F}$-free and $D$ is a clique in the $k$-improved graph of $(G,\pi) \oplus (H,\pi)$, implies that there exist a labelling $\pi \colon D \rightarrow [|D|]$ and a  $[|D|]$-boundaried graph $(H,\pi)$ so that $(G^\star,\pi) \oplus (H,\pi)$ is ${\cal F}$-free and $D$ is a clique in the $k$-improved graph of $(G^\star,\pi) \oplus (H,\pi)$.

Since  there exist a bijection $\pi \colon D \rightarrow [|D|]$ and a  $[|D|]$-boundaried graph $(H,\pi)$ so that $(G^\star,\pi) \oplus (H,\pi)$ is ${\cal F}$-free and $D$ is a clique in the $k$-improvement of $(G^\star,\pi) \oplus (H,\pi)$. This implies that {\em for every bijection} $\pi \colon D \rightarrow [|D|]$, there exists a  $[|D|]$-boundaried graph $(H,\pi)$ so that $(G^\star,\pi) \oplus (H,\pi)$ is ${\cal F}$-free and $D$ is a clique in the $k$-improvement of $(G^\star,\pi) \oplus (H,\pi)$. Indeed,  we could use the same $H$ but label the boundary vertices so that it is consistent with $\pi$. Let $n_\pi$ 
denote the cardinality of the vertex set of a smallest graph $H$ such that  $(H,\pi)$ is a $[|D|]$-boundaried graph and $(G^\star,\pi) \oplus (H,\pi)$ is ${\cal F}$-free and $D$ is a clique in the $k$-improved graph of 
$(G^\star,\pi) \oplus (H,\pi)$. Define ${\cal H}_\pi$ to be the set of all boundaried graphs $(H, \pi)$ on $n_\pi$ vertices such that $(G^\star,\pi) \oplus (H,\pi)$ is ${\cal F}$-free and $D$ is a clique in the $k$-improvement of $(G^\star,\pi) \oplus (H,\pi)$. 

The algorithm constructs  ${\cal H}_\pi$ from $(G^\star, \pi)$ by enumerating all boundaried graphs $(H, \pi)$ in increasing order of $|V(H)|$, testing for each one whether $(G^\star,\pi) \oplus (H,\pi)$ is ${\cal F}$-free, using 
the algorithm of \cite{GroheKMW11} for testing a topological minor in a graph,   
and $D$ is a clique in the $k$-improvement of $(G^\star,\pi) \oplus (H,\pi)$, and stopping after enumerating all the graphs $H$ on $x$ vertices, where $x=n_\pi$ is the smallest number for which the algorithm finds an $H$ for which this test succeeds.

\begin{claim}\label{clm:caseVatomisUnbreakable}
For every  $(H,\pi) \in {\cal H}_\pi$ we have that $(G^\star,\pi) \oplus (H,\pi)$ is $(2^{k}q^\star,k,k)$-improved-clique-unbreakable. Here, $q^\star= \max\{q,n_\pi\}$. 
\end{claim}
\begin{proof}
Let $\tilde{G}=(G^\star,\pi) \oplus (H,\pi)$ and $G^\dagger=\imp{\tilde{G}}{k}$.  Towards a proof we first show that $G^\dagger[B]$ is a $k$-atom.   Note that since $D$ is a 
clique in $\imp{{((G,\pi) \oplus (H,\pi))}}{k}$,  by Lemma~\ref{lem:indImprovement}  we have that
 $\imp{\widehat{G}}{k}=\imp{{((G,\pi) \oplus (H,\pi))}}{k}[V(G)]$. This together with the fact that $B$ is a $k$-atom in $\imp{\widehat{G}}{k}$, we have that $B$ is a $k$-atom  in $\imp{{((G,\pi) \oplus (H,\pi))}}{k}[V(G)]$.  Now by Property $5$ of \ref{lem:unpumpProperties}, we have that the subgraph of the $k$-improved graph of $(G, \pi) \oplus (H,\pi)$ induced by $\chi(r)$ ($\imp{{((G,\pi) \oplus (H,\pi))}}{k}[\chi(r)])]$) and the subgraph of the $k$-improved graph of $(G^\star, \pi) \oplus (H,\pi)$ induced by $\chi^\star(r)$ ($G^\dagger[\chi^\star(r)]$) are the same.  Since $\chi(r)= \chi^\star(r)=B$, and $B$ is a $k$-atom in $\imp{{((G,\pi) \oplus (H,\pi))}}{k}$, we have that $G^\dagger[B]$ is a $k$-atom.

Next we would like to design a  star decomposition $(T',\chi')$ of $G^\dagger=\imp{\tilde{G}}{k}$. We start with the star decomposition $(T,\chi^\star)$ of  $G^\star$ that satisfies the following properties.
\begin{enumerate}
\setlength{\itemsep}{-1pt}
\item $(T,\chi^\star)$ is connectivity-sensitive. 
\item The central bag $\chi^\star(r)$ is  a $k$-atom in $G^\dagger$. 
\item For all $t\in V(T)$, $\sigma^\star(t)\leq k$. 
\item For all leaves $\ell \in V(T)$, $|\chi^\star(\ell)|\leq q$. 
\end{enumerate}
We check if for some $t\in V(T)$, $\sigma^\star(t)=V(H)\cap \chi^\star(r)$. If yes, then we define 
$\chi^\star(t):=\chi^\star(r) \cup V(H)$. Else, we add another leaf $\ell_H $ to $T$ and get a new tree $T'$ and define $\chi^\star(\ell_H)=V(H)$. Let 
 $(T',\chi^\star)$ be a star decomposition of $(G^\star,\pi) \oplus (H,\pi)$.  By construction this is  connectivity-sensitive.  Now consider $G^\dagger=\imp{\tilde{G}}{k}$. Observe that, since for each leaf $t\in V(T')$ we have $\sigma^\star(t)\leq k$, it follows that for each of the  {\em newly added edges} $uv$ in $E(\imp{\tilde{G}}{k})$, there exists  $t\in V(T')$ such that $\{u,v\}\subseteq \chi^\star(t)$.  This implies that $(T,\chi'=\chi^\star)$  is also a star decomposition of $G^\dagger$.  Thus, by Lemma~\ref{lem:unbreakbleAndstarDeco} we have that $G^\dagger$  is $(2^k q^\star,k)$-clique unbreakable. This implies that $(G^\star,\pi) \oplus (H,\pi)$ is $(2^{k}q^\star,k,k)$-improved-clique-unbreakable. This concludes the claim. 
\cqed\end{proof}

We first apply Lemma~\ref{lem:compact}  on $(G^\star,\pi) \oplus (H,\pi)$ and ${\sf col}_{(G^\star,\pi) \oplus (H,\pi)}$  and obtain an isomorphism invariant  compact coloring of $(G^\star,\pi) \oplus (H,\pi)$.  
Then, 
the algorithm calls ${\cal B}$ on $((G^\star,\pi) \oplus (H,\pi),\pi)$, ${\cal F}$, $k$ and $n_\pi$ for every $(H,\pi)\in {\cal H}_\pi$ and obtains for each  $(H,\pi)$ a weakly isomorphism invariant proper labeling $\lambda_{(H,\pi)}$ of  $(G^\star,\pi) \oplus (H,\pi)$.  
(Technically the algorithm ${\cal B}$ takes as input a normal graph, rather than a boundaried graph. Therefore, instead of passing $((G^\star,\pi) \oplus (H,\pi),\pi)$ directly to ${\cal B}$ the algorithm passes the compact color encoding of $((G^\star,\pi) \oplus (H,\pi),\pi)$ instead.) 
By Claim~\ref{clm:caseVatomisUnbreakable} the graph $(G^\star,\pi) \oplus (H,\pi)$ is $(2^{k}q^\star,k,k)$-improved clique unbreakable and therefore this is a well-formed input for the algorithm ${\cal B}$.
If ${\cal B}$ fails (returns $\bot$) then our algorithm fails as well (and this is weakly isomorphism invariant). However, assuming that $k \geq \kappa_{\cal B}$ none of the calls to ${\cal B}$ fail. 
For each $H$ let $\lambda_{(H,\pi)}'$ be the restriction of $\lambda_{(H,\pi)}'$ to the domain $V(G^\star)$.
The algorithm sets $\lambda_\pi^\star$ to be the labeling from the set $\{\lambda_{(H,\pi)}' ~\colon~ H \in {\cal H}_\pi\}$  that makes the properly labeled graph $G^\star$ with this labeling the lexicographically smallest. 
The labeling $\lambda_\pi^\star$  is weakly isomorphism invariant in $(G^\star,\pi)$ because the set ${\cal H}_\pi$ is isomorphism invariant in $G^\star$ and $\pi$, each labeling $\lambda_{(H,\pi)}'$ is weakly isomorphism invariant in $G^\star$, $H$ and $\pi$, and selecting a lexicographically first properly labeled graph from a set of properly labeled graphs is a weakly isomorphism invariant operation. 
The size of ${\cal H}_\pi$ and the time to compute it is upper bounded by a function of $k$ and ${\cal F}$. Thus, the total time of the algorithm up to this point is upper bounded by the time taken by  $|{\cal H}_\pi|$ invocations of ${\cal B}$.

For each permutation $\pi \colon D \rightarrow [|D|]$ the algorithm applies the lifting procedure (see Section~\ref{sec:unpumping}) on the lifting bundle $(G, D, (T, \chi), b, \{\lambda_{\pi, \ell}\}, G^\star, (T, \chi^\star), \iota, \lambda_\pi^\star)$ and and obtains a proper labeling $\lambda_\pi$ of $(G, \pi)$. By Lemma~\ref{lem:lem:caNlifting},  $\lambda_\pi$ is weakly isomorphism invariant in $(G, \pi)$ (and ${\cal F}$ and $k$), and by Observation~\ref{obs:liftRunTime} the procedure takes time polynomial in the input. We call the lifting procedure $|D|!$ times, one for each $\pi : D \rightarrow [|D|]$. This concludes the running time analysis of this step.

Finally, we analyze the running time of the whole algorithm. The algorithm is a recursive algorithm and hence has a recursion tree $\cal T$. We have already proved that in each node of the recursion tree the algorithm spends time at most  $g(k,{\cal F})n^{\cO(1)}$ and makes at most $g(k,{\cal F})n^{\cO(1)}$ calls to ${\cal B}$. Thus, to prove the claimed upper bound on the running time it suffices to upper bound the number of nodes in ${\cal T}$ by $\cO(n^2)$.

First, observe that in every node of the recursion tree the algorithm only makes calls to instances with strictly lower potential value, where the potential is  $\pot(G,D)= 2n-|D| \leq 2n$. Since the potential value is non-zero for every node of the recursion tree, the length of the longest root-leaf path in ${\cal T}$ is at most $2n$.

A {\em level} of the recursion tree is a set of nodes of ${\cal T}$ at the same distance from the root. For every level we upper bound the number of nodes in that level by $\cO(n)$. Towards this goal tag each node of ${\cal T}$ with the graph $G$ and a set $B$ (containing distinguished set $D$) that is processed in that node. A vertex $v \in V(G) \setminus B$ is called a {\em non-boundary} vertex. 

An inspection of the algorithm reveals the following observations: (i) every node of ${\cal T}$ has at least one non-boundary vertex, (ii) a non-boundary vertex of a node is also a non-boundary vertex of its parent node, (iii) for every node $(G, B)$ of ${\cal T}$ and every non-boundary vertex $v \in V(G) \setminus B$, $v$ is a non-boundary node in at most one child of the call $(G, B)$. The observations (i), (ii) and (iii) together yield that the size of every level of the recursion tree is upper bounded by $n$. This concludes the proof.
\end{proof}

%% file: unbreakabilityTools.tex
\subsubsection{$k$-Important Separators and Important Separator Extension}
We will need the following variation of important $X$-$Y$ separators. An important $X-Y$-separator $S$ of size at most $k$ is called a $k$-important $X-Y$-separator if there does not exist any other important separator $S'$ of size at most $k$ such that $R_{G\setminus S}(X \setminus S) \subseteq R_{G\setminus S'}(X \setminus S')$.

\begin{lemma}\label{lem:impSepBasisChange}
Let $S$ be a $k$-important $\{v\}-D$ separator distinct from $\{v\}$ and suppose $u \in R_{G \setminus S}(v)$. Then $S$ is a $k$-important $\{u\}-D$ separator.
\end{lemma}

\begin{proof}
Clearly $S$ separates $u$ from $D$ (since it separates $v$ from $D$). Further, if some $S'$ of size at most $k$ separates $u$ from $D$ and $R_{G \setminus S}(u) \subset R_{G \setminus S'}(u)$ then  $R_{G \setminus S}(v) = R_{G \setminus S}(u) \subset R_{G \setminus S'}(u) = R_{G \setminus S'}(v)$ contradicts that $S$ is a $k$-important $\{v\}-D$-separator.
\end{proof}

Let $G$ be a graph, $D \subset V(G)$ be a vertex set, and $k$ be an integer. We define the $k$-{\em important separator extension} of $D$ to be the set $A$ of all vertices $u$ that satisfy at least one of the following two conditions. 
{\em (i)} There exists a $v \neq u$ and a $\{v\}$-$D$ $k$-important separator $S$ that contains $u$.
{\em (ii)} $\{u\}$ is the unique $\{u\}$-$D$ separator of size at most $k$. The definition of the $k$-important separator extension of $D$ immediately implies that it is isomorphism invariant in $(G, D)$, we state this as a lemma without proof. 
\begin{lemma}\label{lem:extensionIsIsoInvariant} The $k$-important separator extension $A$ of $D$ is isomorphism invariant in $(G, D)$.
\end{lemma}
Next we observe that the definition of the the $k$-important separator extension of $D$, together with Lemmata~\ref{lem:imp-seps} and~\ref{lem:f-f} gives rise to an efficient algorithm to compute it. 
We note that Lemmata~\ref{lem:impSepExtensionAlgorithm} and~\ref{lem:extSmallComponentBoundary} are essentially re-statements of Lemmata 13 and 19 of~\cite{CyganLPPS19} in our terminology. We include their proofs for completeness. 

\begin{lemma}\label{lem:impSepExtensionAlgorithm}
There exists an algorithm that takes as input a graph $G$, vertex set $D$ and integer $k$ and computes the $k$-important separator extension of $D$ in time $\cO(16^kn^{\cO(1)})$.
\end{lemma}

We remark that both the dependence on $k$ and on $n$ in the running time of Lemma~\ref{lem:impSepExtensionAlgorithm} can be substantially improved, but we opt for a simpler and slower algorithm.

\begin{proof} 
Initialize $A = \emptyset$. For every vertex $v$ we can list the set of all important $\{v\}-D$ separators of size at most $k$ in time $\cO(16^kk(n+m))$ using Lemma~\ref{lem:imp-seps}.

If $\{v\}$ is the unique such separator then add $v$ to $A$. 
For every  $\{v\}-D$ important separator $S$ of size at most $k$, iterate over all other $\{v\}-D$ important separator $S'$ of size at most $k$, and check whether $R_{G\setminus S}(X \setminus S) \subseteq R_{G\setminus S'}(X \setminus S')$. If no such $S'$ exists then $S$ is a $k$-important $\{v\}-D$ separator, we add all vertices in $S$ to $A$. Finally output $A$ as the $k$-important separator extension of $D$.

Since the algorithm explicitly checks for every vertex $v$ whether it satisfies condition (i) or (ii) of $k$-important separator extensions and inserts $v$ into $A$ if it satisfies at least one of these, correctness follows. For each vertex $v \in V(G)$ the algorithm spends time at most $\cO(16^kn^{\cO(1)})$ (looping over every pair of important separators in a list of size at most $4^k$, spending polynomial time per pair), yielding the claimed running time bound. 
\end{proof}

\begin{lemma}\label{lem:extSmallComponentBoundary}
Let $A$ be the $k$-important separator extension of $D$. Then, for every connected component $C$ of $G \setminus A$ we have $|N(C)| \leq k^24^k$.
\end{lemma}

\begin{proof}
Let $c$ be a vertex in $C$. Since $c \notin A$ there exists at least one $k$-important $c$-$D$ separator $S$ distinct from $\{c\}$. We have that $S \subseteq A$. We now prove the following claim: for every $v \in N(C) \setminus S$ there exists a vertex $u \in S$ such that $v$ is contained in a $k$-important $\{u\}$-$D$-separator. 

Since $v \in A$ and $S$ is a separator of size at most $k$ that separates $v$ from $D$ it follows that there exists a $k$-important $w$-$D$-separator $S'$ such that $w \neq v$ and $v \in S'$. 
Let $C'$ be the component of $G \setminus S'$ that contains $w$. Since $S'$ is an inclusion wise minimal $w-D$ separator, $C'$ contains a neighbor $v'$ of $v$. By Lemma~\ref{lem:impSepBasisChange}, $S'$ is a $k$-important $\{v'\}$-$D$-separator. Thus, if $v' \in S$ the statement of the claim holds with $u = v'$. Thus, suppose $v' \notin S$.
Since  $c$ can reach $v$ in $G \setminus S$ and $v$ and $v'$ are neighbors, it follows that
$v' \in R_{G \setminus S}(c)$. Since $S'$ is a $k$-important $\{v'\}$-$D$-separator and $|S| \leq k$ it follows that $R_{G\setminus S'}(v') \nsubseteq R_{G\setminus S}(v')$. Thus, since $R_{G\setminus S'}(v')$ is connected there exists a path $P$ disjoint from $S'$ that starts in $v'$ and ends outside of $R_{G\setminus S}(v')$. Let $u$ be the first vertex on $P$ outside of $R_{G\setminus S}(v')$, we have that $u \in S$. Finally, since $u$ and $v'$ are in the same component of $G \setminus S'$ it follows that
$S'$ is a $k$-important $\{u\}$-$D$-separator, as claimed. 

Having proved the claim the statement of the Lemma is immediate, as $|S| \leq k$, and for every $u \in S$ there are at most $k \cdot 4^k$ vertices that are contained in at least one $k$-important $\{u\}$-$D$-separator. 
\end{proof}

\subsubsection{Unbreakablity-Tied Sets}

Let $D$ and $A$ be vertex sets, $t$ be a positive integer and $q : \mathbb{N} \rightarrow \mathbb{N}$ or $q : [t]  \rightarrow \mathbb{N}$ be a
function. We will say that $A$ is $q$-unbreakability-tied to $D$ if every separation $(L, R)$ of order $i$, where $i$ is in the domain of $q$, such that $D \subseteq R$, satisfies $|L \cap A| \leq q(i)$.
If $V(G)$ is  $q$-unbreakability-tied to $D$ we will say that $G$ is $q$-unbreakability-tied to $D$. 
The intuition is that if $A$ is unbreakability-tied to $D$ then every separation of small order that splits $A$ in two large parts also has to split $D$ in two large parts. Thus, if $D$ is unbreakable then $A$ inherits some of the unbreakability of $D$. We now formalize this intuition.

\begin{lemma}\label{lem:tiedPlusUnbreakGivesUnbreak}
Let $t_2$ be a positive integer and $q_1 : \mathbb{N} \rightarrow \mathbb{N}$ and $q_2 : [t_2] \rightarrow \mathbb{N}$ be non-decreasing functions such that $1 + q_1(1) \leq t_2$.
If $G$ is a graph, $D$ is a $q_1$-unbreakable vertex set in $G$, and $A$ is a vertex set that is $q_2$-unbreakability tied to $D$, then $A$ is $q_3$-unbreakable, where  $q_3 : [t_3] \rightarrow \mathbb{N}$ is defined as $q_3(i) = q_2(i + q_1(i))$ and $t_3$ is the largest integer $j$ such that $j + q_1(j) \in [t_2]$.
\end{lemma}

\begin{proof}
Let $(L, R)$ be a separation of order $i \leq t_3$. Since $D$ is $q_1$-unbreakable, either $|L \cap D| \leq q_1(i)$ or $|R \cap D| \leq q_1(i)$. Without loss of generality we have $|L \cap D| \leq q_1(i)$. 
Let $R' = R \cup D$, we have that $(L, R')$ is a separation of order $i + q_1(i) \leq t_2$ and $D \subseteq R'$. Thus, since $A$ is $q_2$-unbreakability tied to $D$ it follows that $|A \cap L| \leq q_2(i + q_1(i)) = q_3(i)$ as claimed. 
\end{proof}

\begin{lemma}\label{lem:extensionIsTied}
Let $A$ be the $k$-important separator extension of $D$ and $q : [k] \rightarrow \mathbb{N}$ be $q(i) = i \cdot (12^i + 1)$. Then $A$ is $q$-unbreakability tied to $D$.
\end{lemma}

\begin{proof}
Let $(L, R)$ be a separation of order $i \leq k$ such that $D \subseteq R$. 
Let $S = L \cap R$.
For every vertex $u \in A \cap (L \setminus R)$ we have that $S$ is a $\{u\}$-$D$-separator of size at most $k$, and $u \notin S$. It follows that there exists a $k$-important $\{u\}-D$ separator distinct from $\{u\}$, and therefore, since $u \in A$ there must exist a $v \neq u$ and a $k$-important $\{v\}$-$D$-separator $X$ such that $u \in X$. Since $X$ is a minimal $v$-$D$ separator, $R_{G \setminus X}(v)$ contains a neighbor $v'$ of $u$, and $X$ is also a $k$-important $\{v'\}$-$D$-separator by Lemma~\ref{lem:impSepBasisChange}. Thus, without loss of generality we may assume that $v \in N(u)$. Since $u \in L \setminus R$ it follows that $v \in L$. Define $S_X = S \cap X$, $S_Y = S \cap R_{G \setminus X}(v)$, $S_N = S \setminus (S_X \cup S_Y)$ and $X_L = X \setminus R$ and $X_R = X \setminus L$.

We claim that $|X \cap L| \leq i$. Since $S_N$, $S_X$ and $S_Y$ form a partition of $S$ we have that $|S_N| + |S_X| \leq i$. Observe now that $X' = (X \setminus X_L) \cup S_N$ is a $\{v\}$-$D$ separator, and that $R_{G \setminus X}(v) \subseteq R_{G \setminus X'}(v)$. Since $X$ is a $k$-important $\{v\}$-$D$ separator we have that $|X| \leq k \leq |X'|$ and therefore $|X_L| \leq |S_N|$. Putting the two inequalities involving $|S_N|$ together we obtain that
$$|X \cap L| = |X_L| + |S_X| \leq |S_N| + |S_X| \leq i$$
This proves $|X \cap L| \leq i$ as claimed. 

We now claim that $X_L$ is an important $S_Y$-$S_N$ separator in $G[L \setminus S_X]$. First, it is a $S_Y$-$S_N$ separator, because otherwise there is a path from $S_Y$ to $S_N$ in $G-X$ contradicting that $S_Y \subseteq R_{G \setminus X}(v)$ and $S_N \cap R_{G \setminus X}(v) = \emptyset$.
Second, $X_L$ is important: suppose for contradiction that there exists another $S_Y$-$S_N$ separator in $G[L \setminus S_X]$, say $X_L'$, such that $|X_L'| \leq |X_L|$ and $R_{G[L] \setminus (S_X \cup X_L)}(S_Y) \subset R_{G[L] \setminus (S_X \cup X_L')}(S_Y \setminus X_L')$. Consider the set $X' = X_R \cup S_X \cup X_L'$. We have that $|X'| = |X_R| + |S_X| + |X_L'| \leq |X_R| + |S_X| + |X_L| = |X|$. We claim that $X' = X_R \cup S_X \cup X_L'$ is a $\{v\}-D$ separator, and that $R_{G \setminus X}(v) \subset R_{G \setminus X'}(v)$.
To see that $R_{G \setminus X}(v) \subset R_{G \setminus X'}(v)$ observe that
\begin{align*}
R_{G-X}(v) & = R_{G \setminus X}(S_Y) \\
& = R_{G[L] \setminus (X_L \cup S_X)}(S_Y) \cup R_{G[R] \setminus (X_R \cup S_X)}(S_Y) \\
& \subset R_{G[L] \setminus (X_L' \cup S_X)}(S_Y) \cup R_{G[R] \setminus (X_R \cup S_X)}(S_Y) \\
& = R_{G \setminus X'}(S_Y)
\end{align*}
Since $\{v\} \cup S_Y \subseteq R_{G-X}(v) \subseteq R_{G-X'}(S_Y)$ all vertices in $\{v\} \cup S_Y$ are in the same component of $G-X'$. Thus $R_{G-X'}(S_Y) = R_{G-X'}(v)$. We conclude that $R_{G \setminus X}(v) \subset R_{G \setminus X'}(v)$.
Furthermore, $$R_{G \setminus X'}(v) = R_{G[L] \setminus (X_L' \cup S_X)}(S_Y) \cup R_{G[R] \setminus (X_R \cup S_X)}(S_Y).$$ We have that 
$R_{G[L] \setminus (X_L' \cup S_X)}(S_Y) \cap D = \emptyset$, 
because $D \cap L \subseteq S_N$ and $X_L'$ is a $S_Y$-$S_N$ separator in $G \setminus S_X$. Finally, 
$R_{G[R] \setminus (X_R \cup S_X)}(S_Y) \cap D = \emptyset$
because $X$ is a $\{v\}$-$D$-separator. Thus $X'$ is also a $\{v\}-D$ separator, as claimed. 
But now $X'$ is a $\{v\}-D$ separator of size at most $|X|$ such that $R_{G \setminus X}(v) \subset R_{G \setminus X'}(v)$, contradicting that $X$ is an important $\{v\}$-$D$ separator. Hence we can conclude that $X_L$ is an important $S_Y$-$S_N$ separator in $G[L \setminus S_X]$. 

We now have that every vertex $u \in A \setminus R$ is contained in an important $S_Y$-$S_N$ separators of size at most $i$ in $G[L] - \setminus S_X$ for some partition $S_Y$, $S_N$, $S_X$ of $S$. 
However, there are at most $3^i$ partitions of $S$ into $S_Y$, $S_N$, and $S_X$, and at most $4^i$ important $S_Y$-$S_N$ separators of size at most $i$ (by Lemma~\ref{lem:imp-seps}) in $G[L] \setminus S_X$ for each such partition. Therefore  $|A \setminus R| \leq 3^i \cdot 4^i \cdot i + i \leq i \cdot 12^i$. We conclude that  $|A \cap L| \leq i \cdot (12^i + 1)$, as claimed. 
\end{proof}

\begin{lemma}\label{lem:concludeGstarUnbreakable}
Let $G$ be a graph, $k$ be an integer, $q : [k] \rightarrow \mathbb{N}$ and $q' : \mathbb{N} \rightarrow \mathbb{N}$ be non-decreasing functions. If $G$ has a connectivity sensitive star decomposition $(T, \chi)$ with central bag $b$ such that $\chi(b)$ is $q$-unbreakable and every leaf $\ell$ of $T$ is $q'$-unbreakability tied to $\sigma(\ell)$, then $G$ is $q^\star$-unbreakable, where $q^\star : [k] \rightarrow \mathbb{N}$ is defined as $q^\star(i) = q(i) + (2^{q(i)} + i) \cdot q'(q(i) + i)$.
\end{lemma}

\begin{proof}
Let $(L, R)$ be a separation of order $i$. Since $\chi(b)$ is $q$-unbreakable we have that at least one of $|\chi(b) \cap L|$ and $|\chi(b) \cap R|$ is at most $q(i)$. Without loss of generality we assume $|\chi(b) \cap L| \leq q(i)$. We now prove that $|L| \leq q^\star(i)$.

We first prove that for every leaf $\ell$, at most $q'(q(i) + i)$ vertices of $\chi(\ell)$ are in $L$. Indeed, $(L, R \cup (L \cap \chi(b))$ is a separation of order at most $i + q(i)$ such that $\sigma(\ell) \subseteq R$. But $\chi(\ell)$ is $q'$-unbreakability tied to $\sigma(\ell)$, and therefore $|\chi(\ell) \cap L| \leq q'(q(i) + i)$. 

Next we prove that at most $2^{q(i)} + i$ leaves $\ell$ have a vertex in $L \cap (\chi(\ell) \setminus \sigma(\ell))$. Since $(T, \chi)$ is connectivity sensitive there are at most $2^{q(i)}$ leaves $\ell$ of $T$ such that $\sigma(\ell) \subseteq L$. All remaining leaves $\ell$ satisfy $\sigma(\ell) \cap (R \setminus L) \neq \emptyset$. 

We claim that at most $i$ leaves have both $(\chi(\ell) \setminus \sigma(\ell)) \cap L \neq \emptyset$ and $\sigma(\ell) \cap (R \setminus L) \neq \emptyset$. 
Towards a proof of this claim let $u \in \sigma(\ell) \cap (R \setminus L)$ and  $v \in (\chi(\ell) \setminus \sigma(\ell)) \cap L$ be vertices. 
Since $T, \chi$ is connectivity sensitive there exists a $v$-$u$ path $P$ whose all vertices except for $u$ are in $(\chi(\ell) \setminus \sigma(\ell))$.
But $v \in L$ and $u \in R \setminus L$, $(L, R)$ is a separation and therefore some vertex on $P$ other than $u$ is in $L \cap R \cap (\chi(\ell) \setminus \sigma(\ell))$. 
However $|L \cap R| \leq i$ while the sets $(\chi(\ell) \setminus \sigma(\ell))$ are disjoint for distinct leaves $\ell$, and so the claim follows. Thus there are at most $2^{q(i)} + i$ leaves $\ell$ that have a vertex in $L \cap (\chi(\ell) \setminus \sigma(\ell))$. We conclude that $|L| \leq q(i) + (2^{q(i)} + i) \cdot q'(q(i) + i) = q^\star(i)$.
\end{proof}

\subsubsection{Unbreakability of Weak Representatives}

\begin{lemma}\label{lem:unbreakabilityTieSmallLeaves}
Let $k$ and $s$ be positive integers, $q : [k] \rightarrow \mathbb{N}$ be a non-decreasing function, $G$ be a graph, $D \subseteq V(G)$, and $(T, \chi)$ be a connectivity-sensitive star decomposition of $G$ with central bag $b$ such that $\chi(b)$ is $q$-unbreakability tied to $D$ and every leaf $\ell$ of $T$ satisfies $\chi(\ell) \leq s$. Then $G$ is $q'$-unbreakability tied to $D$, where $q'(i) = (2^{q(i)} + i)s$ for all $i \in [k]$.
\end{lemma}

\begin{proof}
Consider a separation $(L, R)$ of $G$ of order at most $i \leq k$ such that $D \subseteq R$. We have that $|\chi(b) \cap L| \subseteq q(i)$. Since $(T, \chi)$ is connectivity sensitive, different leaves $\ell$ have different $\sigma(\ell)$, and each $\sigma(\ell)$ is a subset of $\chi(b)$. There are at most $2^{q(i)}$ leaves $\ell$ of $T$ such that $\sigma(\ell) \subseteq L$, and the sum of their sizes is at most $2^{q(i)}s$.

Every other leaf $\ell$, satisfies $\sigma(\ell) \setminus L \neq \emptyset$. We claim that at most $i$ such leaves $\ell$ can have $(\chi(\ell) \setminus \chi(b)) \cap L \neq \emptyset$. Indeed, let $u \in (\chi(\ell) \setminus \chi(b)) \cap L$ and $v \in \sigma(\ell) \setminus L$. Then there is a $u$-$v$ path $P$ contained entirely in $(\chi(\ell) \setminus \chi(b)) \cup \{v\}$. However, since $u \in L$ and $v \in R \setminus L$ it follows that some vertex $w$ on $P$ other than $v$ is in $L \cap R$. However $w \in (\chi(\ell) \setminus \chi(b))$ and the sets $(\chi(\ell) \setminus \chi(b))$ are pairwise disjoint for different leaves $\ell$. Since $|L \cap R| \leq i$ the claim follows. We conclude that $|L| \leq (2^{q(i)} + i)s$, completing the proof of the lemma.
\end{proof}

\begin{lemma}\label{lem:repsAreUnbreakabilityTiedToBoundary}
Let $(G, \iota)$ be a boundaried graph, $h$ be an integer and $(G_R, \iota_R)$ be its $h$-weak representative. 
Then $G_R$ is $q$-unbreakability tied to $\delta(G_R)$,
for  $q : \mathbb{N} \rightarrow \mathbb{N}$ defined as $q(i) = (2^{i(12^i + 1)} + 1)\eta_w(i^2 4^i, h)$
\end{lemma}

\begin{proof}
Let $D = \delta(G_R)$ and let $i$ be an arbitrary positive integer. 
Let $A$ be the $i$-important separator extension of $D$ in $G_R$. By Lemma~\ref{lem:extensionIsIsoInvariant}, $A$ is isomorphism invariant in $G_R$ and $D$, and by Lemma~\ref{lem:extensionIsTied}, $A$ is $j \cdot (12^j + 1)$-unbreakability tied to $D$ for  $j \in [i]$. By Lemma~\ref{lem:extSmallComponentBoundary}, every connected component $C$ of $G_R \setminus A$ satisfies $|N(C)| \leq i^2 4^i$. Let $(T, \chi)$ be the unique connectivity sensitive star decomposition of $G_R$ with center bag $\chi(b) = A$ (see Observation~\ref{obs:uniqueStarDec}). 

We claim that every leaf $\ell$ of $T$ satisfies $|\chi(\ell)| \leq \eta_w(i^2 4^i, h)$. Suppose not, and define for every leaf $\ell$ and permutation $\pi : \sigma(\ell) \rightarrow [|\sigma(\ell)|]$ the proper labeling $\lambda_{\pi, \ell}$ of $G[\chi(\ell)]$ to be the proper labeling that makes the properly labeled boundaried graph $(G[\chi(\ell)], \pi)$ lexicographically smallest among all compact labelings that coincide with $\pi$ on $\sigma(\ell)$.
We have that for every leaf $\ell$ of $T$ and permutation $\pi : \sigma(\ell) \rightarrow [|\sigma(\ell)|]$ that $\lambda_{\pi, \ell}$ is weakly isomorphism invariant in $(G[\chi(\ell)], \pi)$, and therefore the set $\{\lambda_{\ell,\pi}\}$ is a compact locally weakly isomorphism invariant leaf labeling of $(T, \chi)$.
Thus $G_R$, $D$, $(T, \chi)$, $\{\lambda_{\pi, \ell}\}$ of $(T, \chi)$ is a well-formed input to the unpumping operation (if the coloring of $G_R$ is not compact we color-compact the coloring of $G_R$ before unpumping). Weak unpumping with parameter $h$ produces a graph $G_R^\star$ and star decomposition $(T, \chi^\star)$.
By construction (since no edges are added or removed between vertices in $\chi(b)$, and no colors in $\chi(b)$ are changed) $(G_R, \iota_R)$ and $(G_R^\star, \iota_R)$ are boundary-consistent. By Lemma~\ref{lem:unpumpProperties} they have the same weak cut signature and $h$-folio.
By Lemma~\ref{lem:unpumpSameExtPerm} they have the same set of extendable permutations. 

It follows that $(G_R, \iota_R) \equiv (G_R^\star, \iota_R)$ where $\equiv$ is the equivalence relation defined in the proof
of Lemma~\ref{lem:equivalentWeakGadget}. Since $(G_R, \iota_R)$ is a (minimum size) representative of the equivalence relation $\equiv$ it follows that $|V(G_R)| \leq |V(G_R^\star)|$.
However, 
$$|V(G_R)| = |\chi(b)| + \sum_{\ell \in V(T) \setminus \{b\}} |\chi(\ell) \setminus \chi(b)|,$$
$$|V(G_R^\star)| = |\chi^\star(b)| + \sum_{\ell \in V(T) \setminus \{b\}} |\chi^\star(\ell) \setminus \chi^\star(b)|,$$ 
$|\chi(b)| \geq |\chi^\star(b)|$ (because unpumping leaves $\chi(b)$ unchanged), and for every leaf $\ell$ we have $|\chi(\ell) \setminus \chi(b)| \geq |\chi^\star(\ell) \setminus \chi^\star(b)|$ (by Lemma~\ref{lem:equivalentWeakGadget}).
We assumed for contradiction that there is at least one leaf $\ell$ such that $|\chi(\ell)| > \nu(i^2 4^i, h) \geq \nu(|\sigma(\ell)|, h)$. For this leaf the inequality $|\chi(\ell) \setminus \chi(b)| \geq |\chi^\star(\ell) \setminus \chi^\star(b)|$ is strict, implying that $|V(G_R^\star)| < |V(G_R)|$, contradicting that  $|V(G_R)| \leq |V(G_R^\star)|$.

We have that $(T, \chi)$ is a star decomposition with central bag $b$ such that $\chi(b)$ is 
$q$-unbreakability tied to $D$ for $q : [i] \rightarrow \mathbb{N}$ defined as $q(j) = 
j \cdot (12^j + 1)$. Furthermore every leaf $\ell$ of $T$ satisfies $|\chi(\ell)| \leq \eta_w(i^2 4^i, h)$.

Lemma~\ref{lem:unbreakabilityTieSmallLeaves} now implies that $G_R$ is $q'$-unbreakability tied to $D$ where $q'(j) = (2^{q(j)} + 1)\eta_w(i^2 4^i, h)$ for $j \in [i]$.
Because the bound holds for every arbitrarily chosen $i$, 
$G_R$ is $(2^{i(12^i + 1)} + 1)\eta_w(i^2 4^i, h)$-tied to $\delta(G_R)$, as claimed. 

\end{proof}

\subsubsection{Stable Separations}

Let $G$ be graph and let $S\subseteq V(G)$ be any subset of vertices. We will make use of the following definition from~\cite{LokshtanovPPS17}, together with one of its properties. 

\begin{definition}{\rm~\cite[Definition 12]{LokshtanovPPS17}}
Suppose that $(S_L,S_R)$ is a partition of $S$. We say that a separation $(A,B)$ of $G$ is {\em{stable for $(S_L,S_R)$}} if $(A,B)$ is a minimum-order $S_L$-$S_R$ separation. A separation $(A,B)$ is {\em{$S$-stable}} if it is stable for some partition of $S$.
\end{definition}

\begin{lemma}\label{lem:magical}{\rm~\cite[Lemma 13]{LokshtanovPPS17}}
Let $G$ be a graph, let $S\subseteq V(G)$ be any subset of vertices, and let ${\cal S}$ be any finite family of $S$-stable separations. Define
$$X:=S\cup \bigcup_{(A,B)\in {\cal S}} A\cap B.$$
Then, for every connected component of $G\setminus X$, if $Z$ is its vertex set, then $|N(Z)|\leq |S|$.
\end{lemma}

%% file: reductionToUnbreakable.tex
\subsubsection{Reduction Algorithm}
The goal of this subsection is to prove Lemma~\ref{lem:unbreakToClique}, namely to give a reduction from canonization of improved-clique-unbreakable graphs to canonization of unbreakable graphs. The proof of Lemma~\ref{lem:unbreakToClique} splits in two cases, which are handled by Lemma~\ref{lem:cliqueUnbreakableCentralClique} and Lemma~\ref{lem:unbreakableToCliqueUnbreakable}, respectively. We first state and prove these two lemmas, before proceeding with the proof of Lemma~\ref{lem:unbreakToClique}. If the reader wishes to understand the motivation for the statements of  Lemma~\ref{lem:cliqueUnbreakableCentralClique} and Lemma~\ref{lem:unbreakableToCliqueUnbreakable}, they may skip to the end of this sub-section and read the proof of Lemma~\ref{lem:unbreakToClique} first, before coming back to this point.

Throughout the section, we assume that ${\cal F}$ denote a finite family of graphs, and let  $n_{\cal F}$ denote the maximum size of a graph (in terms of the number of vertices) in ${\cal F}$.  Further, recall from Section~\ref{sec:representatives} that  ${\eta}_w(f, t)$ denotes the smallest integer $\rho$ such that every $t$-boundaried graph has an $f$-weak representative on at most $\rho$ vertices.

\begin{lemma}\label{lem:cliqueUnbreakableCentralClique}
Let $\hat{q}: \mathbb{F}^\star \times \mathbb{N} \rightarrow \mathbb{N}$ be a function defined as 
$\hat{q}({\cal F},i) = 2^{2^{10i}}\eta_w(2^{8i},h)$,
and
let $\mathbb{F} \subseteq \mathbb{F}^\star$ be a collection of finite sets of graphs. 
If there exists a canonization algorithm ${\cal A}$ for 
$\hat{q}$-unbreakable, 
$\mathbb{F}$-free classes (for some  $\kappa_{\cal A} : \mathbb{F} \rightarrow \mathbb{N}$) 
then there exists 
a weakly isomorphism invariant algorithm with the following specifications:

\begin{itemize}
\item The algorithm takes as input ${\cal F} \in \mathbb{F}$, positive integers $k$, $s$ and a $(s,k,k)$-improved clique unbreakable graph $G$, and an isomorphism invariant in $G$ connectivity sensitive star decomposition $(T, \chi)$ with central bag $b$ such that $\imp{G}{k}[\chi(b)]$ is a clique, and for every leaf $\ell$ of $T$, $|\chi(\ell)| \leq s$ and $|\sigma(\ell)| \leq k$.

\item The algorithm terminates on all inputs and either fails (outputs $\bot$) or succeeds and outputs a labeling of $G$, which is weakly isomorphism invariant in $G$, ${\cal F}$, $k$ and $s$. If $k/2 \geq \kappa_{\cal A}({\cal F})$ then the algorithm succeeds.

\item The running time of the algorithm on an instance $(G,{\cal F},k,s,(T,\chi))$ is bounded by $\cO(s!n^{\cO(1)})$ plus the total time taken by one invocation of ${\cal A}$ on $({\cal F}, G^\star, k, q)$, where ${\cal F} \in \mathbb{F}$, $q : [k/2] \rightarrow \mathbb{N}$ is a function such that $q(i) \leq \hat{q}({\cal F}, i)$ for $i \leq k/2$, and $G^\star$ is an ${\cal F}$-free, $q$-unbreakable graph on at most $n$ vertices.
\end{itemize}
\end{lemma}

\begin{proof}
Since $A = \imp{G}{k}[\chi(b)]$ is a clique, there is no separation $(L, R)$ of $G$ of order at most $k$ such that both $(L \cap \chi(b)) \setminus R$ and $(R \cap A) \setminus L$ are non-empty.
Thus for the function $q_1 : [k] \rightarrow \mathbb{N}$ defined as $q(i) = i$, $\chi(b)$ is $q_1$-unbreakable in $G$. For each leaf $\ell$ the algorithm computes for every permutation $\pi : \sigma(\ell) \rightarrow [|\sigma(\ell)|]$ a canonical labeling $\lambda_{\pi, \ell}$ of $G[\chi(\ell)]$ such that $\lambda_{\pi, \ell}(v) = \pi(v)$ for every $v \in \sigma(\ell)$, by choosing the proper labeling $\lambda_{\pi, \ell}(v) : \chi(\ell) \rightarrow [|\chi(\ell)|]$ that makes $G[\chi(\ell)]$ with this labeling lexicographically smallest.
This takes at most $s!n^{\cO(1)}$ time and the family $\{ \lambda_{\pi, \ell} \}$ of labelings is a locally weakly isomorphism invariant leaf labeling of $(T, \chi)$.

Recall that, $n_{\cal F}$ denotes the maximum size of a graph (in terms of the number of vertices) in the given family of topological minor free graphs ${\cal F}$.  

We set $D = \emptyset$ and apply the weak unpumping operation on $(G, D, (T,\chi),b,\{ \lambda_{\pi, \ell} \}))$ and parameter $h = n_{\cal F}$. 
We aim to call the algorithm ${\cal A}$ on $G^\star$. 
To that end we need to ensure that $G^\star$ satisfies the input requirements of  ${\cal A}$, in particular that it is sufficiently unbreakable.  We first define a set of functions, that will be used in the proof.  
\begin{itemize}
\setlength{\itemsep}{-2pt}
\item $q_1 : [k] \rightarrow \mathbb{N}$,   $q_1(i) = i $.
\item $q_2 : \mathbb{N} \rightarrow \mathbb{N}$ defined as $q_2(i) = (2^{i(12^i + 1)} + 1)\eta_w(i^2 4^i, h)$. 
Here, ${\eta}_w(f, t)$ denotes the smallest $\ell$ such that every $t$-boundaried graph has an $f$-weak representative on at most $\ell$ vertices. Therefore, $q_2(i) \leq 2^{2^{5i}}\eta_w(2^{4i},h)$

\item $q_3 : [k/2] \rightarrow \mathbb{N}$ is defined as $q_3(i) = q_1(i) + (2^{q_1(i)} + i) \cdot q_2(q_1(i) + i)$. Therefore, $q_3(i) \leq 2^{i+2} \cdot 2^{2^{10i}}\eta_w(2^{8i},h) = \hat{q}({\cal F},i)$.
\end{itemize}

Observe that $A$ is $q_1$-unbreakable in $G$.
By Property~\ref{itm:unpumpKeepsUnbreak} of Lemma~\ref{lem:unpumpProperties}, $A$ is $q_1$-unbreakable in $G^\star$.
For every leaf $\ell$ of $G^\star$ the definition of weak unpumping yields that $(G^\star[\chi^\star(\ell)], \pi\langle \ell \rangle)$ is an $h$-weak representative of $(G[\chi'(\ell)], \pi\langle \ell \rangle)$.
From Lemma~\ref{lem:repsAreUnbreakabilityTiedToBoundary} it follows that for every leaf $\ell$, $\chi^\star(\ell)$ is $q_2$-unbreakability tied to $\sigma^\star(\ell)$.
By Property~\ref{itm:connSensitiveGstar} of Lemma~\ref{lem:unpumpProperties} the star decomposition $(T, \chi^\star)$ is connectivity-sensitive. 
Now Lemma~\ref{lem:concludeGstarUnbreakable} yields that $G^\star$ is $q_3$-unbreakable.  

The algorithm calls ${\cal A}$ on $(G^\star, {\cal F}, k/2, q_3)$. Because  $G^\star$ is a $q_3$-unbreakable ${\cal F}$-free graph this is a well-formed input for the algorithm ${\cal A}$. If ${\cal A}$ fails (returns $\bot$) then our algorithm fails as well (and this is weakly isomorphism invariant). However, when $k/2 \geq \kappa_{\cal A}(\cal F)$, none of the calls to ${\cal A}$ fail.  
The call to the algorithm $\cal A$ returns a weakly isomorphism invariant proper labeling $\lambda^\star$ of $G^\star$.  
The algorithm now sets $\iota : D \rightarrow [|D|]$ to be the null labeling, and applies the lifting procedure (see Section~\ref{sec:unpumping}) on the lifting bundle $(G, D, (T, \chi), b, \{\lambda_{\pi, \ell}\}, G^\star, (T, \chi^\star), \iota, \lambda^\star)$ and obtains a labeling $\lambda$ of $(G, \iota) = G$. By Lemma~\ref{lem:lem:caNlifting},  $\lambda$ is weakly isomorphism invariant in $G$ (and ${\cal F}$, $s$ and $k$), and by Observation~\ref{obs:liftRunTime} the procedure takes time polynomial in the input. 

After computing the locally weakly isomorhism invariant leaf labeling of $(T, \chi)$ the algorithm spends polynomial time and makes a single invocation of the algorithm ${\cal A}$. Thus the algorithm runs within the claimed running time. This concludes the proof.
\end{proof}

The lemma below is the main engine of the algorithm of Lemma~\ref{lem:unbreakToClique}.

\begin{lemma}\label{lem:unbreakableToCliqueUnbreakable}
There exists a weakly isomorphism invariant algorithm with the following specifications.
The algorithm takes as input 
a set ${\cal F} \in \mathbb{F}$, integers $k, s$ and a graph $G$ with a distinguished set $D$, such that $|D| \leq 4^k k^2$, with the property that there exists an $\iota : D \rightarrow [|D|]$ and a $[|D|]$-boundaried graph $(H,\iota)$ so that 
\begin{enumerate}
\item $(G,\iota) \oplus (H,\iota)$ is ${\cal F}$-free,
\item The $k$-improved graph of $(G,\iota) \oplus (H,\iota)$ has a connectivity sensitive star decomposition $(T, \chi)$ with center bag $b$ such that $\chi(b)$ is either a $k$-atom or a clique of size at most $k$, $\chi(b) \setminus V(G) \neq \emptyset$, and for all  leaves $\ell \in V(T)$, $|\chi(\ell)|\leq s$. 
\end{enumerate}
The algorithm either fails (i.e. outputs $\bot$) or succeeds.  If $k/2 \geq \kappa_{\cal A}({\cal F})$ then the algorithm succeeds  and 
it outputs for every permutation $\pi : D \rightarrow [|D|]$ a proper labeling $\lambda_\pi$ of $(G, \pi)$ such that $\lambda_\pi(v) = \pi(v)$ for every $v \in D$. The labeling $\lambda_\pi$ is weakly isomorphism invariant in $(G, \pi)$ (and ${\cal F}$, $k$ and $s$).
The running time of the algorithm is upper bounded by $g({\cal F}, k, s)n^{\cO(1)}$ plus the time taken by $h({\cal F},k)n^{\cO(1)}$ invocations of ${\cal A}$ on instances $(G^\star, {\cal F}, k/2, q^\star)$ where $G^\star$ is a $q^\star$-unbreakable, ${\cal F}$-topological minor free graph on at most $n$ vertices and $q^\star : [k/2] \rightarrow \mathbb{N}$ is a function such that
$q^\star \leq 2^{2^{2^{14i}}}\eta_w(2^{2^{13i}},n_{\cal F})$. %
\end{lemma}
\begin{proof}
 In what follows, for an integer $k$, let us define the following functions of $k$, that will be used in the proof.  
\begin{eqnarray*}
\Upsilon &=& k, \\
\eta & = & k^2 \cdot 4^k, \\
\zeta & = & 2^{2\eta}\eta,\\
|D|& = &h\leq 4^k k^2, \mbox{ and let}\\
G^\dagger & =&  \imp{(G,\iota) \oplus (H,\iota)}{k},  \mbox{ and}\\
\pot(G,X) & = &2(|V(G)|-|X|)+|X|. 
\end{eqnarray*}
We ensure that the algorithm does not loop using a potential $\pot(G,X)=2(|V(G)|-|X|)+|X|$. Further,  in every recursive call, we ensure that the boundary is upper bounded by $\eta =k^2 \cdot 4^k$.  

Formally, every call of the algorithm will use only calls for inputs with a strictly smaller potential, and the potential is always nonnegative. Initial potential is $\pot(G,D)$. 
In other words, it is an induction on $2(n-|D|)+|D|=2n-|D|$.

For the running time analysis we will upper bound the time spent by the algorithm in each individual node of the recursion tree, without accounting for the time spent inside the recursive calls of the descendants of that node. These bounds are then added up over all the nodes in the recursion tree (this analysis is at the end of the proof). 

\paragraph{Base case: $2n-|D|\leq 0$.} In this situation we know that the whole graph is upper bounded by $h$. 
We first apply Lemma~\ref{lem:compact} and obtain an isomorphism invariant  compact coloring of $G$. Then, we define $(T, \chi)$, an  isomorphism invariant (with respect to $G$ and $D$) 
star decomposition of $G$ with central bag $b \in V(T)$. The central bag $\chi(b)=V(G)$ and $T$ has no leaves.   
Now we apply Lemma~\ref{lem:smallCentersAlgorithm}, with $(G, D,  (T,\chi,b,\{ \lambda_{\pi, \ell} \}))$  
and for every labeling $\pi : D \rightarrow [|D|]$,  obtain 
 a proper labeling $\lambda_\pi$ of $(G, \pi)$ such that $\lambda_\pi(v) = \pi(v)$. 
 The labeling $\lambda_\pi$ is weakly isomorphism invariant in $(G, \pi)$. The running time of this step is upper bounded by $h!n^{\cO(1)}$. Thus, the running time is upper bounded by 
 $2^{\cO(h \log h)} n^{\cO(1)} $. We now proceed with the case that $2n-|D|>0$.

\paragraph{A remark on the case distinction:} We have $4$ main cases, 
Case I is that $|D| \leq \Upsilon$  and $D$ is not a clique in  $\imp{G}{k}$,
Case II is that $|D| \leq \Upsilon$   and $D$ is a clique in   $\imp{G}{k}$.
Case III is that $ \Upsilon<|D| \leq \eta $   and $D$ is breakable in  $G$.
Case IV is that  $ \Upsilon < |D| \leq \eta $    and $D$ is unbreakable in  $G$.
Here Case IV is the ``main interesting case'' for which we built all the tools in Subsection~\ref{sec:lastDecomp}. For these tools to apply we need $D$ to be sufficiently large and unbreakable.
All the preceeding cases handle the ways in which $D$ could fail to meet these requirements. 
Case III handles the case when $D$ is breakable (for Case III we require $D$ to be sufficiently large and breakable, we do not actually need to do this, but it makes the four cases mutually exclusive), while %
Cases I and II handle the case when $D$ is too small. Case II is the only place we use $k$-improved unbreakability.
The proofs of Case I and III are very similar, but different enough to warrant separate discussion.

  We apply these cases in the order of appearance. 
\paragraph{Case I: $|D| \leq \Upsilon$  and $D$ is not a clique in  $\imp{G}{k}$. }
Consider any pair of vertices $x,y\in D$ such that $xy\notin E(\imp{G}{k})$. In  $\imp{G}{k}$, we have the property that 
$\conn_{\imp{G}{k}}(x,y)\leq k$, and hence the minimum-order separations that separate $x$ and $y$ have order at most $k$. Let $(L_{xy},R_{xy})$ be the minimum-order separation separating $x$ and $y$ that are pushed towards $y$. Recall that, by the definition of an $x-y$ separation, the separator $L_{xy} \cap R_{xy}$ does not contain
$x$ or $y$. We define the set $X$ as follows:
\begin{eqnarray}\label{eq:casesmall}
B := D\cup \bigcup_{\substack{x,y\in D,\\ xy\notin E(\imp{G}{k})}}  (L_{xy}\cap R_{xy})
\end{eqnarray}
In other words, we enhance $D$ by adding, for every ordered pair $(x,y)$ of nonadjacent vertices from $D$,  the extreme minimum separator, pushed towards $y$, separating them. By construction $B$ is isomorphism invariant  (with respect to $G$ and $D$). Observe also that $|B|\leq |D|+k |D|^2=h+k h^2 \leq k+k^3\leq \bigadh$, since $|D|\leq \bagthresh$. 
Note that in this case $B$ can be computed in time $k^{\Oh(1)}\cdot (|V(\imp{G}{k})|+|E(\imp{G}{k})|)$, by considering every pair of non-adjacent vertices of $D$,
and for each of them running Lemma~\ref{lem:f-f} with $k$ as the bound on the order of the separation.

In this case we recurse on $\widetilde{G}=G[N[C]]$ for each connected component $C$ of $G-B$. However, to do this we need to show that $\widetilde{G}$ satisfies the premise of Lemma~\ref{lem:unbreakableToCliqueUnbreakable}. We take restriction of ${\sf col}_G$ to $\widetilde{G}$ as the coloring ${\sf col}_{\widetilde{G}}$. 
The set $\widetilde{D}=N(C)$ is defined to be the distinguished set.  Since $B$ is bounded by $\bigadh$ and $N(C)\subseteq B$, we have that 
$|\widetilde{D}|=|N(C)|\leq \bigadh$. In fact, note that $N(C)\subsetneq B$. Indeed, since there exists a pair of vertices $x,y\in D$  such that  $xy\notin E(\imp{G}{k})$, we added a separator between $x$ and $y$ to $B$. This implies that both $x,y$ cannot belong to $N(C)$. Hence, $N(C)\subsetneq B$. Let $\tilde{\iota} \colon \widetilde{D} \to [|\widetilde{D}|]$, be an arbitrary bijection from $\widetilde{D} \to [|\widetilde{D}|]$.

 Let $G_C=G-C$. That is, the graph obtained by deleting the vertices of connected component $C$ from $G$. Observe that since $D \cap C=\emptyset$, we can define $\widetilde{H}$ as  
$(G_C,\iota)\oplus (H,\iota)$.   We view $(\widetilde{H},\tilde{\iota})$ as a $[|\widetilde{D}|]$-boundaried graph with boundary $\widetilde{D}$ and the labelling $\tilde{\iota} \colon \widetilde{D} \to [|\widetilde{D}|]$, defined above. 
Observe that $(\widetilde{G},\tilde{\iota})\oplus (\widetilde{H},\tilde{\iota})$ {\em is same as} $(G,\iota)\oplus (H,\iota)$. This implies that $(\widetilde{G},\tilde{\iota})\oplus (\widetilde{H},\tilde{\iota})$ is ${\cal F}$-free, as $(G,\iota)\oplus (H,\iota)$  is ${\cal F}$-free.  Finally, we need to show that  the $k$-improved graph of  $(\widetilde{G},\tilde{\iota})\oplus (\widetilde{H},\tilde{\iota})$ has a connectivity sensitive star decomposition $(T', \chi')$ satisfying the properties stated in the statement of the lemma. Since, $(\widetilde{G},\tilde{\iota})\oplus (\widetilde{H},\tilde{\iota})$ is same as $(G,\iota)\oplus (H,\iota)$, we 
take the connectivity sensitive star decomposition $(T, \chi)$ of $(G,\iota)\oplus (H,\iota)$ as the connectivity sensitive star decomposition $(T', \chi')$ of  $(\widetilde{G},\tilde{\iota})\oplus (\widetilde{H},\tilde{\iota})$. Using the fact that $V(\widetilde{G})\subseteq V(G)$, we can can easily observe that $(T', \chi')$ satisfies all the desired properties. 

  To apply the inductive argument we need to show that 
 $\pot(\widetilde{G},\widetilde{D})<\pot(G,D)$.  
 Observe that, 
\begin{eqnarray*}
 \pot(\widetilde{G},\widetilde{D})   &\leq & 2( n-|B|)+ |\widetilde{D}|\\
&=&  2n-2|B|+|D|-|D|+|\widetilde{D}| \\
& = & 2n-|D|-(2|B|-|D|- |\widetilde{D}|) \\
& < & \pot(G,D)~~~~~~~~~~~~~~~~~~~~~~~~~~~~~~~~~~~~~~~~~~~~(\widetilde{D} \subsetneq B \mbox{ and } D \subseteq B )
 \end{eqnarray*}

Finally, the algorithm either outputs $\bot$ or it outputs for every labelings $\pi : \widetilde{D} \rightarrow [|\widetilde{D}|]$ a proper labeling $\lambda_\pi$ of $(\widetilde{G},\pi)$ such that $\lambda_\pi(v) = \pi(v)$. The labeling $\lambda_\pi$ is weakly isomorphism invariant in $(\widetilde{G}, \pi)$. If the algorithm outputs $\bot$, we output $\bot$ on the original instance $({\cal F},k,G,D)$ and terminate (and this is weakly isomorphism invariant).  Thus, we assume that the algorithm succeeds on each recursive call made on the children. 

Before we go further, we first apply Lemma~\ref{lem:compact}  on $G$ and ${\sf col}_{G}$  and obtain an isomorphism invariant  compact coloring of $G$. We also need to compact the labelings obtained from applying recursive steps. To do this we again apply Lemma~\ref{lem:compact}   and obtain compact labelings. 

Next, we define $(T^\star, \chi^\star)$, an isomorphism invariant (with respect to $G$ and $B$) star decomposition of 
$G$ with central bag $b \in V(T^\star)$. The central bag $\chi(b)=B$. Further, for every connected component 
$C$ of $G-B$, there is a leaf $\ell\in V(T^\star)$, such that $\chi^\star(\ell)=N[C]$.  Using the output of the recursive 
calls, we define a leaf-labeling of $(T^\star, \chi^\star)$.  That is, it  is a set of labelings $\{ \lambda_{\pi, \ell} \}$ with a proper labeling  $\lambda_{\pi, \ell}$ of $G[\chi^\star(\ell)]$ for every leaf $\ell$ of $T$ and every permutation $\pi : \sigma(\ell) \rightarrow [|\sigma(\ell)|]$. For every $\ell$ and $\pi$ the labeling $\lambda_{\pi, \ell}$ satisfy $\lambda_{\pi, \ell}(v) = \pi(v)$ for every $v \in \sigma(\ell)$. By construction, leaf labeling are locally weakly isomorphism invariant. Now we apply Lemma~\ref{lem:smallCentersAlgorithm}, with $(G, B,  (T^\star,\chi^\star,b,\{ \lambda_{\pi, \ell} \}))$  
and for every labeling $\pi : B \rightarrow [|B|]$,  obtain 
 a proper labeling $\lambda_\pi$ of $(G, \pi)$ such that $\lambda_\pi(v) = \pi(v)$. 
 The labeling $\lambda_\pi$ is weakly isomorphism invariant in $(G, \pi)$. 
For each labeling $\pi' : D \rightarrow [|D|]$ the algorithm 
iterates over all bijections $\pi : B \rightarrow [|B|]$ that coincide with $\pi'$ on $D$ and sets $\lambda_{\pi'}$ to be the labeling $\lambda_\pi$ of $(G, \pi')$ that makes the labeled graph $(G, \lambda_{\pi'})$ lexicographically smallest. 
The labeling $\lambda_{\pi'}$ is weakly isomorphism invariant in $(G, \pi')$ (and ${\cal F}$, $k$ and $s$), hence it can be output by the algorithm.
 The running time of this step is upper bounded by $|B|!n^{\cO(1)}=(k+k^3)!n^{\cO(1)} =2^{\cO(k \log k)} n^{\cO(1)} $. %

\paragraph{Case II:  $|D| \leq \Upsilon$   and $D$ is a clique in   $\imp{G}{k}$.}

In this case we will show that $|V(G)|$ is bounded by $k+3s$. We know that the $k$-improved graph $G^\dagger$ of $(G,\iota) \oplus (H,\iota)$ has a star decomposition $(T, \chi)$ with center bag $b$ such that $\chi(b)$ is either a $k$-atom or a clique of size at most $k$, $\chi(b) \setminus V(G) \neq \emptyset$, and for all  leaves $\ell \in V(T)$, $|\chi(\ell)|\leq s$.

Since $D \leq \Upsilon = k$ we have that $(V(G), V(H))$ is a clique separation of size at most $k$ in $G^\dagger$. Indeed, it is a clique by assumption, and it is a separation of order at most $k$ in $(G,\iota) \oplus (H,\iota)$ and therefore also a separation in $G^\dagger$.
But $\chi(b)$ is either a clique or a $k$-atom in $G^\dagger$ and therefore $\chi(b) \subseteq V(G)$ or $\chi(b) \subseteq V(H)$. However, by assumption, $\chi(b) \setminus V(G) \neq \emptyset$ and therefore $\chi(b) \subseteq V(H)$. 

It follows that every connected component of $G - D$ is also a connected component of $G^\dagger - \chi(b)$. Thus, for every subset $X \subseteq D$ the union $\hat{C}$ of all connected components of $G - D$ with neighborhood precisely $X$ is a subset of $\chi(\ell) \setminus \chi(b)$ for some leaf $\ell$ of $T$. Thus $|\hat{C}| \leq s$. Since there are $2^k$ choices for $X$ we conclude that $|V(G) \setminus D| \leq 2^ks$, and therefore $|V(G)| \leq 2^ks+k$.
Once we have upper bounded the $|V(G)|$ of the graph, the rest of the arguments are identical to the Base Case and thus omitted.

\paragraph{Case III:  $ \Upsilon<|D| \leq \eta $   and $D$ is breakable in  $G$.}
In this case $D$ is breakable in $G$. We construct the set $B$ as follows. Consider all the pairs $(L,R)$ of subsets of $D$ such that $L\cap R=\emptyset$ and $|L|\leq |R|$. For every such pair, let us verify whether the minimum order of a separation separating $L$ and $R$ is at most $|L|-1$, and in this case let us compute the minimum-order $L$-$R$ separation $(A_{LR},B_{LR})$ that are pushed towards $R$ (note that here the vertices of $L$ and $R$ can be included in the separator). We now define $B$ similar to the previous case:
\begin{eqnarray}\label{eq:casebig}
B := D\cup \bigcup_{\substack{L,R\subseteq S,\ L\cap R=\emptyset,\\ |L|\leq |R|,\ \conn(L,R)\leq |L|-1}} 
(A_{LR}\cap B_{LR}). 
\end{eqnarray}
By construction $B$ is isomorphism invariant  (with respect to $G$ and $D$). Further, $|B|\leq \bigbag$. 
This follows directly from the definition and the fact that $|D|\leq \bigadh$. Indeed, there are at most $2^{2\eta}$ choices of $|L|$ and $|R|$, and the size of the separator is bounded by $|L|\leq |D|\leq \eta$. 
We proceed to show that for every component $C$ of $G-B$, we have that $|N(C)|\leq \bigadh$. 

Towards this consider any pair $(L,R)$ considered in the union in (\ref{eq:casebig}), and let us look at a separation 
$(A_{LR},B_{LR})$. 
We can easily construct a partition $(D_L,D_R)$ of $D$ such that $L\subseteq D_L\subseteq A_{LR}$ and $R\subseteq D_R\subseteq B_{LR}$, by assigning vertices of $L$ to $D_L$, vertices of $R$ to $D_R$, and assigning vertices of $D\setminus (L\cup R)$ according to their containment to $ A_{LR}$ or $ B_{LR}$ (thus, we have a unique choice for all the vertices apart from $(D\setminus (L\cup R))\cap (A_{LR}\cap B_{LR})$). Then $(A_{LR},B_{LR})$ is an $D_L$-$D_R$ separation, and since it was a minimum-order $L$-$R$ separation, it must be also a minimum-order $D_L$-$D_R$ separation. Hence, $(A_{LR},B_{LR})$ is a $D$-stable separation. We infer that all the separations considered in the union in (\ref{eq:casebig}) are $D$-stable. From Lemma~\ref{lem:magical} it follows that $|N(C)|\leq |D|\leq \bigadh$, for every connected component $C$ of $G-B$. This concludes that $|N(C)|\leq \bigadh$.

Finally, we show that $B\supsetneq N(C)$, that is, $B$ is a proper superset of $N(C)$. Recall that  $D$ is breakable in $G$ and let $(A,B)$ be a witnessing separation. Then, we have that   $|A \cap B| < p=\min\{|A \cap D|, |B \cap D|\}$. 
Let  $L$ and $R$ be any two subsets of $(A\setminus B)\cap D$ and $(B\setminus A)\cap D$, respectively, that have sizes $p$.  Observe that the existence of separation $(A,B)$ certifies that $\conn(L,R)< p$, and hence the pair $(L,R)$ is considered in the union in (\ref{eq:casebig}). Recall that $(A_{LR},B_{LR})$ is then the corresponding $L$-$R$ separation of minimum order that is pushed towards $R$. Since $|A_{LR}\cap B_{LR}|\leq p-1$ and $|L|,|R|=p$, there exist some vertices $x,y$ such that $x\in L\setminus (A_{LR}\cap B_{LR})$ and $y\in R\setminus (A_{LR}\cap B_{LR})$. Similar to Case $I$, observe that we added a separator between $x$ and $y$ to $B$. This implies that both $x,y$ cannot belong to $N(C)$. Hence, $N(C)\subsetneq B$. %

Note that in this case $B$ can be computed in time $\Oh(2^{2\eta}\eta \cdot (|V(G)|+|E(G)|))$ as follows. There are at most $(2^\bigadh)^2=2^{2\bigadh}$ possible choices for $L$ and $R$ such that $|L|\leq |R|$, and for each of them we may check whether $\conn(L,R)\leq |L|-1$ (and compute $(A_{LR},B_{LR})$, if needed) using Lemma~\ref{lem:f-f}.

In this case we recurse on $\widetilde{G}=G[N[C]]$ for each of the connected component $C$ of $G-B$. However, to do this we need to show that $\widetilde{G}$ satisfies the premise of Lemma~\ref{lem:unbreakableToCliqueUnbreakable}. We take restriction of ${\sf col}_G$ to $\widetilde{G}$ as the coloring ${\sf col}_{\widetilde{G}}$. 
The set $\widetilde{D}=N(C)$ is defined to be the distinguished set.  By an earlier argumentation we have that 
$N(C)\subsetneq B$, and hence, $|\widetilde{D}|=|N(C)|\leq \bigadh$. 
Let $\tilde{\iota} \colon \widetilde{D} \to [|\widetilde{D}|]$, be an arbitrary bijection from $\widetilde{D} \to [|\widetilde{D}|]$.

 Let $G_C=G-C$. That is, the graph obtained by deleting the vertices of connected component $C$ from $G$. Observe that since $D \cap C=\emptyset$, we can define $\widetilde{H}$ as  
$(G_C,\iota)\oplus (H,\iota)$.   We view $(\widetilde{H},\tilde{\iota})$ as a $[|\widetilde{D}|]$-boundaried graph with boundary $\widetilde{D}$ and the labelling $\tilde{\iota} \colon \widetilde{D} \to [|\widetilde{D}|]$, defined above. 
Observe that $(\widetilde{G},\tilde{\iota})\oplus (\widetilde{H},\tilde{\iota})$ {\em is same as} $(G,\iota)\oplus (H,\iota)$. This implies that $(\widetilde{G},\tilde{\iota})\oplus (\widetilde{H},\tilde{\iota})$ is ${\cal F}$-free, as $(G,\iota)\oplus (H,\iota)$  is ${\cal F}$-free.  Finally, we need to show that  the $k$-improved graph of  $(\widetilde{G},\tilde{\iota})\oplus (\widetilde{H},\tilde{\iota})$ has a connectivity sensitive star decomposition $(T', \chi')$ satisfying the properties stated in the statement of the lemma. Since, $(\widetilde{G},\tilde{\iota})\oplus (\widetilde{H},\tilde{\iota})$ is same as $(G,\iota)\oplus (H,\iota)$, we 
take the connectivity sensitive star decomposition $(T, \chi)$ of $(G,\iota)\oplus (H,\iota)$ as the connectivity sensitive star decomposition $(T', \chi')$ of  $(\widetilde{G},\tilde{\iota})\oplus (\widetilde{H},\tilde{\iota})$. Using the fact that $V(\widetilde{G})\subseteq V(G)$, we can can easily observe that $(T', \chi')$ satisfies all the desired properties. 

  To apply the inductive argument we need to show that 
 $\pot(\widetilde{G},\widetilde{D})<\pot(G,D)$.  
 Observe that, 
\begin{eqnarray*}
 \pot(\widetilde{G},\widetilde{D})   &\leq & 2( n-|B|)+ |\widetilde{D}|\\
&=&  2n-2|B|+|D|-|D|+|\widetilde{D}| \\
& = & 2n-|D|-(2|B|-|D|- |\widetilde{D}|) \\
& < & \pot(G,D)~~~~~~~~~~~~~~~~~~~~~~~~~~~~~~~~~~~~~~~~~~~~(\widetilde{D} \subsetneq B \mbox{ and } D \subseteq B )
 \end{eqnarray*}

Finally, the algorithm either outputs $\bot$ or it outputs for every labeling $\pi : \widetilde{D} \rightarrow [|\widetilde{D}|]$ a proper labeling $\lambda_\pi$ of $(\widetilde{G},\pi)$ such that $\lambda_\pi(v) = \pi(v)$. The labeling $\lambda_\pi$ is weakly isomorphism invariant in $(\widetilde{G}, \pi)$. If the algorithm outputs $\bot$, we output $\bot$ on the original instance $({\cal F},k,G,D)$ and terminate (and this is weakly isomorphism invariant).  Thus, we assume that the algorithm succeeds on each recursive call made on the children. 

Before we go further, we first apply Lemma~\ref{lem:compact}  on $G$ and ${\sf col}_{G}$  and obtain an isomorphism invariant  compact coloring of $G$. We also need to compact the labelings obtained from applying recursive steps. To do this we again apply Lemma~\ref{lem:compact}   and obtain compact labelings. 

Next, we define $(T^\star, \chi^\star)$, an isomorphism invariant (with respect to $G$ and $B$) star decomposition of 
$G$ with central bag $b \in V(T^\star)$. The central bag $\chi(b)=B$. Further, for every connected component 
$C$ of $G-B$, there is a leaf $\ell\in V(T^\star)$, such that $\chi^\star(\ell)=N[C]$.  Using the output of the recursive 
calls, we define a leaf-labeling of $(T^\star, \chi^\star)$.  That is, it  is a set of labelings $\{ \lambda_{\pi, \ell} \}$ with a proper labeling  $\lambda_{\pi, \ell}$ of $G[\chi^\star(\ell)]$ for every leaf $\ell$ of $T$ and every permutation $\pi : \sigma(\ell) \rightarrow [|\sigma(\ell)|]$. For every $\ell$ and $\pi$ the labeling $\lambda_{\pi, \ell}$ satisfies $\lambda_{\pi, \ell}(v) = \pi(v)$ for every $v \in \sigma(\ell)$. By construction, the leaf labeling is locally weakly isomorphism invariant. Now we apply Lemma~\ref{lem:smallCentersAlgorithm}, with $(G, B,  (T^\star,\chi^\star,b,\{ \lambda_{\pi, \ell} \}))$  
and for every labeling $\pi : B \rightarrow [|B|]$,  obtain 
 a proper labeling $\lambda_\pi$ of $(G, \pi)$ such that $\lambda_\pi(v) = \pi(v)$. 
 The labeling $\lambda_\pi$ is weakly isomorphism invariant in $(G, \pi)$. 
For each labeling $\pi' : D \rightarrow [|D|]$ the algorithm 
iterates over all bijections $\pi : B \rightarrow [|B|]$ that coincide with $\pi'$ on $D$ and sets $\lambda_{\pi'}$ to be the labeling $\lambda_\pi$ of $(G, \pi')$ that makes the labeled graph $(G, \lambda_{\pi'})$ lexicographically smallest. 
The labeling $\lambda_{\pi'}$ is weakly isomorphism invariant in $(G, \pi')$ (and ${\cal F}$, $k$ and $s$), hence it can be output by the algorithm.
 The running time of this step is upper bounded by $|B|!n^{\cO(1)}=2^{\cO(\bigbag \log \bigbag)} n^{\cO(1)} $.

\paragraph{Case IV:  $ \Upsilon < |D| \leq \eta $    and $D$ is unbreakable in  $G$.}
In this case $D$ is unbreakable in $G$. We first apply Lemma~\ref{lem:extensionIsTied} and construct a  $k$-important separator extension $A$ of $D$, which is isomorphism invariant in $(G, D)$. For each $v\in D$, the unique $\{v\}-D$, separator is $\{v\}$ itself and thus, $D\subseteq A$.  Observe that for any vertex $v \notin D$,  either $v$ is added to $A$ or some $v$-$D$ separator $S$ of size at most $k$ is added to $A$ and $v\notin S$. Using this we can show that for each component $C$ of $G-A$, $N(C)\subsetneq A$ as follows. 
Let $v\in C$, then we claim that $N(C)\cap D \subseteq S$. This follows from the fact that $S$ is a subset of $A$ and 
$S$ is a $v$-$D$ separator. Furthermore, since $D\subseteq A$, $|D|>k$ and $|S|\leq k$, we have that $A$ contains at least one vertex that is not present in $N(C)$. Hence, $N(C)\subsetneq A$.

 For each subset $X$ of $A$ so that $X = N(C)$ for some component $C$ of $G-A$, let $\hat{C}$ be the union of all such components. In this case we recurse on $\widetilde{G}=G[N[\hat{C}]]$ for each $\hat{C}$. However, to do this we need to show that $\widetilde{G}$ satisfies the premise of Lemma~\ref{lem:unbreakableToCliqueUnbreakable}. We take restriction of ${\sf col}_G$ to $\widetilde{G}$ as the coloring ${\sf col}_{\widetilde{G}}$.  The set $\widetilde{D}=X=N(C)=N(\hat{C})$ is defined to be the distinguished set for $\widetilde{G}$.  By Lemma~\ref{lem:extSmallComponentBoundary} we have that 
 $|N(C)|\leq \bigadh$, and thus   $|\widetilde{D}|\leq \bigadh$.  
Let $\tilde{\iota} \colon \widetilde{D} \to [|\widetilde{D}|]$, be an arbitrary bijection from $\widetilde{D} \to [|\widetilde{D}|]$.

Let $G_C=G-\hat{C}$. That is, the graph obtained by deleting the vertices of connected components in $\hat{C}$ from $G$. Observe that since $D \cap \hat{C}=\emptyset$, we can define $\widetilde{H}$ as  
$(G_C,\iota)\oplus (H,\iota)$.   We view $(\widetilde{H},\tilde{\iota})$ as a $[|\widetilde{D}|]$-boundaried graph with boundary $\widetilde{D}$ and the labelling $\tilde{\iota} \colon \widetilde{D} \to [|\widetilde{D}|]$, defined above. 
Observe that $(\widetilde{G},\tilde{\iota})\oplus (\widetilde{H},\tilde{\iota})$ {\em is same as} $(G,\iota)\oplus (H,\iota)$. This implies that $(\widetilde{G},\tilde{\iota})\oplus (\widetilde{H},\tilde{\iota})$ is ${\cal F}$-free, as $(G,\iota)\oplus (H,\iota)$  is ${\cal F}$-free.  Finally, we need to show that  the $k$-improved graph of  $(\widetilde{G},\tilde{\iota})\oplus (\widetilde{H},\tilde{\iota})$ has a connectivity sensitive star decomposition $(T', \chi')$ satisfying the properties stated in the statement of the lemma. Since, $(\widetilde{G},\tilde{\iota})\oplus (\widetilde{H},\tilde{\iota})$ is same as $(G,\iota)\oplus (H,\iota)$, we 
take the connectivity sensitive star decomposition $(T, \chi)$ of $(G,\iota)\oplus (H,\iota)$ as the connectivity sensitive star decomposition $(T', \chi')$ of  $(\widetilde{G},\tilde{\iota})\oplus (\widetilde{H},\tilde{\iota})$. Using the fact that $V(\widetilde{G})\subseteq V(G)$, we can can easily observe that $(T', \chi')$ satisfies all the desired properties. 

  To apply the inductive argument we need to show that 
 $\pot(\widetilde{G},\widetilde{D})<\pot(G,D)$.  
 Observe that, 
\begin{eqnarray*}
 \pot(\widetilde{G},\widetilde{D})   &\leq & 2( n-|A|)+ |\widetilde{D}|\\
&=&  2n-2|A|+|D|-|D|+|\widetilde{D}| \\
& = & 2n-|D|-(2|A|-|D|- |\widetilde{D}|) \\
& < & \pot(G,D)~~~~~~~~~~~~~~~~~~~~~~~~~~~~~~~~~~~~~~~~~~~~(\widetilde{D} \subsetneq A \mbox{ and } D \subseteq A )
 \end{eqnarray*}

Finally, the algorithm either outputs $\bot$ or it outputs for every bijection $\pi : \widetilde{D} \rightarrow [|\widetilde{D}|]$ a proper labeling $\lambda_\pi$ of $(\widetilde{G},\pi)$ such that $\lambda_\pi(v) = \pi(v)$. The labeling $\lambda_\pi$ is weakly isomorphism invariant in $(\widetilde{G}, \pi)$. If the algorithm outputs $\bot$, we output $\bot$ on the original instance $({\cal F},k,G,D)$ and terminate (and this is weakly isomorphism invariant).  Thus, we assume that the algorithm succeeds on each recursive call made on the children. 

Before we go further, we first apply Lemma~\ref{lem:compact}  on $G$ and ${\sf col}_{G}$  and obtain an isomorphism invariant  compact coloring of $G$. We also need to compact the labelings obtained from applying recursive steps. To do this we again apply Lemma~\ref{lem:compact}   and obtain compact labelings.

Next, we define $(T', \chi')$, an isomorphism invariant (with respect to $G$ and $A$) star decomposition of 
$G$ with central bag $b \in V(T')$. The central bag $\chi'(b)=A$. Furthermore, for each subset $X$ of $A$ so that $X = N(C)$ for some component $C$ of $G-A$, let $\hat{C}$ be the union of all such components. For each  $\hat{C}$ there is a leaf $\ell\in V(T')$, such that $\chi'(\ell)=N[\hat{C}]$. Using the output of the recursive 
calls, we define a leaf-labeling of $(T', \chi')$.  That is, it  is a set of labelings $\{ \lambda_{\pi, \ell} \}$ with a labeling  $\lambda_{\pi, \ell}$ of $G[\chi'(\ell)]$ for every leaf $\ell$ of $T'$ and every permutation $\pi : \sigma(\ell) \rightarrow [|\sigma(\ell)|]$. For every $\ell$ and $\pi$ the labeling $\lambda_{\pi, \ell}$ satisfy $\lambda_{\pi, \ell}(v) = \pi(v)$ for every $v \in \sigma(\ell)$. By construction, the leaf labelings are locally weakly isomorphism invariant. Further, by construction we have that  $\chi'(\ell) \cap \chi'(b) \neq \chi'(\ell') \cap \chi'(b)$ for every pair of distinct leaves $\ell, \ell'$ in $T'$, and for every leaf $\ell'$ of $T'$ and every connected component $C$ of $G[\chi'(\ell) \setminus \chi'(b)]$ it holds that $N(C) = \chi'(\ell) \setminus \chi'(b)$. Thus, the constructed star decomposition $(T', \chi')$ with central bag $b$ is connectivity-sensitive.  

Recall that, $n_{\cal F}$ denotes the maximum size of a graph (in terms of the number of vertices) in the given family of topological minor free graphs 
${\cal F}$.  
We apply the weak unpumping operation on $(G, D, (T',\chi'),b,\{ \lambda_{\pi, \ell} \}))$ and parameter $h = n_{\cal F}$.

We aim to call the algorithm ${\cal A}$ on $(G^\star, \iota)$ for every bijection $\iota : D \rightarrow [|D|]$.
To that end we need to ensure that $G^\star$ satisfies the input requirements of  ${\cal A}$, in particular that it is sufficiently unbreakable.  We first define a set of functions, that will be used in the proof.  
\begin{itemize}
\setlength{\itemsep}{-2pt}
\item $q_1 : \mathbb{N} \rightarrow \mathbb{N}$,   $q_1(i) = i $.
\item $q_2 : [k] \rightarrow \mathbb{N}$,   $q_2(i) = i \cdot (12^i + 1)$. Therefore, $q_2(i) \leq 2^{5i}$.
\item  $q_3 : [\frac{k}{2}] \rightarrow \mathbb{N}$ is defined as $q_3(i) = q_2(2i)$. Therefore $q_3(i) \leq 2^{10i}$.
\item $q_4 : \mathbb{N} \rightarrow \mathbb{N}$ defined as $q_4(i) = (2^{i(12^i + 1)} + 1)\eta_w(i^2 4^i, h)$. 
Here, ${\eta}_w(f, t)$ denotes the smallest $\ell$ such that every $t$-boundaried graph has an $f$-weak representative on at most $\ell$ vertices. Therefore, $q_4(i) \leq 2^{2^{5i}}\eta_w(2^{4i},h)$

\item $q^\star : [\frac{k}{2}] \rightarrow \mathbb{N}$ is defined as $q^\star(i) = q_3(i)+(2^{q_3(i)}+i)\cdot q_4(q_3(i)+i)$. Therefore, $q^\star(i) \leq 2^{2^{2^{14i}}}\eta_w(2^{2^{13i}},h)$. 
\end{itemize}

Observe that $D$ is $q_1$-unbreakable in $G$. 
By Lemma~\ref{lem:extensionIsTied}, $A$ is $q_2$-unbreakability tied to $D$ in $G$.
By Lemma~\ref{lem:tiedPlusUnbreakGivesUnbreak}, $A$ is $q_3$-unbreakable in $G$.
By Property~\ref{itm:unpumpKeepsUnbreak} of Lemma~\ref{lem:unpumpProperties}, $A$ is $q_3$-unbreakable in $G^\star$.
For every leaf $\ell$ of $G^\star$ the definition of weak unpumping yields that $(G^\star[\chi^\star(\ell)], \pi\langle \ell \rangle)$ is an $h$-weak representative of $(G[\chi'(\ell)], \pi\langle \ell \rangle)$.
From Lemma~\ref{lem:repsAreUnbreakabilityTiedToBoundary} it follows that for every leaf $\ell$, $\chi^\star(\ell)$ is $q_4$-unbreakability tied to $\sigma^\star(\ell)$.
By Property~\ref{itm:connSensitiveGstar} of Lemma~\ref{lem:unpumpProperties} the star decomposition $(T', \chi^\star)$ is connectivity-sensitive. 
Now Lemma~\ref{lem:concludeGstarUnbreakable} yields that $G^\star$ is $q^\star$-unbreakable.  

For every bijection $\iota : D \rightarrow [|D|]$ the algorithm calls ${\cal A}$ on $((G^\star, \iota), {\cal F}, k/2, q^\star)$. 
Because  $G^\star$ is a $q^\star$-unbreakable ${\cal F}$-free graph this is a well-formed input for the algorithm ${\cal A}$ (Technically the algorithm ${\cal A}$ takes as input a normal graph, rather than a boundaried graph. Therefore, instead of passing $(G^\star, \iota)$ directly to ${\cal A}$ the algorithm passes the compact color encoding of $(G^\star, \iota)$ instead.) 
If ${\cal A}$ fails (returns $\bot$) then our algorithm fails as well (and this is weakly isomorphism invariant). However, assuming that $k/2 \geq \kappa_{\cal A}$ none of the calls to ${\cal A}$ fail.  

For each $\iota$, the call to the algorithm $\cal A$ returns a weakly isomorphism invariant proper labeling $\lambda_\iota^\star$ of $(G^\star, \iota)$.  
The algorithm applies the lifting procedure (see Section~\ref{sec:unpumping}) on the lifting bundle $(G, D, (T', \chi'), b, \{\lambda_{\pi, \ell}\}, G^\star, (T', \chi^\star), \iota, \lambda_\iota^\star)$ and obtains a labeling $\lambda_\iota$ of $(G, \iota)$. By Lemma~\ref{lem:lem:caNlifting},  $\lambda_\iota$ is weakly isomorphism invariant in $(G, \iota)$, and by Observation~\ref{obs:liftRunTime} the procedure takes time polynomial in the input. We call the lifting procedure $|D|!$ times, one for each $\pi : D \rightarrow [|D|]$. This concludes the running time analysis for this case.

Finally, we analyze the running time of the whole algorithm. The algorithm is a recursive algorithm and hence has a recursion tree $\cal T$. We have already proved that in each node of the recursion tree the algorithm spends time at most  $g({\cal F}, k, s)n^{\cO(1)}$ and makes at most $g(k,{\cal F})n^{\cO(1)}$ ($\eta! = ( k^2 \cdot 4^k)!$) calls to ${\cal A}$. 
Thus, to prove the claimed upper bound on the running time it suffices to upper bound the number of nodes in ${\cal T}$ by $\cO(n^2)$.

First, observe that in every node of the recursion tree the algorithm only makes calls to instances with strictly lower potential value, where the potential is  $\pot(G,D)= 2n-|D| \leq 2n$. Since the potential value is non-zero for every node of the recursion tree, the length of the longest root-leaf path in ${\cal T}$ is at most $2n$.

A {\em level} of the recursion tree is a set of nodes of ${\cal T}$ at the same distance from the root. For every level we upper bound the number of nodes in that level by $\cO(n)$. Towards this goal label each node of ${\cal T}$ with the graph $G$ and a set $B$ (or $A$; both $A$ and $B$ contain the distinguished set $D$) that is processed in that node. A vertex $v \in V(G) \setminus B$ or $v \in V(G) \setminus A$ is called a {\em non-boundary} vertex. 

An inspection of the algorithm reveals the following observations: (i) every node of ${\cal T}$ has at least one non-boundary vertex, (ii) a non-boundary vertex of a node is also a non-boundary vertex of its parent node, (iii) for every node $(G, B)$ ($(G, A)$) of ${\cal T}$ and every non-boundary vertex $v \in V(G) \setminus B$ ($v \in V(G) \setminus A$), $v$ is a non-boundary node in at most one child of the call $(G, B)$ ($(G, A)$). The observations (i), (ii) and (iii) together yield that the size of every level of the recursion tree is upper bounded by $n$. This concludes the proof.
\end{proof}

We are now in position to prove Lemma~\ref{lem:unbreakToClique}.

\medskip
\noindent
{\bf Lemma~\ref{lem:unbreakToClique} (restated). }{\em 
There exists a function $q^\star: \mathbb{F}^\star \times \mathbb{N} \rightarrow \mathbb{N}$ such that the following holds. 
For every collection of finite sets of graphs $\mathbb{F} \subseteq \mathbb{F}^\star$, %
if there exists a canonization algorithm ${\cal A}$ for $(q^\star, \kappa)$-unbreakable $\mathbb{F}$-free classes (for some  $\kappa : \mathbb{F} \rightarrow \mathbb{N}$) 
then there exists a function $\kappa_{\cal B} : \mathbb{F} \rightarrow \mathbb{N}$ and a canonization algorithm ${\cal B}$ for $\kappa_{\cal B}$-improved-clique-unbreakable $\mathbb{F}$-free classes.
The running time of ${\cal B}$ on an instance $(G,{\cal F},k,s)$ is upper bounded by $g({\cal F},k,s)n^{\cO(1)}$ and the total time taken by at most $h({\cal F},k)n^{\cO(1)}$ invocations of ${\cal A}$ on $({\cal F}, G, k/2, q)$, where ${\cal F} \in \mathbb{F}$, $q : [k/2] \rightarrow \mathbb{N}$ is a function such that $q(i) \leq q^\star(i)$ for $i \leq k/2$, and $G$ is an ${\cal F}$-free, $q$-unbreakable graph on at most $n$ vertices.
}

\begin{proof}[Proof of Lemma~\ref{lem:unbreakToClique}]
Let $q^\star: \mathbb{F}^\star \times \mathbb{N} \rightarrow \mathbb{N}$ be the function defined as $q^\star({\cal F},i)= 2^{2^{2^{14i}}}\eta_w(2^{2^{13i}},h)$. We claim it satisfies the conclusion of the Lemma. 
Let  $\kappa_{\cal A} : \mathbb{F} \rightarrow \mathbb{N}$ be a function and let ${\cal A}$ be a canonization algorithm  for $(q^\star, \kappa_{\cal A})$-unbreakable $\mathbb{F}$-free classes. We define a function $\kappa_{\cal B} : \mathbb{F} \rightarrow \mathbb{N}$ as $\kappa_{\cal B}({\cal F})=2\kappa_{\cal A}({\cal F})$.  Next we design a canonization algorithm ${\cal B}$ for $\kappa_{\cal B}$-improved-clique-unbreakable $\mathbb{F}$-free classes using $\cal A$. 
${\cal B}$ takes as input a finite list ${\cal F} \in \mathbb{F}$, integers $k$ and $s$, and an ${\cal F}$-topological minor free, $(s, k, k)$-improved-clique-unbreakable colored graph $G$, and proceeds as follows.
${\cal B}$ obtains in polynomial time, by using Lemma~\ref{lem:unbreakbleAndstarDeco} on $\imp{G}{k}$ with parameters $k$ and $s$, an isomorphism invariant star decomposition $(T,\chi)$ of $\imp{G}{k}$, with the following properties. 
\begin{enumerate}
\setlength{\itemsep}{-2pt}
\item The central bag $\chi(b)$ is either a clique of size at most $k$ or a $k$-atom. 
\item For all $t\in V(T)$, $\sigma(t)\leq k$ (adhesion size is bounded by $k$). 
\item For all leaves $\ell \in V(T)$, $|\chi(\ell)|\leq 3s$.
\item $(T,\chi)$ is connectivity sensitive. 
\end{enumerate}

There are two cases, either $\chi(b)$ is a clique in $\imp{G}{k}$ or it is not. If $\chi(b)$ is a clique in $\imp{G}{k}$ then ${\cal B}$ calls the algorithm of Lemma~\ref{lem:cliqueUnbreakableCentralClique} on ${\cal F}$, $k$, $3s$, $G$, and $(T, \chi)$. If $k \geq \kappa_{\cal B} \geq 2\kappa_{\cal A}$ the algorithm of Lemma~\ref{lem:cliqueUnbreakableCentralClique} succeeds and returns a labeling $\lambda$ of $G$, which is weakly isomorphism invariant in~${\cal F}$, $k$, $3s$, $G$, and $(T, \chi)$ and therefore in ${\cal F}$, $k$, $s$ and $G$, using time at most $\cO((3s)!n^{\cO(1)})$ and at most one invocation of the algorithm ${\cal A}$ on an instance $({\cal F}, G, k/2, q_3)$, where ${\cal F} \in \mathbb{F}$, $q_3 : [k/2] \rightarrow \mathbb{N}$ is a function such that $q_3(i) \leq \hat{q}({\cal F}, i) \leq q^\star({\cal F}, i)$ for $i \leq k/2$, and $G$ is an ${\cal F}$-free, $q_3$-unbreakable graph on at most $n$ vertices. Here $\hat{q}$ is the function in the statement of Lemma~\ref{lem:cliqueUnbreakableCentralClique}, which satisfies $\hat{q}({\cal F}, i) \leq q^\star({\cal F}, i)$ for all ${\cal F}$ and $i \geq 1$.

Suppose now that $\chi(b)$ is not a clique in $\imp{G}{k}$. Then there exists a pair $u, v$ of distinct vertices in  $\chi(b)$, such that there exists a $u$-$v$ separation in $G$ of order at most $k$. This is just a non-edge in $\imp{G}{k}[\chi(b)]$. In this case, the algorithm proceeds as follows.
For every (ordered) pair $u$, $v$ of distinct vertices in  $\chi(b)$, such that there exists a $u$-$v$ separation in $G$ of order at most $k$ the algorithm computes the minimum order $u$-$v$ separation $(A, B)$ pushed towards $v$. Let $D = A \cap B$.
The algorithm ${\cal B}$ calls the algorithm of Lemma~\ref{lem:unbreakableToCliqueUnbreakable} on ${\cal F}$ and $G[A]$ with distinguished set $D$ and parameters $k$ and $3s$, and on ${\cal F}$ and $G[B]$ with distinguished set $D$ and parameters $k$ and $3s$.
Note that these are well-formed instances to the algorithm of Lemma~\ref{lem:unbreakableToCliqueUnbreakable} because, for an arbitrary $\iota : D \rightarrow [|D|]$ we have that $G = (G[A], \iota) \oplus (G[B], \iota)$ and the star decomposition $(T, \chi)$ of $G$ satisfies the premise of Lemma~\ref{lem:unbreakableToCliqueUnbreakable}.

If at least one of the calls to Lemma~\ref{lem:unbreakableToCliqueUnbreakable} outputs $\bot$ then ${\cal B}$ outputs $\bot$ as well (and this is weakly isomorphism invariant). However, if $k \geq \kappa_{\cal B}$ then $k/2 \geq \kappa_{\cal A}$ and therefore all calls to Lemma~\ref{lem:unbreakableToCliqueUnbreakable} succeed. We now assume that all calls to Lemma~\ref{lem:unbreakableToCliqueUnbreakable} succeed.
The invocations of Lemma~\ref{lem:unbreakableToCliqueUnbreakable} returned for every bijection $\pi : D \rightarrow [|D|]$ a weakly isomorphism invariant proper labeling $\lambda_\pi^A$ of $(G[A], \pi)$ and a weakly isomorphism invariant proper labeling $\lambda_\pi^B$ of $(G[B], \pi)$.

We now turn the separation $(A, B)$ into a star decomposition $(T, \chi)$ with central bag $b$ and two leaves $\ell_A$ and $\ell_B$.
We set $\chi(b) = S$, $\chi(\ell_A) = A$ and $\chi(\ell_B) = B$, and observe that the sets of labelings $\{\lambda_\pi^A ~:~ \pi : S \rightarrow [|S|]\}$ and  $\{\lambda_\pi^B ~:~ \pi : S \rightarrow [|S|]\}$ is a weakly isomorphism invariant leaf labeling of $(T, \chi)$. 

The algorithm invokes Lemma~\ref{lem:smallCentersAlgorithmWrapper} and obtains for every $\pi : S \rightarrow [|S|]$ a proper labeling $\lambda_\pi^{A,B}$ which is weakly isomorphism invariant in $G$, $A$, $B$ and $\pi$.
Finally, it chooses among all proper labelings $\lambda_\pi^{A,B}$ over all choices of $(A, B)$ and $\pi$ the proper labeling $\lambda$ that makes $G$ with labeling $\lambda$ lexicographically smallest. 

The set of ordered pairs $(u, v)$ such that $uv \notin E(\imp{G}{k})$ is isomorphism invariant in $G$ and $k$.
The separation $(A, B)$ pushed towards $v$ is isomorphism invariant in $G$, $u$ and $v$ by the uniqueness of pushed separators. 
Therefore the set of separations $(A, B)$ iterated over by the algorithm is isomorphism invariant in $G$ and $k$. We conclude that the set $\lambda_\pi^{A,B}$ of labelings is weakly isomorphism invariant in $G$, ${\cal F}$, $k$ and $s$. Since selection of a lexicographically smallest graph from a set of graphs is a weakly isomorphism invariant operation, $\lambda$ is weakly isomorphism invariant in $G$, ${\cal F}$, $k$ and $s$. 

The running time of the algorithm ${\cal B}$ is upper bounded by $n^{\cO(1)}k!$ plus $n^{2}$ invocations of the algorithm of Lemma~\ref{lem:unbreakableToCliqueUnbreakable}. Each of those invocations in turn take at most time 
$g({\cal F}, k, s)n^{\cO(1)}$ plus the time taken by $h({\cal F},k)n^{\cO(1)}$ invocations of ${\cal A}$ for some functions $g$, $h$. This concludes the proof.

\end{proof}

%% file: NE-interface.tex
In this section we state formally the main result of the subordinate paper~\cite{subordinate}
(which was informally announced as Theorem~\ref{thm:intro:rigid}) and then deduce from it
Theorem~\ref{thm:unbrk-minor}.
In a sense, the main result of~\cite{subordinate} strips Theorem~\ref{thm:unbrk-minor}
it of its algorithmic layer, in order to focus
on the underlying structural graph theory part. To this end, it is a bit more convenient
to phrase the unbreakability assumption as an \emph{unbreakability chain}. 

\begin{definition}[unbreakability chain]
Let $G$ be a graph and $f\colon \N\to \N$ be a nondecreasing function with $f(x) > x$ for every $x \in \N$.
An {\em{unbreakability chain}} of $G$ with \emph{length} $\zeta$ and \emph{step} $f$ is any sequence $((q_i,k_i))_{i=0}^\zeta$ of pairs of integers satisfying the following:
\begin{itemize}[nosep]
 \item $G$ is $(q_i,k_i)$-unbreakable, for each $0\leq i\leq \zeta$; and
 \item $k_i=f(q_{i-1} + k_{i-1})$, for each $1\leq i\leq \zeta$.
\end{itemize}
We shall often say that such a chain {\em{starts}} at $k_0$.
\end{definition}

We are now ready to state the main result of~\cite{subordinate}.

\begin{theorem}[\cite{subordinate}]\label{thm:rigid}
There exist constants $\ctime$, $\cgenus$, computable functions $\funtw$, $\funstep$, $\funchain$, $\funtime$, $\funcut$, $\fungenus$ (with $\funstep 
\colon \N \to \N$ being nondecreasing
    and satisfying $\funstep(x) > x$ for every $x \in \N$), and an algorithm that, given
a graph $H$ and an $H$-minor-free uncolored graph $G$
together with $((q_i,k_i))_{i=0}^\zeta$ being an unbreakability chain of $G$ of length $\zeta \coloneqq \funchain(H)$ with step $\funstep$ starting at $k_0 \coloneqq \funcut(H)$,
works in time bounded by $\funtime(H,\sum_{i=0}^\zeta q_i) \cdot |V(G)|^{\ctime}$ and computes a partition
$V(G) = \Vtw \uplus \Vgenus$ and a family $\famgenus$ of bijections $\Vgenus \mapsto [|\Vgenus|]$ 
such that 
\begin{enumerate}
\item the partition $V(G) = \Vtw \uplus \Vgenus$ is isomorphism-invariant;
\item $G[\Vtw]$ has treewidth bounded by $\funtw(H,\sum_{i=0}^\zeta q_i)$;
\item $|\famgenus| \leq \fungenus(H,\sum_{i=0}^\zeta q_i) \cdot |V(G)|^{\cgenus}$; and
\item $\famgenus$ is isomorphism-invariant.
\end{enumerate}
\end{theorem}

It remains to show how Theorem~\ref{thm:rigid} implies Theorem~\ref{thm:unbrk-minor}

\begin{proof}[Proof of Theorem~\ref{thm:unbrk-minor} assuming Theorem~\ref{thm:rigid}.]
To describe the algorithm $\mathcal{A}$ of Theorem~\ref{thm:unbrk-minor},
assume we are given $H$, an $H$-minor-free graph $G$,
an integer $k$, and a function $q \colon [k] \to \N$ such that $G$ is $(q(i),i)$-unbreakable
for every $1 \leq i \leq k$. 

Compute $\zeta = \funchain(H)$ and $k_0 = \funcut(H)$. 
Iteratively, for every $i=0,1,\ldots,\zeta$, proceed as follows. If $k_i > k$, stop and return $\bot$. 
Otherwise, look up $q_i \coloneqq q(k_i)$ and compute $k_{i+1} \coloneqq \funstep(q_i+k_i)$. 
If the iteration finished the step $i = \zeta$ without returning $\bot$, we proceed to the next steps of the algorithm.

In the subsequent steps, the algorithm will always return a canonical labeling of $G$ (i.e., it will no longer have an option of returning $\bot$). Before we proceed to the description of these steps,
let us argue about the existence of the bound $k_{H,\funq}$ from the second point in the statement of Theorem~\ref{thm:unbrk-minor}. To this end, let $\funq : \mathbb{N} \to \mathbb{N}$ be such
that $\funq(i) \geq q(i)$ for every $i \in [k]$. 

Consider a sequence $k_i'$ defined as $k_0' \coloneqq k_0 = \funcut(H)$ and $k_{i+1}' = \funstep(\funq(k_i') + k_i')$. 
Note that the assumptions on $\funstep$ imply that $k_{i+1}'$ is strictly increasing
and $\funq(i) \geq q(i)$ implies that $k_{i}' \geq k_i$ for all $i$ for which $k_i$ is defined. 
Recall that $\zeta = \funchain(H)$ and 
observe that if $k \geq k_\zeta'$, then the above iteration runs for all steps $i=0,\ldots,\zeta$ and does not return $\bot$. Thus, we can set $\kappa_{H,\funq} \coloneqq k_\zeta'$. 

We return to the description of the algorithm.
Note that now $((q_i,k_i))_{i=0}^\zeta$ is an unbreakability chain of $G$ of length $\zeta$.
We pass $H$, $G$ (without colors),
and the unbreakability chain $((q_i,k_i))_{i=0}^\zeta$ to the algorithm
of Theorem~\ref{thm:rigid}.
The algorithm returns a partition $V(G) = \Vtw \uplus \Vgenus$ and a family $\famgenus$.

For every $\lambda \in \famgenus$, we additionally color every vertex $v \in \Vtw$
with color $\{\lambda(w)~|~w \in \Vgenus \cap N_G(v)\}$. 
Then, with this enhanced coloring, we run the canonization algorithm for graphs of bounded
treewidth of~\cite{LokshtanovPPS17}\footnote{Formally, in~\cite{LokshtanovPPS17} only canonization for uncolored graphs is discussed. However, it is straightforward to derive from it also a canonization algorithm for colored graphs of bounded treewidth, for instance by attaching to every vertex a suitable number of degree-$1$ vertices to signify its color.},
obtaining a labeling $\Lambda_\lambda'$
of $G[\Vtw]$. Then, we consider a labeling $\Lambda_\lambda$ of $V(G)$
defined as $\Lambda_\lambda(v) = \Lambda_\lambda'(v)$ for $v \in \Vtw$ and $\Lambda_\lambda(v) = |\Vtw| + \lambda(v)$ for $v \in \Vgenus$. 

The above operation results in a family of labelings $\{\Lambda_\lambda~|~\lambda \in \famgenus\}$
of size at most $|\famgenus|\leq \fungenus(H,\sum_{i=0}^\zeta q_i) \cdot |V(G)|^{\cgenus}$. 
We choose and return
the labeling that yields the lexicographically-minimum labeled graph $G$.

The fact that the partition $V(G) = \Vtw \uplus \Vgenus$ and the family $\famgenus$ 
is isomorphism-invariant and the properties of the canonization algorithm
of~\cite{LokshtanovPPS17} ensures that the returned labeling is canonical. 

For the running time bound, observe that the computability of $\funchain$, $\funstep$,
and $\funcut$ implies that the chain $((q_i,k_i))_{i=0}^\zeta$ is 
computable in time depending only on the values of $q$ and the graph $H$. 
Also, the exact values of $\zeta$ and $((q_i,k_i))_{i=0}^\zeta$ depend on $q$ and the graph $H$ only.
Hence, the running time of the call to the algorithm of Theorem~\ref{thm:rigid}
is within the promised running time bound.
Furthermore, the bound on the treewidth of $G[\Vtw]$ depends only on the values of $q$ and the graph
$H$, hence the call to the algorithm of~\cite{LokshtanovPPS17} is also within the promised
running time bound. 
Finally, the bound on the size of $\famgenus$ implies that the final computation runs also 
within the promised time complexity bound.
This finishes the proof of Theorem~\ref{thm:unbrk-minor}, assuming Theorem~\ref{thm:rigid}.
\end{proof}